\documentclass[10pt,journal,onecolumn,draftclsnofoot]{IEEEtran}
\usepackage{epsfig}
\usepackage{amsmath}
\usepackage{amstext}
\usepackage{amsfonts}
\usepackage{amssymb}
\usepackage{eucal}
\usepackage{graphicx}
\usepackage{verbatim}
\usepackage{stmaryrd}
\usepackage{algorithm,algorithmic}
\usepackage{bm}

\usepackage{tikz}
\usepackage{pgfplots}
\usepackage{booktabs}

\newfont{\bbb}{msbm10 scaled 500}

\newfont{\bb}{msbm10 scaled 1100}
\newcommand{\CC}{\mbox{\bb C}}
\newcommand{\RR}{\mbox{\bb R}}

\newcommand{\EE}{\mbox{\bb E}}
\newcommand{\NN}{\mbox{\bb N}}

\newcommand{\hv}{{\bf h}}

\newcommand{\mv}{{\bf m}}
\newcommand{\nv}{{\bf n}}

\newcommand{\uv}{{\bf u}}
\newcommand{\wv}{{\bf w}}

\newcommand{\xv}{{\bf x}}
\newcommand{\yv}{{\bf y}}
\newcommand{\zv}{{\bf z}}
\newcommand{\zerov}{{\bf 0}}

\newcommand{\Am}{{\bf A}}
\newcommand{\Bm}{{\bf B}}

\newcommand{\Dm}{{\bf D}}

\newcommand{\Fm}{{\bf F}}

\newcommand{\Hm}{{\bf H}}
\newcommand{\Id}{{\bf I}}

\newcommand{\Pm}{{\bf P}}
\newcommand{\Qm}{{\bf Q}}
\newcommand{\Rm}{{\bf R}}

\newcommand{\Tm}{{\bf T}}
\newcommand{\Um}{{\bf U}}
\newcommand{\Wm}{{\bf W}}

\newcommand{\Xm}{{\bf X}}

\newcommand{\Zm}{{\bf Z}}

\newcommand{\Cc}{{\cal C}}

\newcommand{\Nc}{{\cal N}}
\newcommand{\Oc}{{\cal O}}

\newcommand{\Rc}{{\cal R}}
\newcommand{\Sc}{{\cal S}}

\newcommand{\Lambdam}{\hbox{\boldmath$\Lambda$}}

\newcommand{\diag}{{\hbox{diag}}}
\newcommand{\trace}{{\hbox{tr}}}

\renewcommand{\Re}{{\rm Re}}
\renewcommand{\Im}{{\rm Im}}

\newcommand{\htp}{^{\sf H}} 
\newcommand{\tp}{^{\sf T}}  
\def\LSB{\left[}        
\def\RSB{\right]}       
\def\LB{\left(}         
\def\RB{\right)}        

\newcommand{\asto}[1]{\xrightarrow[#1]{\text{a.s.}}}

\newcommand{\Exp}{{\mathbb{E}}}
\newcommand{\B}{{\bf B}}

\newcommand{\A}{{\bf A}}
\newcommand{\F}{{\bf F}}
\newcommand{\G}{{\bf G}}
\newcommand{\X}{{\bf X}}
\newcommand{\W}{{\bf W}}

\renewcommand{\P}{{\bf P}}
\newcommand{\T}{{\bf T}}
\newcommand{\V}{{\bf V}}

\newcommand{\R}{{\bf R}}

\renewcommand{\H}{{\bf H}}
\newcommand{\I}{{\bf I}}
\newcommand{\x}{{\bf x}}

\renewcommand{\v}{{\bf v}}

\newcommand{\y}{{\bf y}}
\newcommand{\z}{{\bf z}}
\newcommand{\Z}{{\bf Z}}
\newcommand{\h}{{\bf h}}

\newcommand{\n}{{\bf n}}
\newcommand{\w}{{\bf w}}

\newcommand{\oh}{{\frac{1}{2}}}

\renewcommand{\asto}{\overset{\rm a.s.}{\longrightarrow}}

\newcommand{\herm}{{\sf H}}

\DeclareMathOperator{\tr}{tr}

\newtheorem{theorem}{Theorem}

\newcounter{ccorollary}
\newtheorem{corollary}[ccorollary]{Corollary}
\newcounter{clemma}
\newtheorem{lemma}[clemma]{Lemma}
\newcounter{cdefinition}
\newtheorem{definition}[cdefinition]{Definition}
\newcounter{cremark}
\newtheorem{remark}[cremark]{Remark}

\begin{document}
\bibliographystyle{IEEEtran}

\title{Random Beamforming over \\ Quasi-Static and Fading Channels:\\ A Deterministic Equivalent Approach\thanks{This paper was presented in part at the IEEE International Conference on Communications, Kyoto, Japan, 2011, under the title ``Deterministic Equivalents for the Performance Analysis of Isometric Random Precoded Systems''.}}

\author{Romain~Couillet$^{\dag}$\thanks{$^\dag$Department of Telecommunications, Sup\'elec, 3 rue Joliot Curie, 91192 Gif sur Yvette, France.}, Jakob~Hoydis$^{\ddag}$\thanks{$^\ddag$Bell Labs, Alcatel-Lucent, Lorenzstr. 10, 70435 Stuttgart, Germany.}, and M\'erouane~Debbah$^\star$\thanks{$^\star$Alcatel Lucent Chair on Flexible Radio, Sup\'elec, 3 rue Joliot Curie, 91192 Gif sur Yvette, France.}\thanks{\{\scriptsize\texttt{romain.couillet,merouane.debbah\}@supelec.fr, jakob.hoydis@alcatel-lucent.com}}}

\maketitle

\begin{abstract}
In this work, we study the performance of random isometric precoders over quasi-static and correlated fading channels. We derive deterministic approximations of the mutual information and the signal-to-interference-plus-noise ratio (SINR) at the output of the minimum-mean-square-error (MMSE) receiver and provide simple provably converging fixed-point algorithms for their computation. Although these approximations are only proven exact in the asymptotic regime with infinitely many antennas at the transmitters and receivers, simulations suggest that they closely match the performance of small-dimensional systems. We exemplarily apply our results to the performance analysis of multi-cellular communication systems, multiple-input multiple-output multiple-access channels (MIMO-MAC), and MIMO interference channels. The mathematical analysis is based on the Stieltjes transform method. This enables the derivation of deterministic equivalents of functionals of large-dimensional random matrices. In contrast to previous works, our analysis does not rely on arguments from free probability theory which enables the consideration of random matrix models for which asymptotic freeness does not hold. Thus, the results of this work are also a novel contribution to the field of random matrix theory and applicable to a wide spectrum of practical systems. 
\end{abstract}

\clearpage
\section{Introduction}
\label{sec:intro}
Consider the following discrete time wireless channel model
\begin{align}
	\label{eq:channel}
 \y = \sum_{k=1}^K \H_k\W_k\P_k^\frac12\x_k + \n
\end{align}
where 
\begin{itemize}
\item[(i)] $\y\in\CC^N$ is the channel output vector,
\item[(ii)] $\H_k\in \CC^{N\times N_k},\, k\in\{1,\dots,K\}$, are complex channel matrices, satisfying either of the following properties:
	\begin{itemize}
		\item[(ii-a)] The matrix $\H_k$ is deterministic. In this case, we will denote $\R_k=\H_k\H_k^\herm$.
		\item[(ii-b)] The matrix $\H_k$ is a random channel matrix whose $j$th column vector $\h_{kj}\in\CC^{N}$ is modeled as
\begin{align}
	\label{eq:channelmodel}
	\h_{kj} = \Rm_{kj}^{\frac12}\z_{kj},\qquad j\in\{1,\dots,N_k\} 
\end{align}
where $\Rm_{kj}\in\CC^{N\times N}$ are Hermitian nonnegative definite matrices and the vectors $\z_{kj}\in\CC^{N}$ have independent and identically distributed (i.i.d.) elements with zero mean, variance $1/N$ and $4+\epsilon$ moment of order $\Oc(1/N^{2+\varepsilon/2})$, for some common $\epsilon>0$.
	\end{itemize}
\item[(iii)] $\W_k\in \CC^{N_k\times n_k},\, k\in\{1,\dots,K\}$, are complex (signature or precoding) matrices which contain each $n_k< N_k$ orthonormal columns of independent $N_k\times N_k$ Haar-distributed random unitary matrices,\footnote{We recall that a Haar random matrix $\W_k\in\CC^{N_k\times N_k}$ is defined by $\W_k=\X_k(\X_k\htp\X_k)^{-\frac12}$ for $\X_k$ a random matrix with independent entries $\mathcal{CN}(0,1)$ entries.}
\item[(iv)] $\P_k\in\RR^{n_k\times n_k},\,k\in\{1,\dots,K\}$, are diagonal (power loading) matrices with nonnegative entries,
\item[(v)] $\x_k\sim \mathcal{CN}(0,\I_{n_k})$, $k\in\{1,\dots,K\}$, are random independent transmit vectors,
\item[(vi)] $\n\sim \mathcal{CN}(0,{\sigma^2}\I_{N})$ is a noise vector.
\end{itemize}
In addition, we define the ratios of the matrix dimensions $c_i\triangleq \frac{n_i}{N_i}$ and $\bar{c}_i\triangleq \frac{N_i}{N}$ for $i\in\{1,\dots,K\}$.

\bigskip
\begin{remark}\label{rem:kronecker}
The statistical model \eqref{eq:channelmodel} of the channel $\H_k$ under assumption (ii-b) generalizes several well-known fading channel models of interest (see \cite{WAG10,hoydis2011} for examples). These models comprise in particular the {\it Kronecker} channel model with transmit and receive correlation matrices \cite{CHU02,COU09}, where the matrices $\H_k$ are given by
\begin{align}\label{eq:kronecker}
\H_k = \Rm_k^\frac12\Z_k\T_k^\frac12  
\end{align}
with $\Z_k\in\CC^{N\times N_k}$ a random matrix whose elements are independent $\mathcal{CN}(0,1/N)$ and $\Rm_k\in\CC^{N\times N}$, $\T_k\in\CC^{N_k\times N_k}$ antenna correlation matrices. Since both $\Z_k$ and $\W_k$ are unitarily invariant, we can assume without loss of generality for the statistical properties of $\y$ that $\T_k=\diag(t_{k1},\ldots,t_{kN_k})$. Defining the matrices $\Rm_{kj}=t_{kj}\R_k$ for $j\in\{1,\dots,N_k\}$, we fall back to the channel model in \eqref{eq:channelmodel}. 
Taking instead all $\Rm_{kj}$ to be diagonal matrices makes the entries of $\H_k$ independent with $[\H_k]_{ij}$ of zero mean and variance $[\Rm_{kj}]_{ii}/N$. This corresponds to a centered variance profile model, studied extensively in \cite{HAC07,HAC08,DUM10}.
\end{remark}
\bigskip

The objective of this work is to study the performance of the communication channel \eqref{eq:channel} in the large dimensional regime where $N,N_1,\ldots,N_K,n_1,\ldots,n_K$ are simultaneously large. In the following, we will consider both the {\it quasi-static channel scenario} which assumes hypotheses (i), (ii-a), (iii)-(vi), and the {\it fading channel scenario} which assumes (i), (ii-b), (iii)-(vi). The study of the latter naturally arises as an extension of the study of the quasi-static channel scenario. The respective application contexts and an overview of related works for both scenarios are summarized below.

\subsection{Quasi-static channel scenario (hypothesis (ii-a))}
Possible applications of the channel model \eqref{eq:channel} under assumptions (i), (ii-a), (iii)-(vi) arise in the study of direct-sequence (DS) or multi-carrier (MC) code-division multiple-access (CDMA) systems with isometric signatures over frequency-selective fading channels or space-division multiple-access (SDMA) systems with isometric precoding matrices over flat-fading channels. More precisely, for DS-CDMA systems, the matrices $\H_k$ are either Toeplitz or circulant matrices (if a cyclic prefix is used) constructed from the channel impulse response; for MC-CDMA, the matrices $\H_k$ are diagonal and represent the channel frequency response on each sub-carrier; for flat fading SDMA systems, the matrices $\H_k$ can be of arbitrary form and their elements represent the complex channel gains between the transmit and receive antennas. In all cases, the diagonal entries of the matrices $\P_k$ determine the transmit power of each signature (CDMA) or transmit stream (SDMA).

The large system analysis of random i.i.d. and random orthogonal precoded systems with optimal and sub-optimal linear receivers has been the subject of numerous publications. The asymptotic performance of minimum-mean-square-error (MMSE) receivers for the channel model \eqref{eq:channel} for the case $K=1, \P_1=\I_{n_1}$, and $\H_1$ diagonal with i.i.d. elements has been studied in \cite{DEB03} relying on results from free probability theory. This result was extended to frequency-selective fading channels and sub-optimal receivers in \cite{HAC02}. Although not published, the associated mutual information was evaluated in \cite{HAC00} (this result is recalled in \cite[Theorem 4.11]{COUbook}). The case of i.i.d. and isometric MC-CDMA over Rayleigh fading channels with multiple signatures per user terminal, i.e., $K\ge 1$ and $\H_k$ diagonal with i.i.d. complex Gaussian entries, was considered in \cite{PEA04}, where approximate solutions of the signal-to-noise-plus-interference-ratio (SINR) at the output of the MMSE receiver were provided. Asymptotic expressions for the spectral efficiency of the same model were then derived in \cite{PEA06}. 
DS-CDMA over flat-fading channels, i.e., $K\ge 1$, $n_k=N$, and $\H_k=\I_N$ for all $k$, was studied in \cite{COU09d}, where the authors derived deterministic equivalents of the Shannon- and $\eta$-transform based on the asymptotic freeness \cite[Section 3.5]{COUbook} of the matrices $\W_k\P_k\W_k^\herm$. Besides, a sum-rate maximizing power-allocation algorithm was proposed.
Finally, a different approach via incremental matrix expansion \cite{PEA08} led to the exact characterization of the asymptotic SINR of the MMSE receiver for the general channel model \eqref{eq:channel}. However, the previously mentioned works share the underlying assumption that the spectral distributions of the matrices $\H_k$ and $\P_k$ converge to some limiting distributions or that the matrices $\H_k\H_k^\herm$ are jointly diagonalizable.\footnote{That is, there exists a unitary matrix $\V$ such that $\V\H_k\H_k^\herm\V^\herm$ is diagonal for all $k$.} In addition, the computation of the asymptotic SINR requires the computation of rather complicated implicit equations. These can be solved in most cases by standard fixed-point algorithms but a proof of convergence to the correct solution was not provided. Finally, a closed-form expression for the asymptotic spectral efficiency is missing, although an approximate solution which requires numerical integration was presented in \cite{PEA06}. Alternative combinatoric methods also exist, such as the diagrammatic approach \cite{MOU08}, to evaluate the successive moments of the limiting eigenvalue distribution of such matrix models.

The above results assume non-random communication channels $\H_k$ and can only be applied to the performance analysis of static or slow fading channels. Turning the matrices $\H_k$ into random matrices instead allows for the study of the ergodic performance of fast fading channels with isometric precoders. The next section discusses the practical applications in this broader context.

\subsection{Fading channel scenario (hypothesis (ii-b))}
The second scenario considers the channel model \eqref{eq:channel} under assumptions (i), (ii-b), (iii)-(vi). In contrast to the first scenario, the $\H_k$ matrices are now assumed to be random. Thus, we aim at evaluating both the instantaneous performance for a random channel realization and the ergodic performance. These are appropriate performance measures in fast fading environments. 

Of particular interest in this setting is the evaluation of the multiple-input multiple-output (MIMO) channel capacity under random beamforming. In point-to-point MIMO channels, the ergodic channel capacity has been the object of numerous works and is by now well understood \cite{TEL99,DUM06}. However, the ergodic sum-rate of more involved models, such as the MIMO multiple access channel (MIMO-MAC) \cite{COU09} under individual or sum power constraints, has been studied only recently within the scope of random matrix theory.  Another important aspect is the capacity of MIMO channels with co-channel interference, for which much less is known about the optimal transmission strategies \cite{blum2002,blum2003}.
The first interesting question relates to the problem of how many antennas should be used for transmission and how many independent data streams should be sent, which are the same problem when the channels have i.i.d.\@ entries. With transmit antenna correlation, however, it makes a difference which antennas are selected for transmission and the question of the optimal number of antennas to be used becomes a combinatorial problem. To circumvent this issue, random beamforming can be used. The remaining question is then how many orthogonal streams should be sent, using {\it all} available antennas. We will address this problem later in this article, as our results enable the evaluation of the sum-rate of systems composed of multiple transmitter-receiver pairs, each applying random isotropic beamforming.\bigskip

In summary, regardless of the specific application scenario of the model \eqref{eq:channel}, unitary precoders have gained significant interest in wireless communications \cite{tse2002} (see also the recent work on spatial multiplexing systems \cite{LOV05} and limited feedback beamforming solutions in future wireless standards \cite{LEE09}). Thus, the performance evaluation of isometric precoded systems is compulsory and a field of active research \cite{HUA09}.

\subsection{Contributions}
The object of this article is to propose a new framework for the analysis of large random matrix models involving Haar matrices using the {\it Stieltjes-transform method} initiated by Pastur and fully exploited by Bai and Silverstein \cite{MAR67,SIL06}. This method is considered today as one of the most practical and powerful tools for handling large random matrices in wireless communications research. Our analysis is fundamentally based on a trace lemma for Haar matrices first provided in \cite{DEB03} and recalled in Lemma \ref{le:trace_Haar} (Appendix~\ref{sec:fundlemmas}). Unlike previous contributions, we dismiss most of the practical constraints of free probability theory, combinatorial and incremental matrix expansion methods, such as the need for spectral limits of the deterministic matrices in the model to exist, or the need for the matrices $\H_k\H_k^\herm$ to be diagonalizable in a common eigenvector basis. The expressions we derive appear to be very similar to previously derived expressions when the precoding matrices $\W_k$ have i.i.d.\@ entries instead of being Haar distributed (see in particular Remark~\ref{rem:iid}). This allows for a unified understanding of both models with i.i.d.\@ or Haar matrices. As a consequence, we believe that the generality of the theoretical results presented in this article, supported by a large scope of application contexts, might stimulate further related research. We also mention that an alternative method to prove the results of this paper could be based on the integration by parts formula for Gaussian random matrices developed by Pastur \cite{PAS11}.\bigskip

Before summarizing our main contributions, we introduce some definitions which will be of repeated use. The central object of interest is the matrix $\B_N\in\CC^{N\times N}$, defined as 
\begin{align*}
\Bm_N = \sum_{k=1}^K \Hm_k\Wm_k\Pm_k\Wm_k\htp\Hm_k\htp.
\end{align*}
We denote by $I_N({\sigma^2})$ the normalized mutual information of the channel \eqref{eq:channel}, given by \cite{COV06}
\begin{align*}
 I_N({\sigma^2}) = \frac{1}{N}\log\det\LB \Id_N +\frac{1}{{\sigma^2}}\Bm_N\RB \quad (\text{nats/s/Hz}).
\end{align*}
 We further denote by $\gamma^N_{kj}(\sigma^2)$ the SINR at the output of the linear MMSE detector for the $j$th component of the transmit vector $\xv_k$, which reads \cite{verdubook}
\begin{align*}
 \gamma^N_{kj}({\sigma^2}) = p_{kj}\wv_{kj}\htp\Hm_k\htp\LB{\Bm_N}_{(k,j)} + {\sigma^2}\Id_N\RB^{-1}\Hm_k\wv_{kj}
\end{align*} 
where ${\Bm_N}_{(k,j)} = \Bm_N - p_{kj}\Hm_k\wv_{kj}\wv_{kj}\htp\Hm_k\htp$ and $\wv_{kj}$ is the $j$th column of $\Wm_k$.
We then define the normalized sum-rate with MMSE detection as
\begin{align*}
 R_N({\sigma^2}) = \frac1N\sum_{k=1}^K\sum_{j=1}^{n_k}\log\LB1+\gamma^N_{kj}({\sigma^2})\RB.
\end{align*} 
Depending on whether we consider the quasi-static channel scenario (ii-a) or the fading channel scenario (ii-b), we rename $I_N({\sigma^2})$ by $I^{(a)}_N({\sigma^2})$ and $I^{(b)}_N({\sigma^2})$, the mutual information under hypothesis (ii-a) and (ii-b), respectively. The same holds for $\gamma^N_{kj}({\sigma^2})$ and $ R_N({\sigma^2})$.

The technical contributions of this paper are as follows: we derive deterministic approximations $\bar{I}_N({\sigma^2})$, $\bar{\gamma}^N_{kj}(\sigma^2)$, and $\bar{R}_N(\sigma^2)$ of $I_N({\sigma^2})$, $\gamma^N_{jk}(\sigma^2)$, and $R_N(\sigma^2)$, respectively, which are (almost surely) asymptotically tight as the system dimensions $N,N_i,n_i$ grow large at the same rate (denoted simply $N\to\infty$). These approximations, often referred to as \emph{deterministic equivalents}, are easy to compute as they are shown to be the limits of simple (provably converging) fixed-point algorithms, they are given in closed form and do not require any numerical integration, and they require only very general conditions on the matrices $\H_k$ and $\P_k$.

We then present several applications of our results to wireless communications. First, we consider a cellular uplink orthogonal SDMA communication model with inter-cell interference, assuming independent codes in adjacent cells and quasi-static channels at all communication pairs.
We then study a MIMO multiple access channel (MAC) from several multi-antenna transmitters to a multi-antenna receiver under the fading channel scenario (hypothesis (ii-b)). The transmitters are unaware of the channel realizations and send an arbitrary number of independent data streams using isometric random beamforming vectors. The receiver is assumed to be aware of all instantaneous channel realizations and beamforming vectors. Under this setting, we derive an approximation for the achievable sum-rate and mutual information. 
Finally, we address the problem of finding the optimal number of independent streams to be transmitted in a two-by-two interference channel. Although the use of deterministic approximations in this context requires an exhaustive search over all possible stream-configurations, it is computationally much less expensive than Monte Carlo simulations. Extensions to more than two transmit-receive pairs and possible different objective functions, e.g., weighted sum-rate or sum-rate with MMSE decoding, are straightforward and not presented. 

For all these applications, numerical simulations show that the deterministic approximations are very tight even for small system dimensions. In the interference channel model, these simulations suggest in particular that, at low SNR, it is optimal to use all streams while, at high SNR, stream-control, i.e., transmitting less than the maximal number of streams, is beneficial. 

Our work also constitutes a novel contribution to the field of random matrix theory as we introduce new proof techniques based on the Stieltjes transform method for random isometric matrices. Namely, we provide in Theorem~\ref{th:sumRWT} (Appendix~\ref{app:sumRXT}) a deterministic equivalent $\bar{F}_N$ of the eigenvalue distribution $F_N$ of $\B_N$, referred to as the {\it empirical spectrum distribution} (e.s.d.). That is, $\bar{F}_N$ is such that, as $N\to\infty$, $F_N - \bar{F}_N \Rightarrow 0$, this convergence being valid almost surely. Although deterministic equivalents of e.s.d. are by now more or less standard and have been developed for rather involved random matrix models \cite{HAC07,COU09,WAG10}, results for the case of isometric (Haar) matrices are still an exception. In particular, most results on Haar matrices are based on the assumption of {\it asymptotic freeness} of the underlying matrices, a requirement which is rarely met for the matrices in the channel model \eqref{eq:channel} of interest here. The approach taken in this work is therefore novel as it does not rely on free probability theory \cite{VOI92,PET06} and we do not require any of the matrices in \eqref{eq:channel} to be asymptotically free. Interestingly, a very recent extension of free probability theory, coined free deterministic equivalents \cite{SPE11}, has come as a response to the present article in which free probability tools are developed to tackle the aforementioned limitations.

The remainder of this article is structured as follows: in Section \ref{sec:results}, we introduce the main results of this work, the proofs of which are postponed to the appendices. In Section \ref{sec:simu}, the results are applied to the practical wireless communication models discussed above. Section \ref{sec:conclusion} concludes the article.

\vspace{10pt}\textbf{Notations}: Boldface lower and upper case symbols represent vectors and matrices, respectively. $\Id_N$ is the size-$N$ identity matrix and $\diag(x_1,\dots,x_N)$ is a diagonal matrix with elements $x_i$. The trace, transpose and Hermitian transpose operators are denoted by $\trace(\cdot)$, $(\cdot)\tp$ and $(\cdot)\htp$, respectively. 
The spectral norm of a matrix $\Am$ is denoted by $\lVert\Am\rVert$, and, for two matrices $\Am$ and $\Bm$, the notation $\Am\succ\Bm$ means that $\Am-\Bm$ is positive-definite. 
The notations $\Rightarrow$ and $\asto$ denote weak and almost sure convergence, respectively. We use $\Cc\Nc\left(\mv,\Rm\right)$ to denote the circular symmetric complex Gaussian distribution with mean $\mv$ and covariance matrix $\Rm$. We denote by $\RR_+$ the set $[0,\infty)$ and by $\CC_+$ the set $\{z\in\CC,\Im[z]>0\}$. 
Denote by $\mathcal{C}(X,Y)$ the set of continuous functions from $X\subset\CC$ to $Y\subset\CC$, by $\mathcal{H}(X,Y)$ the set of holomorphic functions from $X\subset\CC$ to $Y\subset\CC$, and by $\mathcal{S}(X)$ the class {\it Stieltjes transforms} of finite measures supported by $X\subset \RR$ (see Definition \ref{def:ST} in Appendix \ref{app:sumRXT}).

\section{Main results}
\label{sec:results}
In this section, we present the main results of the article. All proofs are deferred to the appendices. We will distinguish the results for the quasi-static and the fading channel scenarios. Since we will make limiting considerations as the system dimensions grow large, some technical assumptions will be necessary:

\begin{description}\label{as:Ninfty}
\item[{\bf A1}]
The notation $N\to\infty$ denotes the simultaneous growth of $N,N_i,n_i$ for all $i$, in such a way that the ratios $c_i=\frac{n_i}{N_i}$ and $\bar{c}_i=\frac{N_i}{N}$ satisfy $0\leq \liminf_N c_i\leq \limsup_N c_i< 1$ and $0<\lim\inf_N \bar{c}_i \leq \lim\sup_N \bar{c}_i<\infty$.
\end{description}
For all convergence results in this paper (as $N\to\infty$), the matrices $\Pm_k=\Pm_k(N)\in\RR_+^{n_k\times n_k}$, $\Hm_k=\Hm_k(N)\in\CC^{N\times N_k}$ (as well as the $\Rm_{kj}=\Rm_{kj}(N)\in\CC^{N\times N}$ under assumption (ii-b)), and $\Wm_k=\Wm_k(N)\in\CC^{N_K\times n_k}$ should be understood as sequences of (random) matrices with growing dimensions. Wherever this is clear from the context, we drop the dependence on $N$ to simplify the notations.

In order to control the power loading matrices as the system grows large, we need the following assumption:
\begin{description}\label{as:limpow}
\item[{\bf A2}]
 There exists $P>0$ such that, for all $k$, $\lim\sup_N\lVert\Pm_k\rVert\leq P$. 
 \end{description}

Under (ii-a), the channel gains will need to remain bounded for all large $N$:
\begin{description}\label{as:bounded_H}
\item[{\bf A3}-a]
 There exists $R>0$ such that $\max_k\lim\sup_N\lVert\R_k\rVert\leq R$, where we recall that $\R_k=\H_k\H_k^\herm$. 
 \end{description}

The equivalent constraint under (ii-b) is that the channel correlations remain bounded for all large $N$:
\begin{description}\label{as:bounded_R}
\item[{\bf A3}-b]
	There exists $R>0$ such that $\lim\sup_N\lVert\R_{kj}\rVert\leq R$ for all $j,k$. 
 \end{description}

Due to some technical issues, it will be sometimes necessary to require the following condition:
\begin{description}\label{as:finset}
\item[{\bf A4}]
	For all random matrices $\H_k$ within a set of probability one, there exists $M>0$ such that $\max_k\|\H_k\H_k^\herm\|<M$ for all large $N$.
\end{description}
Assumption {\bf A4} is met in particular in the situation when there exists $m>0$, such that for all $k,j,N$, $\Rm_{kj}\in\Rc_N$ with $\Rc_N$ a discrete set of cardinality $|\Rc_N|<m$ for all $N$ (see the arguments in \cite{COU09}). For example, this holds true for the scenario of a common correlation matrix at each receiver, i.e.,  $\Rm_{kj}=\bar{\R}_k$ are equal for all $j$.

\subsection{Fundamental Equations}

We first introduce the \emph{fundamental equations} for model \eqref{eq:channel}. These equations provide the core deterministic quantities that will define the deterministic equivalents for $I_N(\sigma^2)$, $\gamma^N_{ij}(\sigma^2)$, and $R_N(\sigma^2)$.

\begin{theorem}[Fundamental equations under (ii-a)]
  \label{th:fundamental_eq}
  \label{th:fundequdet}
  Consider the system model \eqref{eq:channel} under assumptions (i), (ii-a), (iii)-(vi). Let ${\sigma^2}>0$. Then the following system of implicit equations
  \begin{align}
	  \label{eq:pi_a}
	  \bar{a}_k({\sigma^2}) &= \frac1N\tr \P_k\left(a_k({\sigma^2}) \P_k + [\bar{c}_k-a_k({\sigma^2}) \bar{a}_k({\sigma^2})] \I_{n_k} \right)^{-1} \nonumber \\
    	  a_k({\sigma^2}) &= \frac1N\tr \R_k\left(\sum_{j=1}^K \bar{a}_j({\sigma^2}) \R_j +{\sigma^2} \I_N \right)^{-1}
  \end{align}
  with $k\in\{1,\ldots,K\}$, admits a unique solution such that, for all $k$, $a_k(\sigma^2),\bar{a}_k(\sigma^2)\geq 0$, and $0\leq a_k({\sigma^2})\bar{a}_k({\sigma^2}) < c_k\bar{c}_k$. Moreover, this solution is obtained explicitly by the following fixed-point algorithm 
  \begin{equation*} 
	   \bar{a}_k({\sigma^2})=\lim_{t\to\infty}\bar{a}_k^{(t)}({\sigma^2}),\quad {a}_k({\sigma^2})=\lim_{t\to\infty}{a}_k^{(t)}({\sigma^2}), \quad \bar{a}_k^{(t)}({\sigma^2}) = \lim_{l\to \infty} \bar{a}_k^{(t,l)}({\sigma^2})
  \end{equation*}
where, for $k\in\{1,\dots,K\}$,
\begin{align*}
	\bar{a}_k^{(t,l)}({\sigma^2}) &= \frac1N\tr \P_k\left(a_k^{(t)}({\sigma^2}) \P_k + [\bar{c}_k-a_k^{(t)}({\sigma^2}) \bar{a}^{(t,l-1)}_k({\sigma^2})] \I_{n_k} \right)^{-1} \\
	a_k^{(t)}({\sigma^2}) & = \frac1N\tr \R_k\left(\sum_{j=1}^K \bar{a}^{(t-1)}_j({\sigma^2}) \R_j +{\sigma^2}\I_N \right)^{-1}
\end{align*} 
with initial values $\bar{a}_k^{(t,0)}({\sigma^2})=0$ and $a_k^{(0)}({\sigma^2})=0$.
\end{theorem}
\begin{proof}
 The proof is provided in Appendix \ref{app:sumRXT}.
\end{proof}

\bigskip
\begin{remark}
  \label{rem:iid}
  Assume $\bar{c}_k=1$ for every $k$ (e.g., when $\H_k$ is a Toeplitz matrix as in the CDMA case). Extending every $\P_k\in\CC^{n_k\times n_k}$ into $N\times N$ matrices filled with zeros, we may assume $c_k=1$ without affecting the final result. In this scenario, the fundamental equations \eqref{eq:channel} under (ii-a) become
  \begin{align}
	  \label{eq:NiN}
	  \bar{a}_k(\sigma^2) &= \frac1N\tr \P_k\left(a_k(\sigma^2) \P_k + [1-a_k(\sigma^2) \bar{a}_k(\sigma^2)] \I_N \right)^{-1} \\
    	  a_k(\sigma^2) &= \frac1N\tr \R_k\left(\sum_{j=1}^K \bar{a}_j(\sigma^2) \R_j +\sigma^2\I_N \right)^{-1}. \nonumber
  \end{align}
  This can be compared to the scenario where the matrices $\W_k$, instead of being Haar matrices, have i.i.d. entries of variance $1/N$. The fundamental equations of this model were derived in \cite[Corollary 1]{COU09} and are given as follows:
  \begin{align}
    \label{eq:NiNiid}   
    \underline{\bar{a}}_k({\sigma^2}) &= \frac1N\tr \P_k\left( \underline{a}_k({\sigma^2}) \P_k + \I_N \right)^{-1} \\
    \underline{a}_k({\sigma^2}) &= \frac1N\tr \R_k\left(\sum_{j=1}^K \underline{\bar{a}}_j({\sigma^2}) \R_j +{\sigma^2}\I_N \right)^{-1} \nonumber
  \end{align}
  such that $\underline{a}_k({\sigma^2})$ is positive for all $k$. The scalars $\underline{a}_k({\sigma^2})$ and $\bar{\underline{a}}_k({\sigma^2})$ are also defined as the limits of a classical fixed-point algorithm. The only difference between the two sets of equations lies in the additional term $-a_k(\sigma^2)\bar{a}_k(\sigma^2)\I_N$ in \eqref{eq:NiN}, not present in \eqref{eq:NiNiid}. 
\end{remark}
\bigskip

We now turn to the fundamental equations in the fading channel context.
\begin{theorem}[Fundamental equations under (ii-b)]
	\label{th:fundeq}
  Consider the system model \eqref{eq:channel} under assumptions (i), (ii-b), (iii)-(vi). Let ${\sigma^2}>0$. Then, the following system of implicit equations
 \begin{align*}
  \bar{b}_k({\sigma^2}) & = \frac1N\trace\Pm_k\Big( b_k({\sigma^2})\Pm_k + \LSB\bar{c}_k-b_k({\sigma^2})\bar{b}_k({\sigma^2})\RSB\Id_{n_k}\Big)^{-1}\\
b_k({\sigma^2}) & = \frac1N\sum_{j=1}^{N_k}\frac{\zeta_{kj}({\sigma^2})}{1+\bar{b}_k({\sigma^2})\zeta_{kj}({\sigma^2})}\\
\zeta_{kj}({\sigma^2}) &= \frac1N\trace\Rm_{kj}\LB\frac{1}{N}\sum_{k=1}^K\sum_{j=1}^{N_k}\frac{\bar{b}_k({\sigma^2})\Rm_{k,j}}{1+\bar{b}_k({\sigma^2})\zeta_{kj}({\sigma^2})}+{\sigma^2}\Id_N\RB^{-1},\quad j\in\{1,\dots,N_k\}
\end{align*}
with $k\in\{1,\ldots,K\}$, admits a unique solution satisfying $\zeta_{kj}({\sigma^2}),b_k({\sigma^2}),\bar{b}_k(\sigma^2)\ge 0$ and $0\leq b_k({\sigma^2}) \bar{b}_k({\sigma^2})< c_k\bar{c}_k$ for all $k,j$. 
Moreover, this solution is given explicitly by the following fixed-point algorithm
\begin{align*}
	\quad \bar{b}_k({\sigma^2}) &= \lim_{t\to\infty} \bar{b}^{(t)}_k({\sigma^2}),\quad {b}_k({\sigma^2}) = \lim_{t\to\infty} {b}^{(t)}_k({\sigma^2}),\quad \zeta_{kj}({\sigma^2}) = \lim_{t\to\infty} \zeta_{kj}^{(t)}({\sigma^2})
\end{align*} 
where 
\begin{align*}
\bar{b}^{(t)}_k({\sigma^2}) &= \lim_{l\to\infty}\bar{b}^{(t,l)}_k({\sigma^2}), \quad \zeta_{kj}^{(t)}({\sigma^2}) = \lim_{l\to\infty}\zeta_{kj}^{(t,l)}({\sigma^2})\\
b_k^{(t)}({\sigma^2})&=\frac1N\sum_{j=1}^{N_k}\frac{\zeta_{kj}^{(t)}({\sigma^2})}{1+\bar{b}_k^{(t-1)}({\sigma^2})\zeta_{kj}^{(t)}({\sigma^2})}\\
\bar{b}^{(t,l)}_k({\sigma^2})&= \frac1N\trace\Pm_k\LB b_k^{(t-1)}({\sigma^2})\Pm_k+\LSB\bar{c}_k-b_k^{(t-1)}({\sigma^2})\bar{b}_k^{(t,l-1)}({\sigma^2})\RSB\Id_{n_k}\RB^{-1}\\
\zeta_{kj}^{(t,l)}({\sigma^2}) &= \frac1N\trace\Rm_{kj}\LB\frac{1}{N}\sum_{k=1}^K\sum_{j=1}^{N_k}\frac{\bar{b}_k^{(t-1)}({\sigma^2})\Rm_{k,j}}{1+\bar{b}_k^{(t-1)}({\sigma^2})\zeta_{kj}^{(t,l-1)}({\sigma^2})}+{\sigma^2}\Id_N\RB^{-1}
\end{align*}
with the initial values $\zeta_{kj}^{(t,0)}({\sigma^2})=1/{\sigma^2}$, $\bar{b}_k^{(t,0)}=0$ and $b_k^{(0)}({\sigma^2})=0$ for all $k,j$.
\end{theorem}
\begin{proof}
 The proof is provided in Appendix \ref{app:proofs}.
\end{proof}

\subsection{System performance}
The following results are all based on the fundamental equations of Theorem \ref{th:fundamental_eq} and Theorem \ref{th:fundeq}.

\begin{theorem}[Mutual information under (ii-a)]
  \label{th:mutual_info}
  \label{th:logdetdet}
  Consider the system model \eqref{eq:channel} under assumptions (i), (ii-a), (iii)-(vi), and denote, for ${\sigma^2}>0$, 
  \begin{equation*}
	  I^{(a)}_N({\sigma^2})=\frac1N\log\det\left(\I_N + \frac1{{\sigma^2}}\B_N\right).
\end{equation*}
Assume {\bf A1}, {\bf A2}, and {\bf A3}-a. Then, as $N\to\infty$,
\begin{align*} 
	\mathbb{E} I^{(a)}_N({\sigma^2}) - \bar{I}^{(a)}_N({\sigma^2}) &\to 0 \\
	I^{(a)}_N({\sigma^2}) - \bar{I}^{(a)}_N({\sigma^2}) &\asto 0
\end{align*}
where 
\begin{align}
  \label{eq:VN}
  \bar{I}^{(a)}_N({\sigma^2}) &= \frac1N\log\det\left(\I_N + \frac1{{\sigma^2}}\sum_{k=1}^K \bar{a}_k\R_k\right) \nonumber \\ &+ \sum_{k=1}^K \left[\frac1N\log\det\left([\bar{c}_k-a_k\bar{a}_k]\I_{n_k} + a_k\P_k \right) + (1-c_k)\bar{c}_k \log(\bar{c}_k-a_k \bar{a}_k) - \bar{c}_k \log(\bar{c}_k) \right]
  \end{align}
  with $a_k=a_k({\sigma^2})$, $\bar{a}_k=\bar{a}_k({\sigma^2})$, $k\in\{1,\ldots,K\}$, given by Theorem~\ref{th:fundamental_eq}.
\end{theorem}
\begin{proof}
 The proof is provided in Appendix \ref{app:mutual_info}.
\end{proof}

\bigskip
\begin{theorem}[Mutual information under (ii-b)]\label{th:mutinf}
 Consider the system model \eqref{eq:channel} under assumptions (i), (ii-b), (iii)-(vi), and denote, for ${\sigma^2}>0$, 
  \begin{equation*}
	  I^{(b)}_N({\sigma^2})=\frac1N\log\det\left(\I_N + \frac1{{\sigma^2}}\B_N\right).
\end{equation*}
Assume {\bf A1}, {\bf A2}, {\bf A3}-b, and {\bf A4}. Let $\bar{b}_k=\bar{b}_k({\sigma^2})$, $b_k=b_k({\sigma^2})$ and $\zeta_{kj}=\zeta_{kj}({\sigma^2})$ for all $k,j$ be defined as in Theorem~\ref{th:fundeq}. 
Then, as $N\to\infty$,
\begin{align*} 
	\mathbb{E} I^{(b)}_N({\sigma^2}) - \bar{I}^{(b)}_N({\sigma^2}) &\to 0\\
I^{(b)}_N({\sigma^2}) - \bar{I}^{(b)}_N({\sigma^2}) &\asto 0
\end{align*}
where 
 \begin{align}\nonumber
	 \bar{I}^{(b)}_N({\sigma^2}) &= \bar{V}_N({\sigma^2}) + \frac1N\sum_{k=1}^K\log\det\LB\LSB\bar{c}_k-b_k\bar{b}_k\RSB\Id_{n_k}+b_k\Pm_k\RB + \sum_{k=1}^K(1-c_k)\bar{c}_k\log(\bar{c}_k-b_k\bar{b}_k)-\bar{c}_k\log(\bar{c}_k)\\\label{eq:mutinfeq}
\bar{V}_N({\sigma^2}) &= \frac1N\log\det\LB\Id_N +\frac1{\sigma^2}\frac{1}{N}\sum_{k=1}^K\sum_{j=1}^{N_k}\frac{\bar{b}_k\Rm_{k,j}}{1+\bar{b}_k\zeta_{kj}} \RB-\sum_{k=1}^K\bar{b}_kb_k + \frac1N\sum_{k=1}^K\sum_{j=1}^{N_k}\log\LB1+\bar{b}_k\zeta_{kj}\RB .
\end{align}
\end{theorem}
\begin{proof}
 The proof is provided in Appendix \ref{app:proof_mutinf}.
\end{proof}\bigskip

\begin{theorem}[SINR of the MMSE detector under (ii-a)]
  \label{th:SINR_MMSE}
  Consider the system model \eqref{eq:channel} under assumptions (i), (ii-a), (iii)-(vi) and, for ${\sigma^2}>0$, denote
  \begin{equation}
	  \label{eq:mmse}
	  \gamma^{N(a)}_{kj}({\sigma^2}) = p_{kj}\wv_{kj}\htp\Hm_k\htp\LB{\Bm_N}_{(k,j)} + {\sigma^2}\Id_N\RB^{-1}\Hm_k\wv_{kj}.
  \end{equation}
  Assume {\bf A1}, {\bf A2}, and {\bf A3}-a. Then, as $N\to\infty$,
\begin{align*}
	\gamma^{N(a)}_{kj}({\sigma^2}) - \bar{\gamma}^{N(a)}_{kj}({\sigma^2}) &\asto 0
\end{align*}
where 
\begin{equation*}
	\bar{\gamma}^{N(a)}_{kj}({\sigma^2}) = \frac{p_{kj}a_k}{\bar{c}_k-a_k \bar{a}_k}
\end{equation*} 
with $a_k=a_k({\sigma^2})$ and $\bar{a}_k=\bar{a}_k({\sigma^2})$ defined in Theorem~\ref{th:fundamental_eq}.
\end{theorem}
\begin{proof}
 The proof is provided in Appendix~\ref{app:SINR_MMSE}. 
\end{proof}\bigskip

As an (almost immediate) corollary, we have the following result.
\begin{corollary}\label{cor:MMSEa}
	Under the conditions of Theorem \ref{th:SINR_MMSE}, denote
	\begin{equation*}
		R^{(a)}_N({\sigma^2}) = \frac1N\sum_{k=1}^K\sum_{j=1}^{n_k}\log\LB1+\gamma^{N(a)}_{kj}({\sigma^2})\RB.
	\end{equation*}
	Then,
\begin{align*}
	\mathbb{E}R^{(a)}_N({\sigma^2}) - \bar{R}^{(a)}_N({\sigma^2}) &\to 0 \\
	R^{(a)}_N({\sigma^2}) - \bar{R}^{(a)}_N({\sigma^2}) &\asto 0
\end{align*}
where 
\begin{align*}
	\bar{R}^{(a)}_N({\sigma^2}) = \frac1N\sum_{k=1}^K\sum_{j=1}^{n_k}\log\LB1+\bar{\gamma}^{N(a)}_{kj}(\sigma^2)\RB.
\end{align*}
\end{corollary}
\begin{proof}
 The proof is provided in Appendix \ref{app:SINR_MMSE}. 
\end{proof}\bigskip

\begin{theorem}[SINR of the MMSE detector under (ii-b)]
	\label{th:sinr}
 Consider the system model \eqref{eq:channel} under assumptions (i), (ii-b), (iii)-(vi) and, for ${\sigma^2}>0$, denote
  \begin{equation*}
	  \gamma^{N(b)}_{kj}({\sigma^2}) = p_{kj}\wv_{kj}\htp\Hm_k\htp\LB{\Bm_N}_{(k,j)} + {\sigma^2}\Id_N\RB^{-1}\Hm_k\wv_{kj}.
  \end{equation*}
  Assume {\bf A1}, {\bf A2}, {\bf A3}-b, and {\bf A4}. Then, as $N\to\infty$,
\begin{align*}
	\gamma^{N(b)}_{kj}({\sigma^2}) - \bar{\gamma}^{N(b)}_{kj}({\sigma^2}) \asto 0
\end{align*}
where 
\begin{align*}
	\bar{\gamma}^{N(b)}_{kj}({\sigma^2})=\frac{p_{kj} b_k}{\bar{c}_k-b_k\bar{b}_k}
\end{align*}
with $b_k=b_k({\sigma^2})$ and $\bar{b}_k=\bar{b}_k({\sigma^2})$, given by Theorem~\ref{th:fundeq}.
\end{theorem}
\begin{proof}
The proof is provided in Appendix~\ref{app:proof_mutinf}.
\end{proof}\bigskip

Similar to the quasi-static channel scenario, we also have the following corollary.
\begin{corollary}\label{cor:MMSE}
	Under the conditions of Theorem \ref{th:sinr}, denote
	\begin{equation*}
		R^{(b)}_N({\sigma^2}) = \frac1N\sum_{k=1}^K\sum_{j=1}^{n_k}\log\LB1+\gamma^{N(b)}_{kj}({\sigma^2})\RB.
	\end{equation*}
	Then,
\begin{align*}
	\mathbb{E}R^{(b)}_N({\sigma^2}) - \bar{R}^{(b)}_N({\sigma^2}) &\to 0 \\
	R^{(b)}_N({\sigma^2}) - \bar{R}^{(b)}_N({\sigma^2}) &\asto 0 
\end{align*}
where 
\begin{align*}
	\bar{R}^{(b)}_N({\sigma^2}) = \frac1N\sum_{k=1}^K\sum_{j=1}^{n_k}\log\LB1+\bar{\gamma}^{N(b)}_{kj}(\sigma^2)\RB.
\end{align*}
\end{corollary}
\begin{proof}
The proof is provided in Appendix~\ref{app:proof_mutinf}
\end{proof}

\bigskip
\begin{remark}
	Surprisingly, the fundamental equations of Theorems~\ref{th:fundamental_eq} and \ref{th:fundeq} {\it cannot} be solved with the proposed fixed-point algorithms for the case $c_k=1$ when the entries of $\Pm_k$ are all non-zero (recall that assumption (iii) of the model imposes $c_k<1$). Moreover, the proof of Theorem~\ref{th:sumRWT} in the appendix cannot be easily extended to this case. However, if $\Pm_k= p_k\Id_{N_k}$, for some $p_k>0$, the random matrix $\Bm_N$ reduces to
\begin{align}
 \Bm_N = \sum_{k=1}^K p_k\Hm_k\Hm_k\htp = \sum_{k=1}^K p_k\Rm_k.
\end{align}
For the quasi-static channel scenario, $\Bm_N$ is thus entirely deterministic. A careful inspection of the fixed-point equations of Theorem~\ref{th:fundamental_eq} reveals that $\bar{a}_k$, with definition extended to $c_k=1$, has two solutions in the adherence of $[0,\bar{c}_k/a_k)$, i.e., $\bar{a}_k = \frac{\bar{c}_k}{a_k}$ or $\bar{a}_k = p_k$. Simulations suggest that, in this scenario, the fixed-point algorithm proposed in Theorem~\ref{th:fundamental_eq} may converge to either of the solutions depending on the choice of the system parameters. Note that, for $\bar{a}_k=p_k$, Theorem~\ref{th:mutual_info} reduces to
$$\bar{I}_N^{\text(a)}(\sigma^2) = \frac1N\log\det\LB\Id_N + \frac1{\sigma^2}\sum_{k=1}^K p_k \Rm_k\RB$$
as it should be. As for $\bar{a}_k = \frac{\bar{c}_k}{a_k}$, this cannot lead to a correct solution as $\bar{I}_N^{\text(a)}(\sigma^2)$ would be independent of $p_k$. These observations are consistent with the condition $\bar{a}_k< \frac{\bar{c}_k}{a_k}$. Similarly, $\bar{b}_k$ in Theorem~\ref{th:fundeq} has the same two possible solutions in this scenario. With $\bar{b}_k=p_k$, the asymptotic mutual information reduces to
\begin{align*}
 \bar{I}_N^{\text(b)}(\sigma^2) = \bar{V}_N(\sigma^2)
\end{align*}
which is the asymptotic mutual information of a channel with a generalized variance profile as provided in Theorem~\ref{th:logdetcor} (Appendix~\ref{sec:appendixB}). Thus our results are consistent for the case $c_k=1$ and $\Pm_k=p_k\Id_{N_k}$. However, if the entries of $\Pm_k$ are not all equal and $c_k=1$, we cannot easily infer the solutions of $\bar{a}_k,\ \bar{b}_k$ and the proposed fixed point algorithms may not converge to the correct solutions.  
\end{remark}\bigskip

\begin{remark}\label{rem:ci_1}
Based on the previous remark, under scenario (ii-b) with $K=1$, $\Pm_1=\Id_{n_1}$, $N_1=n_1=N$, and  $\Rm_{1j}=\Id_N$ for all $j$, the set of implicit equations in Theorem~\ref{th:fundeq} reduces to:
\begin{align*}
	\bar{b}({\sigma^2})=1, \qquad g({\sigma^2}) = \frac{\zeta({\sigma^2})}{1+\zeta({\sigma^2})},\qquad \zeta({\sigma^2})= \frac{1}{\frac{1}{1+\zeta({\sigma^2})}+{\sigma^2}}
\end{align*}
which has a unique solution satisfying $\zeta({\sigma^2})\ge0$ and that can be given in closed-form:
\begin{align*}
 \zeta({\sigma^2}) = \frac{-1+\sqrt{1+\frac{4}{{\sigma^2}}}}{2}.
\end{align*}
We recognize that $\zeta({\sigma^2})$ is the Stieltjes transform of the Mar\u{c}enko-Pastur law with scale parameter $1$ \cite[Equation~(3.20)]{COUbook} evaluated on the negative real axis. This result is consistent with our expectations since $\Bm_N=\Zm_1\Zm_1\htp$, where $\Zm_1\in\CC^{N\times N}$ has i.i.d.\@ entries with zero mean and variance $1/N$. Moreover, the expression of the normalized asymptotic mutual information as given in Theorem~\ref{th:mutinf} reduces to
\begin{align*}
	\bar{I}^{(b)}_N({\sigma^2}) = \bar{V}_N({\sigma^2}) = \log\LB1+\zeta({\sigma^2})+1/{\sigma^2}\RB - \frac{\zeta({\sigma^2})}{1+\zeta({\sigma^2})}
\end{align*}
which is consistent with the asymptotic spectral efficiency of a Rayleigh-fading $N\times N$ MIMO channel \cite[Equation (9)]{verdu99} (see also \cite[Section 13.2.2]{COUbook}). Equivalently, the asymptotic SINR of the MMSE detector and the associated normalized sum-rate can be given as (cf. \cite[Proposition VI.1]{verdu99}):
\begin{align*}
	\bar{\gamma}_j^{N(b)} = \zeta({\sigma^2}),\qquad \bar{R}^{(b)}_N({\sigma^2})=\log(1+\zeta({\sigma^2})).
\end{align*}
\end{remark}
\bigskip

\begin{remark}
	Technically, the results obtained for the quasi-static scenario unfold from the Stieltjes transform framework very similar to \cite{COU09}, \cite{HAC07}. However, some new tools are introduced which simplify the analysis made in these papers, such as the method of standard interference functions to prove existence and uniqueness of the derived deterministic equivalents. As for the results in the fading channel scenario, they unfold from the conjugation of the results obtained in the quasi-static scenario and the results obtained in \cite{WAG10} (recalled in Appendix \ref{sec:appendixB}) for a channel model similar to \eqref{eq:channel} but without the presence of the $\W_k$ matrices. The central tool to allow this conjugation is the Tonelli (or Fubini) theorem, Lemma~\ref{le:tonelli} in Appendix \ref{sec:fundlemmas}, on the product probability space engendering both the (sequences of growing) $\W_k$ and $\H_k$ matrices. 
\end{remark}

\section{Numerical results}
\label{sec:num}
\label{sec:simu}

The results of Section~\ref{sec:results} enable a simple characterization of different performance measures of isometric precoded multi-user systems with large dimensional quasi-static or fading channels, some of which were introduced in Section~\ref{sec:intro}. In the following, we apply these results to three practical examples.

\subsection{Uplink orthogonal SDMA with inter-cell interference}

\begin{figure}
\centering
\pgfimage[width=0.45\textwidth]{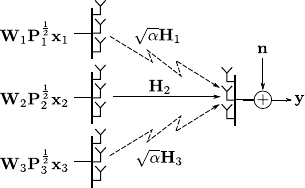}
\caption{Three-cell example: The BS in the center cell decodes the $n$ streams from the UT in its own cell while treating the other signals as interference.\label{fig:system_model}}
\end{figure}

In this first example, we apply the theoretical results of Section \ref{sec:results} under the quasi-static channel scenario (hypothesis (ii-a)) to the uplink channel of an orthogonal SDMA scheme with inter-cell interference. We consider a three cell system with one active user terminal (UT) per cell. The UT in cell $k$ is equipped with $N_k$ transmit antennas. We focus on the central cell, whose base station (BS) is equipped with $N$ antennas, and assume that the signals received from neighboring cells are treated as noise. This setup is schematically depicted in Figure \ref{fig:system_model}. The received signal $\y$ at the BS reads
\begin{equation*}
	\y = \H_2\W_2\P_2^\oh\x_2 + \underbrace{\sqrt{\alpha} \H_1\W_1\P_1^\oh\x_1 + \sqrt{\alpha}\H_3\W_3\P_3^\oh\x_3 +\n}_{\triangleq\z} 
\end{equation*}
with $\H_i\in\CC^{N\times N_i}$ the channel matrix from UT $i$ to the BS, $\x_i\sim \mathcal{CN}(0,\I_{n_i})$ the transmit symbol of UT $i$, $\W_i\in\CC^{N_i\times n_i}$ the isometric precoding matrix composed of $n_i$ orthogonal vectors and $0<\alpha<1$ an inter-cell interference factor. The vector $\z\in\CC^N$ combines the inter-cell interference and the thermal noise. The covariance matrix $\Z\in\CC^{N\times N}$ of $\z$ is given as
\begin{equation*}
\Z = \Exp{\z\z^\herm} = \alpha \left[\H_1\W_1\P_1\W_1^\herm \H_1^\herm + \H_3\W_3\P_3\W_3^\herm \H_i^\herm \right] + {\sigma^2}\I_N.
\end{equation*}

We assume an SDMA system with channel matrices $\H_k\in\CC^{N\times N_k}$ generated as realizations of a random standard Gaussian matrix with entries of zero mean and variance $1/N_k$. For simplicity, we further assume that each UT uses $n_k=n$ different transmit signatures to which it assigns equal unit power, i.e., $\P_k=\I_n$. Under these assumptions, the mutual information $I_N({\sigma^2})$ of the central cell when the interference is treated as noise is given by
\begin{align*}
  I_N({\sigma^2}) &= \frac1N\log\det\left(\I_N + \Z^{-\frac12}\H_2\W_2\W_2^\herm\H_2\Z^{-\frac12}\right) \\
  &= \frac1N\log\det\left(\I_N + \frac1{{\sigma^2}}\sum_{k=1}^3\H_k\W_k\W_k^\herm\H_k\right) - \frac1N\log\det\left(\I_N + \frac1{{\sigma^2}}\sum_{\substack{k=1\\ k\neq 2}}^3\H_k\W_k\W_k^\herm\H_k\right).
\end{align*}

According to \cite{SIL98}, the spectral norm of $\H_k\H_k^\herm$ is almost surely uniformly bounded. For such channel realizations, we are therefore in the conditions of Theorem~\ref{th:mutual_info}. As a consequence, $I_N({\sigma^2})-\bar{I}_N({\sigma^2})\asto 0$, with $\bar{I}_N$ defined in Theorem \ref{th:mutual_info} (termed $\bar{I}^{(a)}_N$). An approximation of the SINR at the output of the MMSE receiver for the $j${th} entry of $\x_2$ can also be computed directly by Theorem \ref{th:SINR_MMSE}. We assume $\alpha=0.25$, $N=16$, $N_1=N_2=N_3=8$ and define $\text{SNR}=1/{\sigma^2}$. We consider a single random realization of the matrices $\H_k$, which is assumed to be static and therefore deterministically known.

Figure \ref{fig:mutual_info} depicts $I_N({\sigma^2})$ and the deterministic equivalent $\bar{I}_N({\sigma^2})$ versus $\text{SNR}$ for different values of $n\in\{1,4,8\}$, scaled to bits/s/Hz instead of nats. Note that for the case $n=8$, the matrix $\Bm_N$ and, thus,  the mutual information are deterministic (see Remark~\ref{rem:ci_1}). We observe a very accurate fit between both results over the full range of $\text{SNR}$ and $n$. This validates the deterministic approximation of the mutual information for systems of even small dimensions. It appears that, at moderate $\text{SNR}$, i.e., when noise dominates interference, the results suggest that using all available data streams (or orthogonal transmit signatures) maximizes the rate of the central cell. On the contrary, at high $\text{SNR}$, the achieved mutual information is maximal when fewer than $N$ transmit signatures are used. These results corroborate the observations of \cite{blum2003}. Additionally, we can perform optimal stream control by numerical comparison of the deterministic equivalents of the achievable rates for each $n$. Such an optimization is performed in Section \ref{sec:interference} for the two-user interference channel.

\begin{figure}
  \centering
  \begin{tikzpicture}[scale=1]
    \renewcommand{\axisdefaulttryminticks}{4}
    \pgfplotsset{every major grid/.append style={densely dashed}}
    \tikzstyle{every axis y label}+=[yshift=-10pt]
    \tikzstyle{every axis x label}+=[yshift=5pt]
    \pgfplotsset{every axis legend/.append style={cells={anchor=west},fill=white, at={(0.02,0.98)}, anchor=north west}}
    \begin{axis}[
      xmin=-5,
      ymin=0,
      xmax=30,
      ymax=1.8,
      bar width=3pt,
      grid=major,
      scaled ticks=true,
      ylabel={$I_N(\sigma^2)$ [bits/s/Hz]},
      xlabel={{\rm SNR} [dB]}
      ]
      \addplot[black,smooth] plot coordinates{
      (-5,0.030297) (0,0.066540) (5,0.119120) (10,0.182050) (15,0.250374) (20,0.321045) (25,0.392577) (30,0.464397) 
      };

      \addplot[red,only marks,mark=*,mark options={scale=0.75},error bars/.cd,y dir=both,y explicit, error bar style={mark size=2.5pt}] plot coordinates{
      (-5,0.030333) +-(0.006501,0.006501) (0,0.066791) +-(0.011012,0.011012) (5,0.119404) +-(0.014716,0.014716) (10,0.182380) +-(0.017007,0.017007) (15,0.250934) +-(0.017914,0.017914) (20,0.321485) +-(0.018012,0.018012) (25,0.392825) +-(0.018163,0.018163) (30,0.464781) +-(0.018251,0.018251) 
      };

      \addplot[black,smooth] plot coordinates{
      (-5,0.123635) (0,0.255115) (5,0.433663) (10,0.643309) (15,0.879068) (20,1.138599) (25,1.414302) (30,1.697749) 
};

      \addplot[red,only marks,mark=*,mark options={scale=0.75},error bars/.cd,y dir=both,y explicit, error bar style={mark size=2.5pt}] plot coordinates{
      (-5,0.123634) +-(0.008375,0.008375) (0,0.255066) +-(0.014449,0.014449) (5,0.434021) +-(0.020492,0.020492) (10,0.643461) +-(0.026844,0.026844) (15,0.879569) +-(0.032923,0.032923) (20,1.138196) +-(0.037653,0.037653) (25,1.414829) +-(0.040537,0.040537) (30,1.698557) +-(0.042252,0.042252) 
};

      \addplot[black,smooth] plot coordinates{
      (-5,0.204115) (0,0.405681) (5,0.653104) (10,0.900385) (15,1.122539) (20,1.306486) (25,1.446863) (30,1.545036) 
};

      \addplot[red,only marks,mark=*,mark options={scale=0.75},error bars/.cd,y dir=both,y explicit, error bar style={mark size=2.5pt}] plot coordinates{
      (-5,0.204115) +-(0.000000,0.000000) (0,0.405681) +-(0.000000,0.000000) (5,0.653104) +-(0.000000,0.000000) (10,0.900385) +-(0.000000,0.000000) (15,1.122539) +-(0.000000,0.000000) (20,1.306486) +-(0.000000,0.000000) (25,1.446863) +-(0.000000,0.000000) (30,1.545036) +-(0.000000,0.000000) 
};

{\small
\node[] at (axis cs:10,1.1) {$n=8$}; 
\node[] at (axis cs:20,0.95) {$n=4$}; 
\node[] at (axis cs:26.5,0.3) {$n=1$}; 
}
      \legend{{deterministic equivalent},{simulation}}
    \end{axis}
  \end{tikzpicture}
  \caption{Mutual information $I_N({\sigma^2})$ versus ${\rm SNR}$ for different numbers of transmit signatures $n$, $N=16$, $N_i=8$, $\P_i=\I_n$, $\alpha=0.5$. Error bars represent one standard deviation on each side.}
  \label{fig:mutual_info}
\end{figure}
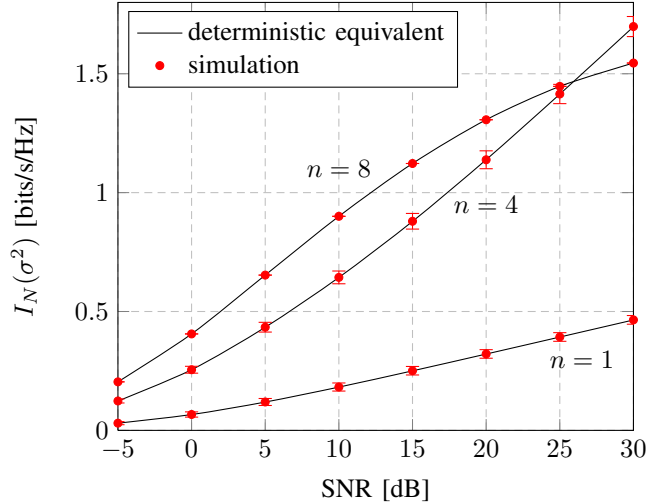

In Figure \ref{fig:sinr}, we compare the per-receive antenna sum rate $R_N({\sigma^2})$ with single-stream MMSE-detection to the associated deterministic equivalent $\bar{R}_N({\sigma^2})$, for the same system conditions as in Figure \ref{fig:mutual_info}. The sum rate $R_N({\sigma^2})$ is explicitly given by
\begin{equation*}
  R_N({\sigma^2}) = \frac1N\sum_{k=1}^n\log \left(1+\gamma^N_{2,k}({\sigma^2})\right)
\end{equation*}
with $\gamma^N_{ij}({\sigma^2})$ defined in \eqref{eq:mmse} (termed $\gamma^{N(a)}_{ij}({\sigma^2}))$. As for $\bar{R}_N({\sigma^2})$, from Theorem \ref{th:SINR_MMSE}, it reads
\begin{equation*}
	\bar{R}_N({\sigma^2}) = c_2\bar{c}_2\log\left( 1 + \frac{a_2({\sigma^2})}{\bar{c}_2-a_2({\sigma^2})\bar{a}_2({\sigma^2})} \right)
\end{equation*}
with $a_2({\sigma^2})$ and $\bar{a}_2({\sigma^2})$ defined in Theorem \ref{th:sumRWT}. For the case $n=8$, we have used $\bar{a}_k=1$ to compute the deterministic equivalents (see Remark~\ref{rem:ci_1}). Similar to the previous observations, the deterministic equivalent provides an accurate approximation for all values of $\text{SNR}$ and $n$, although the precision is slightly less than for the mutual information in Figure \ref{fig:mutual_info}. The same conclusions regarding optimal stream control also hold for the MMSE decoder, where we confirm an interest to perform stream control when the interference dominates the background noise.

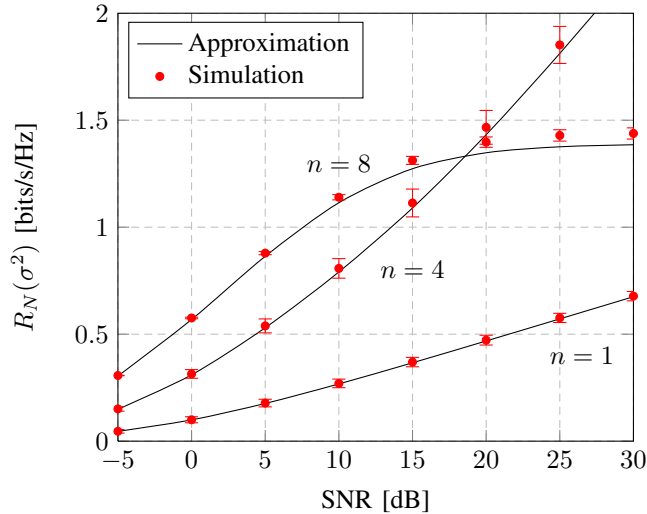
\begin{figure}
  \centering
  \begin{tikzpicture}[scale=1]
    \renewcommand{\axisdefaulttryminticks}{4}
    \pgfplotsset{every major grid/.append style={densely dashed}}
    \tikzstyle{every axis y label}+=[yshift=-10pt]
    \tikzstyle{every axis x label}+=[yshift=5pt]
    \pgfplotsset{every axis legend/.append style={cells={anchor=west},fill=white, at={(0.02,0.98)}, anchor=north west}}
    \begin{axis}[
      xmin=-5,
      ymin=0,
      xmax=30,
      ymax=2,
      bar width=3pt,
      grid=major,
      scaled ticks=true,
      ylabel={$R_N({\sigma^2})$ [bits/s/Hz]},
      xlabel={{\rm SNR} [dB]}
      ]
      \addplot[black,smooth] plot coordinates{
      (-5,0.045785) (0,0.099413) (5,0.176272) (10,0.267546) (15,0.366280) (20,0.468284) (25,0.571498) (30,0.675116) 
      };

      \addplot[red,only marks,mark=*,mark options={scale=0.75},error bars/.cd,y dir=both,y explicit, error bar style={mark size=2.5pt}] plot coordinates{
      (-5,0.046038) +-(0.008230,0.008230) (0,0.100012) +-(0.014079,0.014079) (5,0.178060) +-(0.017816,0.017816) (10,0.269945) +-(0.020056,0.020056) (15,0.369594) +-(0.021947,0.021947) (20,0.472104) +-(0.022859,0.022859) (25,0.575967) +-(0.021502,0.021502) (30,0.677749) +-(0.021959,0.021959) 
      };

      \addplot[black,smooth] plot coordinates{
      (-5,0.149292) (0,0.309933) (5,0.530119) (10,0.790856) (15,1.089890) (20,1.432193) (25,1.812010) (30,2.213498) 
};

      \addplot[red,only marks,mark=*,mark options={scale=0.75},error bars/.cd,y dir=both,y explicit, error bar style={mark size=2.5pt}] plot coordinates{
      (-5,0.150752) +-(0.010505,0.010505) (0,0.314201) +-(0.020196,0.020196) (5,0.538681) +-(0.032589,0.032589) (10,0.807380) +-(0.045730,0.045730) (15,1.113133) +-(0.064991,0.064991) (20,1.466671) +-(0.078969,0.078969) (25,1.851870) +-(0.086617,0.086617) (30,2.254981) +-(0.094006,0.094006) 
};

      \addplot[black,smooth] plot coordinates{
      (-5,0.304390) (0,0.568345) (5,0.863874) (10,1.114322) (15,1.273300) (20,1.347848) (25,1.375826) (30,1.385238) 
};

      \addplot[red,only marks,mark=*,mark options={scale=0.75},error bars/.cd,y dir=both,y explicit, error bar style={mark size=2.5pt}] plot coordinates{
      (-5,0.306859) +-(0.001226,0.001226) (0,0.575686) +-(0.003587,0.003587) (5,0.878732) +-(0.007154,0.007154) (10,1.139987) +-(0.012245,0.012245) (15,1.312178) +-(0.018751,0.018751) (20,1.397554) +-(0.024253,0.024253) (25,1.428833) +-(0.026505,0.026505) (30,1.438116) +-(0.026138,0.026138) 
};

{\small
\node[] at (axis cs:10,1.3) {$n=8$}; 
\node[] at (axis cs:15,0.8) {$n=4$}; 
\node[] at (axis cs:26.5,0.4) {$n=1$}; 
}
      \legend{{Approximation},{Simulation}}
    \end{axis}
  \end{tikzpicture}
  \caption{Sum rate $R_N({\sigma^2})$ at the output of the MMSE decoder for user $2$ versus SNR for different numbers of transmit signatures $n$, $N=16$, $N_i=8$, $\P_i=\I_n$, $\alpha=0.5$. Error bars represent one standard deviation on each side.}
  \label{fig:sinr}
\end{figure}

\subsection{Multiple access channel}\label{sec:MAC}
In this and the following example, we apply the theoretical results of Section \ref{sec:results} under the fading channel scenario (hypothesis (ii-b)).
We consider a MAC from three transmitters to a single receiver as shown in Figure~\ref{fig:MAC}. The channel from each transmitter to the receiver is modeled by the Kronecker model (see Remark~\ref{rem:kronecker}) with individual transmit and receive covariance matrices $\Tm_k$ and $\Rm_k$ and we assume additionally a different path loss $\alpha_k>0$ on each link. The received signal vector $\yv$ for this model reads
\begin{align*}
 \yv = \sum_{k=1}^3 \sqrt{\alpha_k} \Rm_k^{\frac12}\Zm_k\Tm_k^{\frac12}\Wm_k\Pm_k^{\frac12}\xv_k + \nv
\end{align*}
where $\xv_k\sim\Cc\Nc(\zerov,\I_{N_k})$ and $\nv\sim\Cc\Nc(\zerov,{\sigma^2}\Id_N)$. 
We create the correlation matrices according to a generalization of Jakes' model with non-isotropic signal transmission, see, e.g., \cite{CHI00,PED00,MOU00}, where the elements of $\Tm_k$ and $\Rm_k$ are given as
\begin{align}\nonumber
 \LSB\Tm_k\RSB_{ij} &= \frac{1}{\theta^{t,k}_\text{max}-\theta^{t,k}_\text{min}}\int_{\theta^{t,k}_\text{min}}^{\theta^{t,k}_\text{max}}\exp\LB\frac{{\bf{i}} 2\pi}{\lambda}d^{t,k}_{ij}\cos\LB\theta\RB\RB d\theta\\\label{eq:macchnmodel}
\LSB\Rm_k\RSB_{ij} &= \frac{1}{\theta^{r,k}_\text{max}-\theta^{r,k}_\text{min}}\int_{\theta^{r,k}_\text{min}}^{\theta^{r,k}_\text{max}}\exp\LB\frac{{\bf{i}} 2\pi}{\lambda}d^r_{ij}\cos\LB\theta\RB\RB d\theta
\end{align}
where $(\theta^{t,k}_\text{min},\theta^{t,k}_\text{max})$ and $(\theta^{r,k}_\text{min},\theta^{r,k}_\text{max})$ determine the azimuth angles over which useful signal power for the $k$th transmitter is radiated or received, $d^{t,k}_{ij}$ and $d^{r}_{ij}$ are the distances between the antenna elements $i$ and $j$ at the $k$th transmitter and receiver, respectively, and $\lambda$ is the signal wavelength. We assume uniform power allocation for all $k$, i.e., $\Pm_k=\frac1{n_k}\Id_{n_k}$, and define $\text{SNR}=1/{\sigma^2}$. All other parameters are summarized in Table~\ref{tab:param}. 

\begin{figure}
\centering \includegraphics[width=0.45\textwidth]{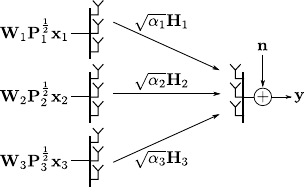}
\caption{MIMO MAC from three transmitters ($k=1,2,3$) with $N_k$ antennas to a receiver with $N$ antennas. Each transmitter sends $n_k$ streams with precoding matrix $\Wm_k$ and power allocation $\Pm_k$ over the channel $\sqrt{\alpha_k}\Hm_k$.\label{fig:MAC}}
\end{figure}

Figure~\ref{fig:MAC_mutinf} compares the normalized mutual information $I_N({\sigma^2})$ and the normalized rate with MMSE decoding $R_N({\sigma^2})$, averaged over $10,000$ different realizations of the matrices $\Hm_k$ and $\Wm_k$, against their deterministic approximations $\bar{I}_N({\sigma^2})$ and $\bar{R}_N({\sigma^2})$. Although we have chosen small dimensions for all matrices (see Table~\ref{tab:param}), the match between both results is almost perfect. Also the fluctuations of  $I_N({\sigma^2})$ and $R_N({\sigma^2})$ are rather small as can be seen from the error bars representing one standard deviation in each direction. 

\begin{table}
\renewcommand{\arraystretch}{1.3}
\caption{Simulation parameters for Figure~\ref{fig:MAC_mutinf}: $N=10$, $d^{r}_{ij}=8\lambda(i-j)$}
\label{tab:param}
\centering
\begin{tabular}{c|cccccccc}
$k$ & $N_k$ & $n_k$ & $\theta^{t,k}_\text{min}$ & $\theta^{t,k}_\text{max}$ & $\theta^{r,k}_\text{min}$ & $\theta^{r,k}_\text{max}$ & $d^{t,k}_{ij}$& $\alpha_k$ \\
\hline\hline
1 & 10 & 8 & $0$      & $\pi/2$ & $-\pi/4$ & $0$     & $4\lambda(i-j)$ & $1$   \\\hline
2 & 5  & 4 & $-\pi/4$ & $\pi/4$ & $0$      & $\pi/3$ & $4\lambda(i-j)$ & $1/2$ \\\hline
3 & 5  & 4 & $-\pi/2$ & $0$     & $-\pi/3$ & $\pi/3$ & $4\lambda(i-j)$ & $1/2$ \\\hline
\end{tabular}
\end{table}

\begin{figure}
\centering
 \includegraphics[width=0.6\textwidth]{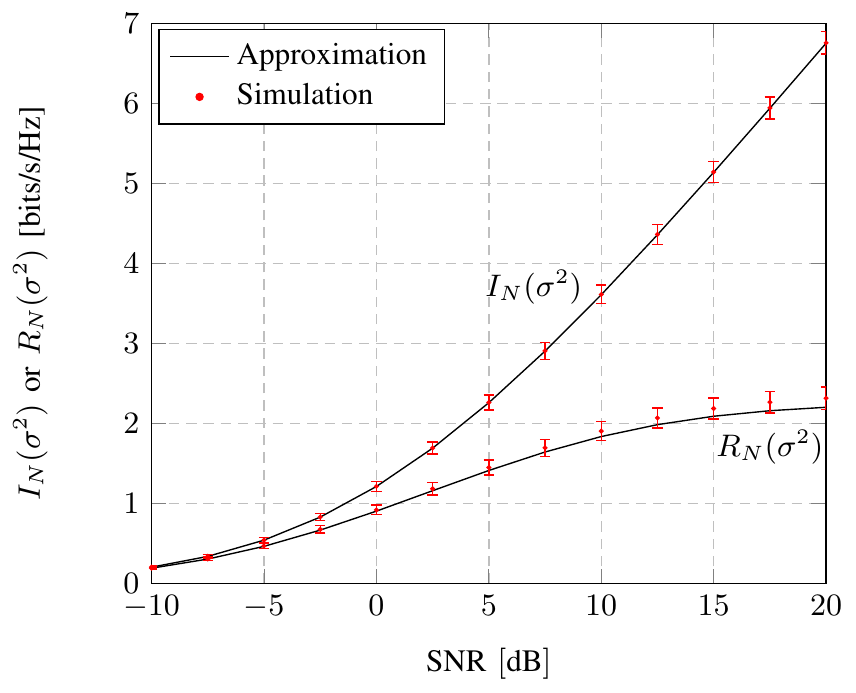}
\caption{Comparison of the average normalized mutual information $I_N({\sigma^2})$ and the normalized rate with MMSE decoding $R_N({\sigma^2})$ with their deterministic approximations $\bar{I}_N({\sigma^2})$ and $\bar{R}_N({\sigma^2})$. Error bars represent one standard deviation in each direction.\label{fig:MAC_mutinf}}
\end{figure}

\subsection{Stream-control in interference channels}
\label{sec:interference}
\begin{figure}
\centering \includegraphics[width=0.50\textwidth]{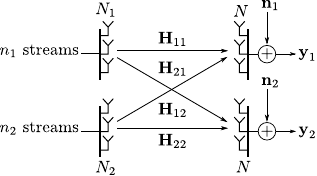}
\caption{Interference channel from two transmitters with $N_k$ ($k=1,2$) antennas, respectively, to two receivers with $N$ antennas each. Each transmitter sends $n_k$  independent data streams to its respective receiver.\label{fig:IC}}
\end{figure}

Our last example considers a MIMO interference channel consisting of two transmitter-receiver pairs as depicted in Figure~\ref{fig:IC}. The received signal vectors $\yv_1,\yv_2\in\CC^N$ are respectively given as
\begin{align*}
 \yv_1 & = \Hm_{11}\Wm_1\P_1^\oh\xv_1 + \Hm_{12}\Wm_2\P_2^\oh\xv_2 + \nv_1\\
 \yv_2 & = \Hm_{21}\Wm_1\P_1^\oh\xv_1 + \Hm_{22}\Wm_2\P_2^\oh\xv_2 + \nv_2
\end{align*}
where $\Hm_{qk}\in\CC^{N\times N_k}$, $\Wm_k\in\CC^{N_k\times N_k}$, $\xv_k\sim\Cc\Nc(\zerov,\I_{N_k})$, $\Pm_k\in\RR_+^{N_k\times N_k}$ satisfying $\frac1{N_k}\trace\Pm_k=1$, and $\nv_k\sim\Cc\Nc(\zerov,{\sigma^2}\Id_N)$, for $q,k\in\{1,2\}$. Assuming that the receivers are aware of both precoding matrices and their respective channels but treat the interfering transmission as noise, the normalized mutual informations between $\xv_1$ and $\yv_1$, and $\xv_2$ and $\yv_2$, are respectively given as
\begin{align*}
 I_1({\sigma^2}) & = \frac1N\log\det\LB\Id_N + \frac1{\sigma^2}\sum_{k=1}^2\Hm_{1k}\Wm_k\Pm_k\Wm_k\htp\Hm_{1k}\htp\RB-\frac1N\log\det\LB\Id_N+\frac1{\sigma^2}\Hm_{12}\Wm_2\Pm_2\Wm_2\htp\Hm_{12}\htp\RB\\
I_2({\sigma^2}) & = \frac1N\log\det\LB\Id_N + \frac1{\sigma^2}\sum_{k=1}^2\Hm_{2k}\Wm_k\Pm_k\Wm_k\htp\Hm_{2k}\htp\RB-\frac1N\log\det\LB\Id_N+\frac1{\sigma^2}\Hm_{21}\Wm_1\Pm_1\Wm_1\htp\Hm_{21}\htp\RB.
\end{align*}
We adopt the same channel model as in Section~\ref{sec:MAC}, where the channel matrices $\Hm_{qk}$ are given as
\begin{align*}
 \Hm_{qk} = \Rm_{qk}^{\frac12}\Zm_{qk}\Tm_k^{\frac12}
\end{align*}
where $\Zm_{qk}\in\CC^{N\times N_k}$ have independent $\Cc\Nc(0,1/N)$ entries and $\Tm_k$ and $\Rm_{qk}$ are calculated according to \eqref{eq:macchnmodel}. We assume that no channel state information is available at the transmitters, so that the matrices $\Pm_k$ are simply used to determine the number of independently transmitted streams: 
\begin{align*}
\Pm_k=\frac{N_k}{n_k}\diag\LB\underbrace{1,\dots,1}_{n_k},\underbrace{0,\dots,0}_{N_k-n_k}\RB. 
\end{align*}
We will now apply the previously derived results to find the optimal number of streams $(n_1^\star,n_2^\star)$ maximizing the normalized ergodic sum-rate of the interference channel above. That is, we seek to find
\begin{align*}
 (n_1^\star,n_2^\star)\ & =\ \max_{n_1,n_2} \mathbb{E}\LSB  I_1({\sigma^2}) +  I_2({\sigma^2})\RSB\\
&\quad\text{s.t. } 1\leq n_1 \leq N_1,\ 1\leq  n_2 \leq N_2
\end{align*}
where the expectation is with respect to both channel and precoding matrices. Due to the complexity of the random matrix model, this optimization problem appears intractable by exact analysis. At the same time, any solution based on an exhaustive search in combination with Monte Carlo simulations becomes quickly prohibitive for large $N_1,N_2$, since $N_1\times N_2$ possible combinations need to be tested. Relying on Theorem~\ref{th:mutinf}, we can calculate an approximation of $\mathbb{E}\LSB  I_1({\sigma^2}) +  I_2({\sigma^2})\RSB$ to find an approximate solution which becomes asymptotically exact as $N_1$ and $N_2$ grow large. Thus, we determine $(\bar{n}_1^\star,\bar{n}_2^\star)$ as the solution to
\begin{align*}
 (\bar{n}_1^\star,\bar{n}_2^\star)\ & =\ \max_{n_1,n_2}  \bar{I}_1({\sigma^2}) +  \bar{I}_2({\sigma^2})\\
&\quad\text{s.t. } 1\leq n_1 \leq N_1,\ 1\leq  n_2 \leq N_2
\end{align*}
where $\bar{I}_1({\sigma^2}),\bar{I}_2({\sigma^2})$ are calculated based on a direct application of Theorem~\ref{th:mutinf} to each of the two log-det terms in $I_1({\sigma^2})$ and $I_2({\sigma^2})$, respectively. The optimal values $(\bar{n}_1^\star,\bar{n}_2^\star)$ are then found by an exhaustive search over all possible combinations. Although we still need to compute $N_1\times N_2$ values, this is computationally much cheaper than Monte Carlo simulations. Although Theorem~\ref{th:mutinf} does not hold for the case $n_i=N_i$ in general, we can compute a deterministic equivalent of the mutual information by letting $\bar{b}_k=1$ since $\Pm_k=\Id_{N_k}$ (see Remark~\ref{rem:ci_1}). In this case, the matrices $\Wm_k$ vanish and $\bar{I}^{\text{(b)}}_N(\sigma^2)$ reduces to the deterministic equivalent of the mutual information of a channel with a variance profile as given by Theorem~\ref{th:logdetcor} in Appendix~\ref{sec:appendixB}. 

Figure~\ref{fig:IC_SNR0} and Figure~\ref{fig:IC_SNR40} show the average normalized sum-rate $\mathbb{E}\LSB  I_1({\sigma^2}) +  I_2({\sigma^2})\RSB$ and the deterministic approximation $\bar{I}_1({\sigma^2}) +  \bar{I}_2({\sigma^2})$, by Theorem~\ref{th:mutinf}, as a function of $(n_1,n_2)$ for the simulation parameters as given in Table~\ref{tab:paramIC}. We have assumed $\text{SNR}=0\,\text{dB}$ and $\text{SNR}=40\,\text{dB}$ in Figure~\ref{fig:IC_SNR0} and Figure~\ref{fig:IC_SNR40}, respectively. In both figures, the solid grid represents simulation results and the markers the deterministic approximations. We observe here again an almost perfect overlap between both sets of results for all values of $(n_1,n_2)$. The optimal values $(n_1^\star,n_2^\star)$ and $(\bar{n}_1^\star,\bar{n}_2^\star)$ coincide for both values of $\text{SNR}$ and are indicated by large crosses.
At low SNR, both transmitters should send as many independent streams as transmit antennas, i.e., $n_1=n_2=10$. At high SNR, one transmitter should use only a single stream ($n_2=1$) and the other transmitter $n_1=N-1=9$ streams. These results are in line with the observations of \cite{blum2003}.

Obviously, the last optimization problem is highly unfair and better solutions can be achieved by using different objective functions, such as weighted sum-rate maximization. Also optimal stream-control with MMSE decoding could be carried out in a similar manner. Although we would still need to perform an exhaustive search over all possible combinations of $n_1,n_2$, the computations based on deterministic equivalents are significantly faster than simulation-based approaches. The development of more intelligent algorithms to determine $(\bar{n}^\star_1,\bar{n}^\star_2)$ is outside the scope of this paper and left to future work. The extension to more than two transmitter-receiver pairs is straightforward.

\begin{table}
\renewcommand{\arraystretch}{1.3}
\caption{Simulation parameters for Figure~\ref{fig:IC_SNR0} and \ref{fig:IC_SNR40}: $N=10$, $d^{r,k}_{ij}=4\lambda(i-j)$, $d^{t,k}_{ij}=4\lambda(i-j)$}
\label{tab:paramIC}
\centering
\begin{tabular}{c|ccccc}
$(q,k)$ & $N_k$ & $\theta^{t,k}_\text{min}$ & $\theta^{t,k}_\text{max}$ & $\theta^{r,q,k}_\text{min}$ & $\theta^{r,q,k}_\text{max}$  \\
\hline\hline
(1,1) & 10 & $0$      & $\pi/2$ & $-\pi/4$ & $0$     \\\hline
(1,2) & 10 & $-\pi/2$ & $0$     & $0$      & $\pi/4$ \\\hline
(2,1) & 10 & $0$      & $\pi/2$ & $-\pi/3$ & $0$     \\\hline
(2,2) & 10 & $-\pi/2$ & $0$     & $0$      & $\pi/3$ \\\hline
\end{tabular}
\end{table}


\section{Conclusions}
\label{sec:conclusion}
In this article, we have studied a class of wireless communication channels with random unitary signature or precoding matrices over quasi-static and fast fading channels, assuming either single or multiple users and cells. For this wide range of system models, we have provided deterministic approximations of the mutual information, the SINR at the output of the MMSE receiver, and the associated sum-rate. These approximations were shown to be asymptotically accurate as the system dimensions grow large, and to be based on fixed-point solutions of a set of fundamental equations. Practical applications of these results were then proposed in the contexts of multi-cell SDMA with unitary precoders under multi-cell interference, MIMO-MAC with random unitary precoding, and interference channels with random beamforming. Simulations of the system performance demonstrate the accuracy of the approximations even for systems of small dimensions. Moreover, the deterministic equivalent framework was used to derive the sum rate maximizing number of streams to transmit in interference channels, which is intractable to solve by exact analysis.  
Lastly, we have proposed a novel technical method for the analysis of matrix models featuring random isometric matrices which goes beyond the current reach of classical free probability approaches. However, the proof for the case $\limsup c_i=1$ and arbitrary power allocation for different streams, i.e., the precoding matrices $\Wm_i$ are square and $\Pm_k\neq p_k\Id_{N_k}$, remains an open problem which might be solved with different methods.

\section*{Acknowledgements}

The authors would like to thank Prof. Loubaton for fruitful discussions on the central question of the holomorphicity of the deterministic equivalents studied in this article.

\clearpage
\begin{figure}
\centering
\includegraphics[width=0.67\textwidth]{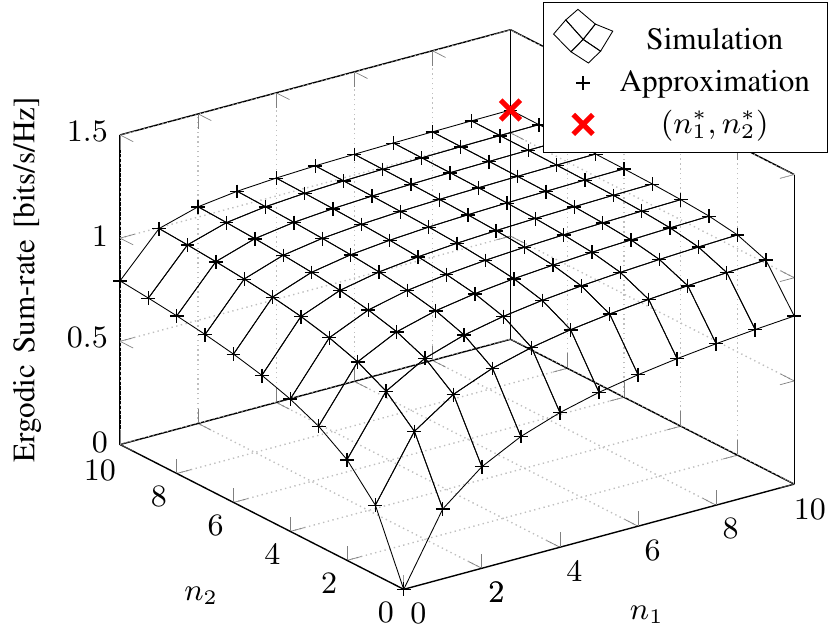}
\caption{Sum-rate versus number of transmitted data-streams $(n_1,n_2)$ for $\text{SNR}=0\,\text{dB}$ and all other parameters as provided in Table~\ref{tab:paramIC}. Solid lines correspond to simulation results, markers to the deterministic approximation by Theorem~\ref{th:mutinf}. As expected, both transmitters should send the maximum number of independent streams.\label{fig:IC_SNR0}}
\end{figure}

\begin{figure}
\centering
\includegraphics[width=0.67\textwidth]{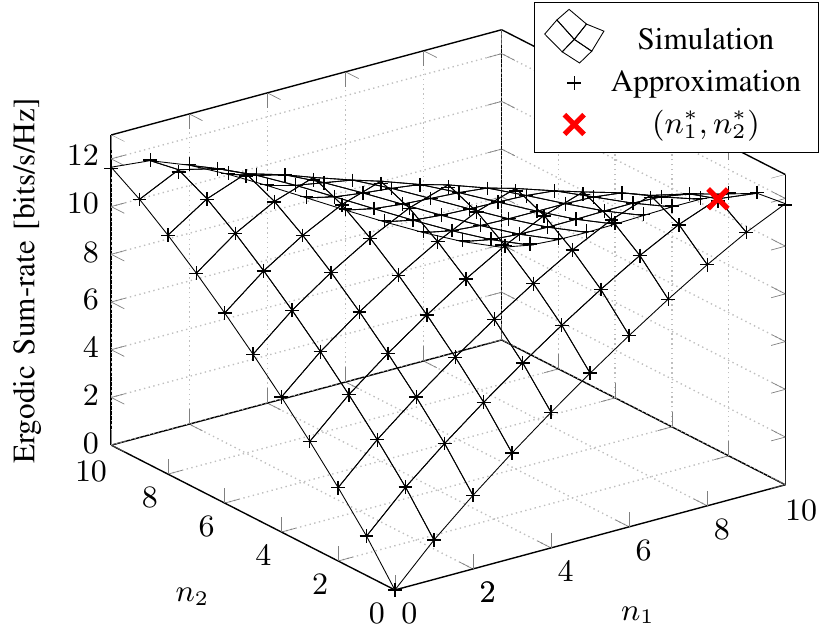}
\caption{Sum-rate versus number of transmitted data-streams $(n_1,n_2)$ for $\text{SNR}=40\,\text{dB}$ and all other parameters as provided in Table~\ref{tab:paramIC}. Solid lines correspond to simulation results, markers to the deterministic approximation by Theorem~\ref{th:mutinf}. As co-channel interference is dominant, there is a clear gain of limiting the number of transmitted streams.\label{fig:IC_SNR40}}
\end{figure}

\clearpage
\appendices

\section{Spectral approximation of $\B_N$ in the quasi-static model}
\label{app:sumRXT}

This section is dedicated to the proof of Theorem \ref{th:sumRWT} as given below. This theorem is the cornerstone result for all other results derived in this article. The  proof is based on the Stieltjes transform method which is extensively documented in \cite{SIL06,COUbook}.\bigskip

We first remind some elementary notions which are needed in the following. For a Hermitian matrix $\A\in\CC^{N\times N}$ with eigenvalues $\lambda_1\leq \ldots \leq \lambda_N$, we denote by $F^\A$ the {\it empirical spectral distribution} (e.s.d.), defined as
\begin{equation*}
	F^\A(t) = \frac1N\sum_{i=1}^N {\bm 1}_{\{\lambda_i\leq t\}}(t).
\end{equation*}
We now recall the definition of a Stieltjes transform.
\begin{definition}
	\label{def:ST}
Let $F$ be the distribution function of a probability measure with support $S$. Then, the Stieltjes transform of $F$, denoted $m_F$, is the function
\begin{align*}
	m_F:\CC\setminus S &\to \CC \\
	z&\mapsto \int \frac1{t-z}dF(t).
\end{align*}
In particular, for $F^\A$ the e.s.d. of a Hermitian matrix $\A$,
\begin{equation*}
	m_{F^\A}(z) = \frac1N\tr \left(\A - z\I_N\right)^{-1}
\end{equation*}
which will often be denoted $m_{\A}$.
\end{definition}
In the course of the derivations, some defining properties of the Stieltjes transform will be needed. These are provided in Lemma \ref{le:properties_ST} (Appendix~\ref{sec:fundlemmas}).\bigskip

\begin{theorem}
	\label{th:sumRWT}
	For $i\in\{1,\ldots,K\}$, let $\P_i\in \CC^{n_i\times n_i}$ be a Hermitian nonnegative matrix with spectral norm bounded uniformly along $n_i$ and $\W_i\in\CC^{N_i\times n_i}$ be $n_i< N_i$ columns of a unitary Haar distributed random matrix. Consider $\H_i\in\CC^{N\times N_i}$ a random matrix such that $\R_i\triangleq \H_i\H_i^\herm\in\CC^{N\times N}$ has uniformly bounded spectral norm along $N$, almost surely. Define $c_i= \frac{n_i}{N_i}$, $\bar{c}_i= \frac{N_i}{N}$, and denote 
  \begin{equation*}
    \B_N = \sum_{i=1}^K \H_i\W_i\P_i\W_i^\herm\H_i^\herm
  \end{equation*}
and $F_N$ the e.s.d. of $\B_N$.
 Then, as $N\to\infty$, with $\bar{c}_i$ and $c_i$ satisfying $0<\lim\inf \bar{c}_i\leq \lim\sup \bar{c}_i<\infty$ and $0\le \liminf c_i \leq \limsup c_i< 1$ for all $i$, the following limit holds true almost surely
  \begin{equation*}
  F_N-\bar{F}_N\Rightarrow 0
  \end{equation*}
  where $\bar{F}_N$ is the distribution function with support on $\RR_+$ and Stieltjes transform $\bar{m}_N(z)$, $z\in\CC\setminus \RR_+$. The latter is defined for $z\in D\triangleq\left\{z=x+{\bf{i}}y: x<0, |y|\le |x|\frac{1-c_i}{c_i}\right\}$ as
  \begin{equation}
	  \bar{m}_N(z) = \frac1N \tr \left(\sum_{i=1}^K \bar{e}_i(z) \R_i-z\I_N \right)^{-1}
  \end{equation}
  where $(z\mapsto \bar{e}_1(z),\dots,z\mapsto\bar{e}_K(z))\in\mathcal H(\CC\setminus\RR_+,\CC)^K$ are defined, for $z\in D$, as the unique solution of the following system of equations
  \begin{align}
	  \label{eq:pi}
	  \bar{e}_i(z) &= \frac1N\tr \P_i\left(e_i(z) \P_i + [\bar{c}_i-e_i(z) \bar{e}_i(z)] \I_{n_i} \right)^{-1} \nonumber \\
    	  e_i(z) &= \frac1N\tr \R_i\left(\sum_{j=1}^K \bar{e}_j(z) \R_j -z\I_N \right)^{-1}
  \end{align}
such that, for $z<0$, $0\leq \bar{e}_i(z) < c_i\bar{c}_i/e_i(z)$, for all $i$, explicitly given by:
    \begin{equation*} 
	  \quad  \bar{e}_i(z)=\lim_{t\to\infty}\bar{e}_i^{(t)}(z), \quad e_i(z)=\lim_{t\to\infty}e_i^{(t)}(z), \quad \bar{e}_i^{(t)}(z) = \lim_{k\to \infty} \bar{e}_i^{(t,k)}(z)
  \end{equation*}
 where, for $k\geq 1$,
\begin{align*}
	e_i^{(t)}(z) & = \frac1N\tr \R_i\left(\sum_{j=1}^K \bar{e}^{(t-1)}_j(z) \R_j -z\I_N \right)^{-1}\\
	\bar{e}_i^{(t,k)}(z) &= \frac1N\tr \P_i\left(e_i^{(t)}(z) \P_i + [\bar{c}_i-e_i^{(t)}(z) \bar{e}^{(t,k-1)}_i(z)] \I_{n_i} \right)^{-1}
\end{align*} 
with the initial values $\bar{e}_i^{(t,0)}(z)=0$ and $e_i^{(0)}(z)=0$ for all $i$. Moreover, $(z\mapsto {e}_1(z),\dots,z\mapsto{e}_K(z))\in\mathcal S(\RR_+)^K$.

\end{theorem}

\bigskip
\begin{remark}
Denoting $a_i({\sigma^2})=e_i(-{\sigma^2})$ for ${\sigma^2}>0$, we see immediately that Theorem~\ref{th:sumRWT} encompasses Theorem~\ref{th:fundamental_eq} as a special case.
\end{remark}
\bigskip

We first provide an outline of the proof for better understanding. The full proof will be given in Appendix \ref{sec:completeproof}.

\subsection{Sketch of the proof}
\label{sec:sketchproof}
As a first step, we wish to prove that there exists a matrix $\F$ of the form $\F=\sum_{i=1}^K\bar{f}_i\R_i$, with $\bar{f}_i\in\CC$, such that, for all nonnegative $\A$ with $\Vert \A\Vert<\infty$ uniformly on $N$ and $z<0$,
	\begin{equation*}
	  \frac1N\tr \A\left(\B_N-z\I_N\right)^{-1}-\frac1N\tr \A\left(\F-z\I_N\right)^{-1}\asto 0.
	\end{equation*}
	Taking $\A=\R_i$ and denoting $f_i\triangleq \frac1N\tr \R_i\left(\B_N-z\I_N\right)^{-1}$, we will have in particular that
	\begin{equation*}
	  f_i - \frac1N\tr\R_i\left(\sum_{j=1}^K\bar{f}_j\R_j-z\I_N\right)^{-1} \asto 0.
      \end{equation*}

Contrary to classical deterministic equivalent approaches for random matrices with i.i.d. entries, finding the approximation $\frac1N\tr \A\left(\F-z\I_N\right)^{-1}$ for $\frac1N\tr \A\left(\B_N-z\I_N\right)^{-1}$ is not straightforward. The reason is that, during the derivation, terms such as $\frac1{N_i-n_i}\tr \left(\I_{N_i} - \W_i\W_i^\herm\right)\H_i^\herm\left(\B_N-z\I_N\right)^{-1}\H_i$ with the $\left(\I_{N_i} - \W_i\W_i^\herm\right)$ prefix will naturally appear which need to be controlled. We proceed as follows.
	\begin{itemize}
	  \item We first denote, for all $i$, $\delta_i \triangleq \frac1{N_i-n_i}\tr \left(\I_{N_i} - \W_i\W_i^\herm\right)\H_i^\herm\left(\B_N-z\I_N\right)^{-1}\H_i$ some auxiliary variable. Then we prove
	    \begin{equation*}
	      f_i-\frac1N\tr \R_i\left(\G-z\I_N\right)^{-1}\asto 0,
	    \end{equation*}
	    with $\G=\sum_{j=1}^K\bar{g}_j\R_j$ and
	    \begin{equation*}
	      \bar{g}_i = \frac1{(1 - c_i)\bar{c}_i + \frac1{N}\sum_{l=1}^{n_i}\frac1{1+p_{il}\delta_i}}\frac1{N}\sum_{l=1}^{n_i}\frac{p_{il}}{1+p_{il}\delta_i},
	    \end{equation*}
	    where $p_{il}$ denotes the $l$th eigenvalue of $\P_i$, and $\delta_i$ is linked to $f_i$ through
	    \begin{equation*}
	      f_i - \left((1-c_i)\bar{c}_i\delta_i + \frac1N\sum_{l=1}^{n_i}\frac{\delta_i}{1+p_{il}\delta_i} \right) \asto 0.
	    \end{equation*}
	  \item This expression of $\bar{g}_i$, which is not convenient under this form, is then shown to satisfy
	    \begin{equation*}
	      \bar{g}_i - \frac1N\sum_{l=1}^{n_i}\frac{p_{il}}{\bar{c}_i+p_{il}f_i-f_i\bar{g}_i}=\bar{g}_i -\frac1N\tr \P_i\left(f_i\P_i + [\bar{c}_i-f_i\bar{g}_i]\I_{n_i} \right)^{-1} \asto 0,
	    \end{equation*}
which induces the $2K$-equation system
\begin{align*}
  f_i-\frac1N\tr \R_i\left(\sum_{j=1}^K\bar{g}_j\R_j -z\I_N\right)^{-1}&\asto 0 \\
  \bar{g}_i - \frac1N\tr \P_i\left(\bar{g}_i\P_i + [\bar{c}_i-f_i\bar{g}_i]\I_{n_i} \right)^{-1} &\asto 0.
\end{align*}
\item These relations are sufficient to infer the deterministic equivalent, but will be made more attractive for further considerations by introducing $\F=\sum_{i=1}^K\bar{f}_i\R_i$, and proving that
\begin{align*}
  f_i-\frac1N\tr \R_i\left(\sum_{j=1}^K\bar{f}_j\R_j -z\I_N\right)^{-1}&\asto 0 \\
  \bar{f}_i - \frac1N\tr \P_i\left(\bar{f}_i\P_i + [\bar{c}_i-f_i\bar{f}_i]\I_{n_i} \right)^{-1} &= 0,
\end{align*}
where, for $z<0$, $\bar{f}_i$ lies in $[0,c_i\bar{c}_i/f_i)$ and is now uniquely determined by $f_i$. In order to establish this convergence, it is necessary to define  an analytic extension of $\bar{f}_i$ in a neighborhood of $\RR_-$. The function $f_i$ can be immediately extended to $\CC\setminus\RR_+$ where it verifies the properties of a Stieltjes transform of a  finite measure supported by $\RR_+$.
\end{itemize}

	This is the very technical part of the proof. We then prove in a second step the existence and uniqueness of a solution to the fixed-point equation
\begin{align*}
  e_i-\frac1N\tr \R_i\left(\sum_{j=1}^K\bar{e}_j\R_j -z\I_N\right)^{-1}&= 0 \\
  \bar{e}_i - \frac1N\tr \P_i\left(\bar{e}_i\P_i + [\bar{c}_i-e_i\bar{e}_i]\I_{n_i} \right)^{-1} &= 0,
\end{align*}
for all finite $N$, $z<0$ and for $\bar{e}_i\in[0,c_i\bar{c}_i/e_i)$. This unfolds from a property of so-called {\it standard functions}. We will show precisely that the vector application $\h=(h_1,\ldots,h_K)$ defined for $z<0$ by
\begin{equation*}
h_i:(x_1,\ldots,x_K)\mapsto \frac1N\tr \R_i\left(\sum_{j=1}^K\bar{x}_j\R_j -z\I_N\right)^{-1}
\end{equation*}
with $\bar{x}_i$ the unique solution to
\begin{equation*}
  \bar{x}_i = \frac1N\tr \P_i\left(\bar{x}_i\P_i + [\bar{c}_i-x_i\bar{x}_i]\I_{n_i} \right)^{-1}
\end{equation*}
lying in $[0,c_i\bar{c}_i/x_i)$, is a standard function. It will unfold, from \cite[Theorem 2]{YAT95}, that the fixed-point equation in $(e_1,\ldots,e_K)$ has a unique solution with positive entries and that this solution can be determined as the limiting iteration of a classical fixed point algorithm. We will further establish that the $e_k(z)$ are Stieltjes transforms of finite measures supported by $\RR_+$ which satisfy the fundamental equations for $z\in D$.

The last step proves that the unique solution $(e_1,\ldots,e_N)$ is such that
\begin{equation*}
e_i - f_i \asto 0,
\end{equation*}
which is solved by standard arguments.  This will entail immediately by classical complex analysis arguments that $m_N(z)-\bar{m}_N(z)\asto 0$ for all $z\in\CC\setminus\RR_+$, form which the almost sure convergence $F_N-\bar{F}_N\Rightarrow 0$ unfolds.

\subsection{Complete proof}
\label{sec:completeproof}
We remind that, as $N$ grows, the ratios $c_i=\frac{n_i}{N_i}$ for $i=\{1,\dots,K\}$ satisfy 
\begin{equation*}
\lim\sup_N c_i < 1.
\end{equation*}
We also assume for the time being that for all $i$, $\Vert\R_i\Vert$ is uniformly bounded.
The case where $\Vert\R_i\Vert$ is uniformly bounded only in the almost sure sense will be treated subsequently.

\bigskip
{\bf Step 1: Convergence}

In this section, we take $z<0$, until further notice. Let us first introduce the following parameters. We will denote $P=\max_i\{\lim\sup\Vert \P_i\Vert\}$, $R=\max_i\{\lim\sup\Vert \R_i\Vert\}$, $c_+=\max_i\{\lim\sup c_i\}$, $\bar{c}_-=\min_i\{\lim\inf \bar{c}_i\}$ and $\bar{c}_+=\max_i\{\lim\sup \bar{c}_i\}$. 

Let $\A\in\CC^{N\times N}$ be a Hermitian nonnegative definite matrix, satisfying $\|\Am\|\le A<\infty$. Recall the definition $\R_i=\H_i\H_i^\herm$. Taking $\G=\sum_{j=1}^K\bar{g}_j \R_j$, with $\bar{g}_1,\ldots,\bar{g}_K$ scalars left undefined for the moment, we have
\begin{align}
    &\frac1N\tr\A(\B_N-z\I_N)^{-1} - \frac1N\tr\A(\G-z\I_N)^{-1} \nonumber \\
    &\overset{(a)}{=} \frac1N \tr \left[ \A(\B_N-z\I_N)^{-1}\sum_{i=1}^K \H_i \left(-\W_i\P_i\W_i^\herm + \bar{g}_i\I_{N_i}\right)\H_i^\herm (\G-z\I_N)^{-1} \right] \nonumber \\
    &\overset{(b)}{=} \sum_{i=1}^K \bar{g}_i \frac1N \tr \A(\B_N-z\I_N)^{-1} \R_i(\G-z\I_N)^{-1} - \frac1N\sum_{i=1}^K\sum_{l=1}^{n_i}p_{il}\w_{il}^\herm\H_i^\herm(\G-z\I_N)^{-1}\A(\B_N-z\I_N)^{-1}\H_i\w_{il} \nonumber \\
    \label{eq:eqb}    &\overset{(c)}{=} \sum_{i=1}^K \bar{g}_i \frac1N \tr \A(\B_N-z\I_N)^{-1} \R_i(\G-z\I_N)^{-1} - \frac1N\sum_{i=1}^K\sum_{l=1}^{n_i} \frac{p_{il}\w_{il}^\herm\H_i^\herm(\G-z\I_N)^{-1}\A(\B_{(i,l)}-z\I_N)^{-1}\H_i\w_{il}}{1+p_{il}\w_{il}^\herm\H_i^\herm(\B_{(i,l)}-z\I_N)^{-1}\H_i\w_{il}}
\end{align}
with $\w_{il}\in\CC^{N_i}$ the $l${th} column of $\W_i$, $p_{i1},\ldots,p_{in_i}$ the eigenvalues of $\P_i$ and $\B_{(i,l)}=\B_N-p_{il}\H_i\w_{il}\w_{il}^\herm\H_i^\herm$. The equality $(a)$ follows from Lemma~\ref{le:res_id}, $(b)$ follows from the decomposition $\W_i\P_i\W_i^\herm=\sum_{l=1}^{n_i}p_{il}\w_{il}\w_{il}^\herm$, while the equality $(c)$ follows from Lemma~\ref{le:sil_matrix_inversion}.

The idea now is to infer the values of the $\bar{g}_i$ such that the differences in \eqref{eq:eqb} go to zero almost surely as $N$ grows large. We will therefore proceed by studying the quantities $\w_{il}^\herm\H_i^\herm(\B_{(i,l)}-z\I_N)^{-1}\H_i\w_{il}$ and $\w_{il}^\herm\H_i^\herm(\G-z\I_N)^{-1}\A(\B_{(i,l)}-z\I_N)^{-1}\H_i\w_{il}$ in the denominator and numerator of the second term in \eqref{eq:eqb}.
	
For every $i\in\{1,\ldots,K\}$, denote 
\begin{equation}
	\label{eq:def_delta}
\delta_i \triangleq \frac1{N_i-n_i}\tr \left(\I_{N_i} - \W_i\W_i^\herm\right)\H_i^\herm\left(\B_N-z\I_N\right)^{-1}\H_i\ .
\end{equation}

Introducing the additional term $(\G-z\I_N)^{-1}\A$ in the argument of the trace in $\delta_i$, we denote
\begin{equation*}
\beta_i \triangleq \frac1{N_i-n_i}\tr \left(\I_{N_i} - \W_i\W_i^\herm\right)\H_i\left(\G-z\I_N\right)^{-1}\A\left(\B_N-z\I_N\right)^{-1}\H_i\ .
\end{equation*}

Under these notations, according to Lemma \ref{le:trace_Haar}, the quantity $\w_{il}^\herm\H_i^\herm(\B_{(i,l)}-z\I_N)^{-1}\H_i\w_{il}$ is asymptotically close to $\delta_i$, and, if $\G$ is independent of $\w_{il}$, the quantity $\w_{il}^\herm\H_i^\herm(\G-z\I_N)^{-1}\A(\B_{(i,l)}-z\I_N)^{-1}\H_i\w_{il}$ is asymptotically close to $\beta_i$.

We also define 
	\begin{equation}
		\label{eq:fi_def}
		f_i \triangleq \frac1N\tr \R_i\left(\B_N-z\I_N\right)^{-1} 
	\end{equation}
for $z\in\CC\setminus\RR_+$. Note that $f_i(z)\ge0$ for $z<0$.
	Remark first, from standard matrix inequalities and the fact that $\w^\herm \A\w\leq \Vert \A\Vert$ for any Hermitian matrix $\A$ and any unitary vector $\w$, that we have the following bounds on $\delta_i$, $\beta_i$ and $f_i$,
	\begin{align*}
	\quad \delta_i &\leq \frac{R}{|z|} ,\quad \beta_i \leq \frac{R A}{|z|^2},\quad f_i \leq \frac{R}{|z|}\ .
	\end{align*}

	From Lemma \ref{le:sil_matrix_inversion}, we have that
	\begin{align}
	  (1-c_i)\bar{c}_i \delta_i &= f_i - \frac1N \sum_{l=1}^{n_i} \w_{il}^\herm \H_i^\herm\left(\B_N-z\I_N\right)^{-1}\H_i \w_{il}\nonumber \\
		\label{eq:deltai_fi}		&= f_i - \frac1N \sum_{l=1}^{n_i} \frac{\w_{il}^\herm \H_i^\herm\left(\B_{(i,l)}-z\I_N\right)^{-1}\H_i \w_{il}}{1+p_{il}\w_{il}^\herm \H_i^\herm\left(\B_{(i,l)}-z\I_N\right)^{-1}\H_i \w_{il}}\ .
	\end{align}

	Since $z<0$, $\delta_i\geq 0$, and $\frac1{1+p_{il}\delta_i}$ is well defined. 
	By adding the term $\frac1N\sum_{l=1}^{n_i} \frac{\delta_i}{1+p_{il}\delta_i}$ on both sides, \eqref{eq:deltai_fi} can be re-written as
	\begin{align*}
	  &(1-c_i)\bar{c}_i \delta_i - f_i + \frac1N\sum_{l=1}^{n_i} \frac{\delta_i}{1+p_{il}\delta_i} \nonumber \\ &= \frac1N \sum_{l=1}^{n_i} \left[ \frac{\delta_i}{1+p_{il}\delta_i} - \frac{\w_{il}^\herm \H_i^\herm\left(\B_{(i,l)}-z\I_N\right)^{-1}\H_i \w_{il}}{1+p_{il}\w_{il}^\herm \H_i^\herm\left(\B_{(i,l)}-z\I_N\right)^{-1}\H_i \w_{il}} \right] \\ 
		&= \frac1N\sum_{l=1}^{n_i} \left[ \frac{\delta_i-\w_{il}^\herm \H_i^\herm\left(\B_{(i,l)}-z\I_N\right)^{-1}\H_i \w_{il} }{\left(1+p_{il}\delta_i\right)\left(1+p_{il}\w_{il}^\herm \H_i^\herm\left(\B_{(i,l)}-z\I_N\right)^{-1}\H_i \w_{il}\right)}\right]\ . 
	\end{align*}

	We now apply Lemma \ref{le:trace_Haar} and Lemma \ref{le:rank1perturbation}, which together with $\delta_i\leq R |z|^{-1}$ ensures that
\begin{equation}
	\label{eq:deltai_moment}
	\Exp \left[ \left|(1-c_i)\bar{c}_i \delta_i - f_i + \frac1N\sum_{l=1}^{n_i} \frac{\delta_i}{1+p_{il}\delta_i}\right|^4\right] \leq 8 \frac{C}{N^2}
\end{equation}
for some constant $C>0$. This determines the asymptotic behavior of $\delta_i$ and, thus, the asymptotic behavior of the quantity $\w_{il}^\herm\H_i^\herm(\B_{(i,l)}-z\I_N)^{-1}\H_i\w_{il}$ in the denominator of \eqref{eq:eqb}. 

We now proceed similarly with $\beta_i$ as with $\delta_i$. Assuming first that $\G$ is independent of $\w_{il}$, we obtain
\begin{align*}
	\beta_i &= \frac1{N_i-n_i}\tr \H_i^\herm\left(\G-z\I_N\right)^{-1}\A\left(\B_N-z\I_N\right)^{-1}\H_i \\ 
	&- \frac1{N_i-n_i}\sum_{l=1}^{n_i}\frac{\w_{il}^\herm \H_i^\herm\left(\G-z\I_N\right)^{-1}\A\left(\B_{(i,l)}-z\I_N\right)^{-1}\H_i \w_{il}}{1+p_{il}\w_{il}^\herm \H_i^\herm\left(\B_{(i,l)}-z\I_N\right)^{-1}\H_i \w_{il}}
\end{align*}
from which we have
\begin{align}
	&\frac1{N_i-n_i}\tr \H_i^\herm\left(\G-z\I_N\right)^{-1}\A\left(\B_N-z\I_N\right)^{-1}\H_i - \frac1{N_i-n_i}\sum_{l=1}^{n_i} \frac{\beta_i}{1+p_{il}\delta_i} - \beta_i  \nonumber \\ \label{eq:beta_intermediate}
	&= \frac1{N_i-n_i}\sum_{l=1}^{n_i}\left[\frac{\w_{il}^\herm \H_i^\herm\left(\G-z\I_N\right)^{-1}\A\left(\B_{(i,l)}-z\I_N\right)^{-1}\H_i \w_{il}}{1+p_{il}\w_{il}^\herm \H_i^\herm\left(\B_{(i,l)}-z\I_N\right)^{-1}\H_i \w_{il}} - \frac{\beta_i}{1+p_{il}\delta_i} \right]\ .
\end{align}

With the same inequalities as above, and with 
\begin{align*}
  \w_{il}^\herm \H_i^\herm\left(\G-z\I_N\right)^{-1}\A\left(\B_{(i,l)}-z\I_N\right)^{-1}\H_i \w_{il}&\leq \frac{R A}{|z|^2}
\end{align*}
we have that
\begin{align}
	\label{eq:beta_i_over1}
	\ &\Exp \left[\left| \frac{\w_{il}^\herm \H_i^\herm\left(\G-z\I_N\right)^{-1}\A\left(\B_{(i,l)}-z\I_N\right)^{-1}\H_i \w_{il}}{1+p_{il}\w_{il}^\herm \H_i^\herm\left(\B_{(i,l)}-z\I_N\right)^{-1}\H_i \w_{il}} - \frac{\beta_i}{1+p_{il}\delta_i}\right|^4 \right] \nonumber \\
	=\ &\Exp \left[\left| \frac{\w_{il}^\herm \H_i^\herm\left(\G-z\I_N\right)^{-1}\A\left(\B_{(i,l)}-z\I_N\right)^{-1}\H_i \w_{il} -\beta_i}{(1+p_{il}\delta_i)(1+p_{il}\w_{il}^\herm \H_i^\herm\left(\B_{(i,l)}-z\I_N\right)^{-1}\H_i \w_{il})}\right.\right. \nonumber \\
	&\qquad+\frac{p_{il}\delta_i\left[\w_{il}^\herm \H_i^\herm\left(\G-z\I_N\right)^{-1}\A\left(\B_{(i,l)}-z\I_N\right)^{-1}\H_i \w_{il}-\beta_i\right]}{(1+p_{il}\delta_i)(1+p_{il}\w_{il}^\herm \H_i^\herm\left(\B_{(i,l)}-z\I_N\right)^{-1}\H_i \w_{il})} \nonumber \\
	&\qquad\left.\left. +\frac{p_{il}\beta_i\left[\delta_i-\w_{il}^\herm \H_i^\herm\left(\B_{(i,l)}-z\I_N\right)^{-1}\H_i \w_{il}\right] }{(1+p_{il}\delta_i)(1+p_{il}\w_{il}^\herm \H_i^\herm\left(\B_{(i,l)}-z\I_N\right)^{-1}\H_i \w_{il})} \right|^4 \right] \nonumber \\
	\leq\ &8 \frac{C'}{N^2} \left(1+ \frac{P^4R^4}{|z|^4}\left(1+\frac{A^4}{|z|^4} \right)\right)
\end{align}
for some $C'>C$. Multiplying \eqref{eq:beta_intermediate} by $\frac{N_i-n_i}{N}$, we obtain 
\begin{align}
  \label{eq:fundbetai}
  &\Exp\left[\left| \frac1N\tr \H_i^\herm\left(\G-z\I_N\right)^{-1}\A\left(\B_N-z\I_N\right)^{-1}\H_i - \beta_i \left( (1-c_i)\bar{c}_i+\frac1N\sum_{l=1}^{n_i}\frac1{1+p_{il}\delta_i} \right) \right|^4\right] \nonumber \\ &\leq 8 \frac{C'}{N^2} \left(1+ \frac{P^4R^4}{|z|^4}\left(1+\frac{A^4}{|z|^4} \right)\right).
\end{align}
This provides the asymptotic behavior of $\beta_i$ or equivalently of the quantity $\w_{il}^\herm\H_i^\herm(\G-z\I_N)^{-1}\A(\B_{(i,l)}-z\I_N)^{-1}\H_i\w_{il}$ in the numerator of \eqref{eq:eqb}.

We are now in position to infer the $\bar{g}_i$ such that $\frac1N\tr\A(\B_N-z\I_N)^{-1} - \frac1N\tr\A(\G-z\I_N)^{-1}$ is asymptotically small. For the previous derivations to hold, the scalars $\bar{g}_k$, $k\in\{1,\ldots,K\}$, were assumed independent of $\w_{il}$. It is however easy to see that these derivations still hold true (up to the choice of larger constants $C$, $C'$) if $\bar{g}_k=\bar{g}_k^{(il)}+\varepsilon_{k,N}^{(il)}$ with $\bar{g}_k^{(il)}$ independent of $\w_{il}$ and $|\varepsilon_{k,N}^{(il)}|\leq C''/N$, for $C''$ constant independent of $k,i,j$. This follows from the fact that 
\begin{align*}
	\left\Vert \sum_{k=1}^K \bar{g}_k \R_k - \sum_{k=1}^K \bar{g}_k^{(il)} \R_k \right\Vert = \left\Vert \sum_{k=1}^K \varepsilon_{k,N}^{(il)} \R_k \right\Vert \leq \frac{KRC''}N.
\end{align*}

We choose
\begin{align}\label{eq:infer_gi}
	\bar{g}_k &= \frac1{(1-c_k)\bar{c}_k + \frac1{N}\sum_{m=1}^{n_k}\frac1{1+p_{km}\delta_k}}\frac1{N}\sum_{m=1}^{n_k}\frac{p_{km}}{1+p_{km}\delta_k}
\end{align}
and remark that $\bar{g}_k-\bar{g}_k^{(il)}=\Oc(1/N)$ with $\bar{g}_k^{(il)}$ defined similar to $\bar{g}_k$ \eqref{eq:infer_gi}, with column $\w_{il}$ removed from the expression of $\B_N$. Indeed, when $\w_{il}$ is removed, $p_{im}=0$ and $\delta_i=0$ are no longer defined, while the term $\delta_k^{(il)}$, $k\neq i$, defined equivalently as $\bar{g}_k^{(il)}$, satisfies $|\delta_k^{(il)}-\delta_k|\leq \frac1{N_k}\frac{1}{(1-c_k)|z|}$ from Lemma \ref{le:rank1perturbation}, from which the result unfolds.

Coming back to the original object of interest, we now have
\begin{align*}
    &\frac1N\tr\A(\B_N-z\I_N)^{-1} - \frac1N\tr\A(\G-z\I_N)^{-1} \\
    &= \sum_{i=1}^K\frac1{(1-c_i)\bar{c}_i + \frac1{N}\sum_{l=1}^{n_i}\frac1{1+p_{il}\delta_i}}\frac1{N}\sum_{l=1}^{n_i}\frac{p_{il}}{1+p_{il}\delta_i} \frac1N\tr \H_i^\herm\left(\G-z\I_N\right)^{-1}\A\left(\B_N-z\I_N\right)^{-1}\H_i \nonumber \\
    &\qquad- \frac1N \sum_{i=1}^K\sum_{l=1}^{n_i} \frac{p_{il}\w_{il}^\herm\H_i^\herm(\G-z\I_N)^{-1}\A(\B_{(i,l)}-z\I_N)^{-1}\H_i\w_{il}}{1+p_{il}\w_{il}^\herm\H_i^\herm(\B_{(i,l)}-z\I_N)^{-1}\H_i\w_{il}} \\
    &= \sum_{i=1}^K \frac1N\sum_{l=1}^{n_i} p_{il} \left[ \frac{\frac1N\tr \H_i^\herm\left(\G-z\I_N\right)^{-1}\A\left(\B_N-z\I_N\right)^{-1}\H_i}{((1-c_i)\bar{c}_i+\frac1N\sum_{l'=1}^{n_i}\frac1{1+p_{il'}\delta_i})(1+p_{il}\delta_i)} - \frac{\w_{il}^\herm\H_i^\herm(\G-z\I_N)^{-1}\A(\B_{(i,l)}-z\I_N)^{-1}\H_i\w_{il}}{1+p_{il}\w_{il}^\herm\H_i^\herm(\B_{(i,l)}-z\I_N)^{-1}\H_i\w_{il}} \right].
\end{align*}

Notice now that $1+p_{il}\delta_i\geq 1$ and
\begin{equation*}
  (1-c_i)\bar{c}_i < (1-c_i)\bar{c}_i+\frac1N\sum_{l=1}^{n_i}\frac1{1+p_{il}\delta_i} \leq \bar{c}_i
\end{equation*}
which ensure that we can divide the term in the expectation of the left-hand side of \eqref{eq:fundbetai} by $1+p_{il}\delta_i$ and $(1-c_i)\bar{c}_i+\frac1N\sum_{l=1}^{n_i}\frac1{1+p_{il}\delta_i}$ without taking the risk of the denominator getting close to $0$. This leads to
\begin{align}
	\label{eq:beta_i_over2}
	\Exp\left[\left| \frac{\beta_i}{1+p_{il}\delta_i} - \frac{\frac1N\tr \H_i^\herm\left(\G-z\I_N\right)^{-1}\A\left(\B_N-z\I_N\right)^{-1}\H_i}{\left((1-c_i)\bar{c}_i+\frac1N\sum_{l=1}^{n_i}\frac1{1+p_{il}\delta_i}\right)\left(1+p_{il}\delta_i\right)} \right|^4\right] \leq 8 \frac{C'}{N^2(1-c_i)^4\bar{c}_i^4} \left(1+ \frac{P^4R^4}{|z|^4}\left(1+\frac{A^4}{|z|^4} \right)\right).
\end{align}

From \eqref{eq:beta_i_over1} and \eqref{eq:beta_i_over2}, we therefore have 
\begin{align*}
	&\Exp\left[\left| \frac{\frac1N\tr \H_i^\herm\left(\G-z\I_N\right)^{-1}\A\left(\B_N-z\I_N\right)^{-1}\H_i}{\left((1-c_i)\bar{c}_i+\frac1N\sum_{l=1}^{n_i}\frac1{1+p_{il}\delta_i}\right)\left(1+p_{il}\delta_i\right)} - \frac{\w_{il}^\herm \H_i^\herm\left(\G-z\I_N\right)^{-1}\A\left(\B_{(i,l)}-z\I_N\right)^{-1}\H_i \w_{il}}{1+p_{il}\w_{il}^\herm \H_i^\herm\left(\B_{(i,l)}-z\I_N\right)^{-1}\H_i \w_{il}} \right|^4\right] \nonumber \\ 
	&\leq 128 \frac{C'}{N^2(1-c_i)^4\bar{c}_i^4} \left(1+ \frac{P^4R^4}{|z|^4}\left(1+\frac{A^4}{|z|^4} \right)\right).
\end{align*}

We finally obtain
\begin{align}
	\label{eq:convBNG}
	\Exp\left[\left|\frac1N\tr\A(\B_N-z\I_N)^{-1} - \frac1N\tr\A(\G-z\I_N)^{-1}\right|^4\right] \leq   128K^4 \frac{C'}{N^2(1-c_+)^4\bar{c}_-^4} \left(1+ \frac{P^4R^4}{|z|^4}\left(1+\frac{A^4}{|z|^4} \right)\right).
\end{align}

This provides a first convergence result as a function of the parameters $\delta_i$, from which a deterministic equivalent can be inferred. Nonetheless, the expression of $\bar{g}_i$ is rather impractical as it stands and we need to go further.

Observe in particular that $\bar{g}_i$ can be written under the form
\begin{equation*}
	\bar{g}_i = \frac1N\sum_{l=1}^{n_i}\frac{p_{il}}{((1-c_i)\bar{c}_i+\frac1N\sum_{l'=1}^{n_i}\frac1{1+p_{il'}\delta_i})+p_{il}\delta_i((1-c_i)\bar{c}_i+\frac1N\sum_{l'=1}^{n_i}\frac1{1+p_{il'}\delta_i})}\ .
\end{equation*}
We will study the denominator of the above expression and show that it can be simplified to a much more attractive form. 

From \eqref{eq:deltai_moment}, we first have
\begin{equation}
  \label{eq:fideltai}  \Exp \left[ \left| f_i - \delta_i\left( (1-c_i)\bar{c}_i+\frac1N\sum_{l=1}^{n_i} \frac1{1+p_{il}\delta_i}\right)\right|^4\right] \leq \frac{8C}{N^2}\ .
\end{equation}

Multiplying \eqref{eq:infer_gi} by $-\delta_i\left((1-c_i)\bar{c}_i+\frac1N\sum_{l=1}^{n_i} \frac1{1+p_{il}\delta_i} \right)$  and adding $\bar{c}_i$ to both sides yields
\begin{equation*}
  \bar{c}_i-\bar{g}_i\delta_i\left((1-c_i)\bar{c}_i+\frac1N\sum_{l=1}^{n_i} \frac1{1+p_{il}\delta_i} \right) = (1-c_i)\bar{c}_i+\frac1N\sum_{l=1}^{n_i} \frac1{1+p_{il}\delta_i}\ .
\end{equation*}
By definition, $\bar{g}_i\leq \frac{P}{(1-c_i)\bar{c}_i}$, and we therefore also have
\begin{align}
  \label{eq:1mgifi}
  \Exp\left[\left|\left(\bar{c}_i-f_i\bar{g}_i\right) - \left((1-c_i)\bar{c}_i+\frac1N\sum_{l=1}^{n_i}\frac1{1+p_{il}\delta_i} \right)\right|^4\right]\leq 8 \frac{C}{N^2} \frac{P^4}{(1-c_+)^4\bar{c}_-^4} \ .
\end{align}

The equations \eqref{eq:fideltai} and \eqref{eq:1mgifi} can now be used to approximate the denominator of $\bar{g}_i$ as follows
\begin{align}
	\label{eq:bargi_fp}
	&\Exp\left[\left|\bar{g}_i - \frac1N\sum_{l=1}^{n_i} \frac{p_{il}}{\bar{c}_i - f_i\bar{g}_i + p_{il} f_i}  \right|^4\right] \nonumber \\
	&=\Exp\left[\left|\frac1N\sum_{l=1}^{n_i}p_{il} \frac{p_{il}\left[f_i - \delta_i((1-c_i)\bar{c}_i+\frac1N\sum_{l'=1}^{n_i}\frac1{1+p_{il'}\delta_i})\right] + \left[\bar{c}_i-f_i\bar{g}_i - ((1-c_i)\bar{c}_i+\frac1N\sum_{l'=1}^{n_i}\frac1{1+p_{il'}\delta_i}) \right]}{\left[(1+p_{il}\delta_i)((1-c_i)\bar{c}_i+\frac1N\sum_{l'=1}^{n_i}\frac1{1+p_{il'}\delta_i}) \right]\left[\bar{c}_i-f_i\bar{g}_i+ p_{il}f_i \right]}\right|^4\right].
\end{align}

Before providing a useful bound, we need to ensure here that the term $\bar{c}_i-f_i\bar{g}_i+p_{il}f_i$ is uniformly away from zero, for all random $f_i$ and for all $N$. For this, we recall the bounds $0\leq f_i \leq \frac{R}{|z|}$ and $0\leq \bar{g}_i \leq \frac{P}{(1-c_i)\bar{c}_i}$.

Let us consider $0<\varepsilon<1$ and take from  now on $z<-\frac{RP}{(1-c_+)\bar{c}_-(\bar{c}_--\varepsilon)}$, so that $\bar{c}_i-f_i\bar{g}_i>\varepsilon$ for all $i$. From \eqref{eq:fideltai}, \eqref{eq:1mgifi} and \eqref{eq:bargi_fp}, we have 
\begin{align*}
  \Exp\left[\left|\bar{g}_i - \frac1N\sum_{l=1}^{n_i} \frac{p_{il}}{\bar{c}_i - f_i\bar{g}_i+p_{il} f_i} \right|^4\right] \leq 64 \frac{C}{N^2} \frac{P^8}{(1-c_i)^4\bar{c}_i^4\varepsilon^4} \left(1+\frac{1}{(1-c_i)^4\bar{c}_i^4}\right)
\end{align*}
which is of order $\Oc(1/N^2)$.

We are now ready to introduce the matrix $\F$. Consider 
\begin{equation*}
	\F = \sum_{i=1}^K \bar{f}_i \R_i, 
\end{equation*}
with $\bar{f}_i$ defined as the unique solution to the equation in $x$
\begin{equation*}
  x = \frac1N \sum_{l=1}^{n_i}\frac{p_{il}}{\bar{c}_i - f_i x+ f_i p_{il} }
\end{equation*}
within the interval $0\leq x < c_i\bar{c}_i/f_i$. To prove the uniqueness of the solution within this interval, note simply that
\begin{align*}
  \frac{c_i\bar{c}_i}{f_i} &> \frac1N \sum_{l=1}^{n_i}\frac{p_{il}}{\bar{c}_i - f_i (c_i\bar{c}_i/f_i)+ f_i p_{il} } \\
  0   	&\leq \frac1N \sum_{l=1}^{n_i}\frac{p_{il}}{\bar{c}_i -f_i\cdot 0 + f_i p_{il} }
\end{align*}
and that the function $x\mapsto \frac1N \sum_{l=1}^{n_i}\frac{p_{il}}{\bar{c}_i - f_i x+ f_i p_{il} }$ is continuously increasing on $x\in [0,c_i\bar{c}_i/f_i)$. Hence the uniqueness of the solution in $[0,c_i\bar{c}_i/f_i)$. We also show that this solution is an attractor of the fixed-point algorithm, when correctly initialized. Indeed, let $x_0,x_1,\ldots$ be defined by 
\begin{equation*}
  x_{n+1}=\frac1N \sum_{l=1}^{n_i}\frac{p_{il}}{\bar{c}_i - f_i x_n+ f_i p_{il} },
\end{equation*}
with $x_0\in[0,c_i\bar{c}_i/f_i)$. Then, $x_n\in [0,c_i\bar{c}_i/f_i)$ implies $\bar{c}_i - f_i x_n+ f_i p_{il}> (1-c_i)\bar{c}_i+f_ip_{il}\geq f_i p_{il}$ and therefore $f_ix_{n+1}< c_i\bar{c}_i$, so $x_0,x_1,\ldots$ is contained in $[0,c_i\bar{c}_i/f_i)$. Now observe that
\begin{equation*}
  x_{n+1}-x_n =\frac1N\sum_{l=1}^{n_i}\frac{p_{il}f_i(x_n-x_{n-1})}{(\bar{c}_i+p_{il}f_i-f_ix_n)(\bar{c}_i+p_{il}f_i-f_ix_{n-1})} 
\end{equation*}
with all terms being nonnegative in the sum, so that the differences $x_{n+1}-x_n$ and $x_n-x_{n-1}$ have the same sign.
The sequence $x_0,x_1,\ldots$ is therefore monotonic and bounded: it converges. Calling $x_\infty$ this limit, we have 
\begin{equation*}
  x_\infty = \frac1N\sum_{l=1}^{n_i}\frac{p_{il}}{\bar{c}_i+p_{il}f_i-f_ix_\infty}
\end{equation*}
as required.

To be able to finally prove that $\frac1N\tr\A(\B_N-z\I_N)^{-1}-\frac1N\tr\A(\F-z\I_N)^{-1}\asto 0$, we want now to show that $\bar{g}_i-\bar{f}_i$ tends to zero at a sufficiently fast rate. For this, we write
\begin{align}
	&\Exp\left[\left|\bar{g}_i-\bar{f}_i\right|^4\right] \nonumber \\
	&\leq 8 \left( \Exp\left[\left| \bar{g}_i - \frac1N\sum_{l=1}^{n_i} \frac{p_{il}}{\bar{c}_i - f_i\bar{g}_i+p_{il} f_i} \right|^4\right] + \Exp\left[\left|\frac1N\sum_{l=1}^{n_i} \frac{p_{il}}{\bar{c}_i - f_i\bar{g}_i+p_{il} f_i} - \frac1N\sum_{l=1}^{n_i} \frac{p_{il}}{\bar{c}_i - f_i\bar{f}_i+p_{il} f_i} \right|^4\right] \right) \nonumber \\
	\label{eq:bargibarfi}	&= 8 \left( \Exp\left[\left| \bar{g}_i - \frac1N\sum_{l=1}^{n_i} \frac{p_{il}}{\bar{c}_i - f_i\bar{g}_i+p_{il} f_i} \right|^4\right] + \Exp\left[\left|\bar{g}_i-\bar{f}_i\right|^4\left|\frac1N\sum_{l=1}^{n_i} \frac{p_{il}f_i }{(\bar{c}_i - f_i\bar{f}_i+p_{il} f_i)(\bar{c}_i - f_i\bar{g}_i+p_{il} f_i)} \right|^4\right] \right)
\end{align}
where we have simply written $\bar{g}_i-\bar{f}_i=(\bar{g}_i - \frac1N\sum_{l=1}^{n_i} \frac{p_{il}}{\bar{c}_i - f_i\bar{g}_i+p_{il} f_i})+(\frac1N\sum_{l=1}^{n_i} \frac{p_{il}}{\bar{c}_i - f_i\bar{g}_i+p_{il} f_i} -\bar{f}_i)$ and used the triangular inequality on the fourth power of each term.

We only need to ensure now that the coefficient multiplying $\left|\bar{g}_i-\bar{f}_i\right|$ in the right-hand side term is uniformly smaller than $1$. For this, observe that, as $z\to -\infty$, $|p_{il}f_i|\leq \frac{P R}{|z|}\to 0$ in the numerator. In the denominator, we already know that $\bar{c}_i-f_i\bar{f}_i+p_{il}f_i\geq (1-c_i)\bar{c}_i$ and we also have that $\bar{c}_i-f_i\bar{g}_i+p_{il}f_i\geq \bar{c}_i-\frac{RP}{(1-c_i)|z|}$, which is greater than some $\eta>0$ for $|z|$ taken large. 

Take $0<\eta<1$ and choose $z$ to be such that, for all $i$,
\begin{align*}
  \left|\frac1N\sum_{l=1}^{n_i} \frac{p_{il}f_i }{(\bar{c}_i - f_i\bar{f}_i+p_{il} f_i)(\bar{c}_i - f_i\bar{g}_i+p_{il} f_i)} \right| &\leq \frac{P R}{|z|(1-c_i)\bar{c}_i\eta}< \frac{1-\eta}8
\end{align*}
That is, from now on, we take $z<\min\left(-\frac{8PR}{\eta(1-\eta)(1-c_+)\bar{c}_-},-\frac{RP}{(1-c_+)\bar{c}_-(1-\varepsilon)}\right)$.

From the inequality \eqref{eq:bargibarfi}, gathering the terms in $\Exp\left[\left|\bar{g}_i-\bar{f}_i\right|^4\right]$ on the left side, we finally have
\begin{align}
  \label{eq:gibfib}
  \Exp\left[\left|\bar{g}_i - \bar{f}_i\right|^4\right] \leq \frac{512}{\eta^4} \frac{C}{N^2} \frac{P^8}{(1-c_i)^4\bar{c}_i^4\varepsilon^4} \left(1+\frac{1}{(1-c_i)^4\bar{c}_i^4}\right)\ .
\end{align}

We can now proceed to prove the deterministic equivalent relations:
\begin{align*}
	&\frac1N\tr \A\left(\G-z\I_N\right)^{-1} - \frac1N\tr\A\left(\F-z\I_N\right)^{-1} \nonumber \\
	&= \sum_{i=1}^K \frac1N\sum_{l=1}^{n_i} p_{il} \left[ \frac{\frac1N\tr \H_i^\herm\A\left(\G-z\I_N\right)^{-1}\left(\F-z\I_N\right)^{-1}\H_i}{((1-c_i)\bar{c}_i+\frac1N\sum_{l'=1}^{n_i}\frac1{1+p_{i,l'}\delta_i})(1+p_{il}\delta_i)} - \frac{\frac1N\tr \H_i^\herm\A\left(\G-z\I_N\right)^{-1}\left(\F-z\I_N\right)^{-1}\H_i}{\bar{c}_i-f_i\bar{f}_i + p_{il}f_i} \right] \\
	&= \sum_{i=1}^K \frac1N\sum_{l=1}^{n_i} p_{il} \left[\left( \frac{\frac1N\tr \H_i^\herm\A\left(\G-z\I_N\right)^{-1}\left(\F-z\I_N\right)^{-1}\H_i}{((1-c_i)\bar{c}_i+\frac1N\sum_{l'=1}^{n_i}\frac1{1+p_{i,l'}\delta_i})(1+p_{il}\delta_i)} - \frac{\frac1N\tr \H_i^\herm\A\left(\G-z\I_N\right)^{-1}\left(\F-z\I_N\right)^{-1}\H_i}{\bar{c}_i-f_i\bar{g}_i + p_{il}f_i}\right) \nonumber \right. \\
	&\qquad+ \left. \left(\frac{\frac1N\tr \H_i^\herm\A\left(\G-z\I_N\right)^{-1}\left(\F-z\I_N\right)^{-1}\H_i}{\bar{c}_i-f_i\bar{g}_i + p_{il}f_i} - \frac{\frac1N\tr \H_i^\herm\left(\G-z\I_N\right)^{-1}\left(\F-z\I_N\right)^{-1}\H_i}{\bar{c}_i-f_i\bar{f}_i + p_{il}f_i}\right) \right] \\
	&=\sum_{i=1}^K \frac1N\tr \H_i^\herm\A\left(\G-z\I_N\right)^{-1}\left(\F-z\I_N\right)^{-1}\H_i \frac1N\sum_{l=1}^{n_i} p_{il} \left[  \frac{f_i(\bar{g}_i-\bar{f}_i)}{(\bar{c}_i-f_i\bar{f}_i + p_{il}f_i)(\bar{c}_i-f_i\bar{g}_i + p_{il}f_i)} \right. \\
	&\qquad \left. + \frac{\left((\bar{c}_i-f_i\bar{g}_i)-((1-c_i)\bar{c}_i+\frac1N\sum_{l'=1}^{n_i}\frac1{1+p_{i,l'}\delta_i}) \right) + p_{il}\left(f_i - \delta_i((1-c_i)\bar{c}_i+\frac1N\sum_{l'=1}^{n_i}\frac1{1+p_{i,l'}\delta_i}) \right)}{((1-c_i)\bar{c}_i+\frac1N\sum_{l'=1}^{n_i}\frac1{1+p_{i,l'}\delta_i})(1+p_{il}\delta_i)(\bar{c}_i-f_i\bar{f}_i + p_{il}f_i)}\right]
\end{align*}
Therefore, from \eqref{eq:fideltai}, \eqref{eq:1mgifi} and \eqref{eq:gibfib}, 
\begin{align*}
  \Exp\left[\left|\frac1N\tr \A\left(\G-z\I_N\right)^{-1} - \frac1N\tr\A\left(\F-z\I_N\right)^{-1}\right|^4\right] \leq \frac{64R^4P^4A^4K}{|z|^8(1-c_+)^8\bar{c}_-^8}\frac{C}{N^2}\left(1+\frac{1}{(1-c_+)^4\bar{c}_-^4} \right)^4\left[1+\frac{64 R^4P^4}{|z|^4\eta^4\varepsilon^4} \right] 
\end{align*}
which is of order $\Oc(1/N^2)$.

Together with \eqref{eq:convBNG}, applying the Markov inequality \cite[(5.31)]{BIL08} and the Borel Cantelli lemma \cite[Theorem 4.3]{BIL08}, we finally have
\begin{equation}
  \label{eq:convBNF}
	\frac1N\tr\A\left(\B_N-z\I_N\right)^{-1} - \frac1N\tr\A\left(\F-z\I_N\right)^{-1} \asto 0,
\end{equation}
as $N$ grows large for realizations of $\{\W_1,\ldots,\W_K\}$ taken from a set $A_z\subset \Omega$ of probability one (we use $\Omega$ here to denote the sample space of the probability space generating the sequences of matrices $\{\W_1,\ldots,\W_K\}$ of growing sizes). This therefore holds true for countably many $z$ (smaller than the established bound) with a cluster point in $\RR_-$, on a set $A\subset \Omega$ of probability one. 

Before we can extend the convergence to the entire negative real axis, we need to define an analytic extension of $\bar{f}_i$ in a neighborhood of $\RR_-$. Take $D=\left\{z=x+{\bf{i}}y: x<0, |y|\le |x|\frac{1-c_i}{c_i}\right\}$. For $z\in D$, the following holds
\begin{align}
 \Re\{f_i\}\ge 0 \qquad \text{and}\qquad |\Im\{f_i\}|\le \Re\{f_i\}\frac{1-c_i}{c_i}.
\end{align}
To see this, consider $\Bm_N=\Um\Dm\Um\htp$ the eigenvalue decomposition of $\Bm_N$, where $\Um=[\uv_1\dots\uv_N]\in\CC^{N\times N}$ is unitary and $\Dm=\diag(d_1,\dots,d_N)$ contains the nonnegative eigenvalues of $\Bm_N$. Denoting $z=x+{\bf{i}}y$, we have
\begin{align}
 f_i &=\frac1N\trace\Rm_i\LB \Bm_N -z\Id_N\RB^{-1}\nonumber\\
&= \frac1N\trace\Rm_i\LB \Bm_N -z\Id_N\RB^{-1} \LB \Bm_N -z^*\Id_N\RB \LB \Bm_N -z^*\Id_N\RB^{-1} \nonumber\\
&= \frac1N\sum_{j=1}^N \frac{d_j-x}{|d_j-z|^2}\uv_j\htp\Rm_i\uv_j + {\bf i} y \frac1N\sum_{j=1}^N \frac{1}{|d_j-z|^2}\uv_j\htp\Rm_i\uv_j.
\end{align}
From the last equation, it follows that $x<0$ and $|y|\le |x|\frac{1-c_i}{c_i}$ imply $\Re\{f_i\}\ge 0$ and $|\Im\{f_i\}|\le\Re\{ f_i\}\frac{1-c_i}{c_i}$.

Consider now the sequence $\{q_{i,n}\}_{n\ge 0}$ of complex numbers, recursively defined as
\begin{align}\label{eq:analytic_ext}
 q_{i,{n}} = \frac{f_i}{(1-c_i)\bar{c}_i + \frac1N\sum_{l=1}^{n_i}\frac{1}{1+p_{il}q_{i,{n-1}}}},\quad n\ge 1
\end{align}
and $q_{i,0}=0$. 
We will now show that $|q_{i,n}|\le \frac{|f_i|}{(1-c_i)\bar{c}_i}$ for all $n$ and $z\in D$. First, notice that
\begin{align}\label{eq:extbound1}
 |q_{i,n}| \le \frac{|f_i|}{(1-c_i)\bar{c}_i}
\end{align}
whenever $\Re\{q_{i,{n-1}}\}\ge 0$. After some simple algebra, one arrives at
\begin{align}\label{eq:extcond1}
 \Re\{q_{i,{n}}\} & = \frac{\Re\{f_i\}\LSB(1-c_i)\bar{c}_i + \frac1N\sum_{l=1}^{n_i} \frac{1}{|1+p_{il}q_{i,{n-1}}|^2}\RSB + \frac1N\sum_{l=1}^{n_i}\frac{p_{il}\LB \Re\{f_i\}\Re\{q_{i,{n-1}}\} - \Im\{f_i\}\Im\{q_{i,{n-1}}\}\RB}{|1+p_{il}q_{i,{n-1}}|^2}}{\left|(1-c_i)\bar{c}_i + \frac1N\sum_{l=1}^{n_i}\frac{1}{1+p_{il}q_{i,{n-1}}} \right|^2}\\\label{eq:extcond2}
\Im\{q_{i,{n}}\} &= \frac{\Im\{f_i\}\LSB(1-c_i)\bar{c}_i + \frac1N\sum_{l=1}^{n_i} \frac{1}{|1+p_{il}q_{i,{n-1}}|^2}\RSB + \frac1N\sum_{l=1}^{n_i}\frac{p_{il}\LB \Im\{f_i\}\Re\{q_{i,{n-1}}\} + \Re\{f_i\}\Im\{q_{i,{n-1}}\}\RB}{|1+p_{il}q_{i,{n-1}}|^2}}{\left|(1-c_i)\bar{c}_i + \frac1N\sum_{l=1}^{n_i}\frac{1}{1+p_{il}q_{i,{n-1}}} \right|^2}.
\end{align}

Now, if we assume $\Re\{q_{i,n-1}\}\ge 0$, we have
\begin{align}\nonumber
 \Re\{q_{i,{n}}\} & \ge \frac{\Re\{f_i\}(1-c_i)\bar{c}_i -|\Im\{f_i\}| \frac1N\sum_{l=1}^{n_i}\frac{p_{il}|\Im\{q_{i,{n-1}}\}|}{|1+p_{il}q_{i,{n-1}}|^2}}{\left|(1-c_i)\bar{c}_i + \frac1N\sum_{l=1}^{n_i}\frac{1}{1+p_{il}q_{i,{n-1}}} \right|^2}\\
&\ge \frac{\Re\{f_i\}(1-c_i)\bar{c}_i -|\Im\{f_i\}|c_i\bar{c}_i }{\left|(1-c_i)\bar{c}_i + \frac1N\sum_{l=1}^{n_i}\frac{1}{1+p_{il}q_{i,{n-1}}} \right|^2}.
\end{align}
The right-hand side of the last equations is nonnegative whenever $$ |\Im\{f_i\}|\le \Re\{f_i\}\frac{1-c_i}{ci} .$$ As this condition is always satisfied for $z\in D$ and we have defined $q_{i,0}=0$, we can conclude that \eqref{eq:extbound1} and $\Re\{q_{i,n}\}\ge 0$ hold for all $n$.

Additionally, we have from \eqref{eq:extcond1} and \eqref{eq:extcond2} that
\begin{align}\label{eq:useful_bound}
 \Re\{f_i\}\Re\{q_{i,{n-1}}\} + \Im\{f_i\}\Im\{q_{i,{n-1}}\} = \frac{\LB \Re\{f_i\}^2 + \Im\{f_i\}^2\RB  \LSB(1-c_i)\bar{c}_i + \frac1N\sum_{l=1}^{n_i} \frac{1 +p_{il}\Re\{q_{n-2}\}}{|1+p_{il}q_{n-2}|^2}\RSB }{\left|(1-c_i)\bar{c}_i + \frac1N\sum_{l=1}^{n_i}\frac{1}{1+q_{il}x_{n-2}} \right|^2}\ge 0.
\end{align}

Until here, we have proved that $\{q_{i,n}\}$ is a sequence of bounded analytic functions on $z\in D$ (the analyticity follows from the fact that $f_i$ is analytic on $\CC\setminus\RR_+$ and $q_{i,n}$ is a rational function with no pole in $D$). Let us now focus on the negative real axis, i.e., $z<0$, which lies in the interior of $D$. Here, the following holds
\begin{align}
 q_{i,{n+1}}-q_{i,{n}} = \LB q_{i,n} - q_{i,{n-1}}\RB f_i \frac{\frac1N\sum_{l=1}^{n_i} \frac{1}{\LB 1+p_{il}q_{i,n}\RB\LB 1+p_{il}q_{i,{n-1}}\RB}}{\LSB(1-c_i)\bar{c}_i + \frac1N\sum_{l=1}^{n_i}\frac{1}{1+p_{il}q_{i,{n}}}\RSB\LSB(1-c_i)\bar{c}_i + \frac1N\sum_{l=1}^{n_i}\frac{1}{1+p_{il}q_{i,{n-1}}}\RSB}.
\end{align}
As $f_i$ and all terms in the fraction of the right-hand side of the last equation are nonnegative, the differences $ q_{i,{n+1}}-q_{i,{n}}$ and $q_{i,n} - q_{i,{n-1}}$ have the same sign. Thus, $\{q_{i,n}\}$ is either monotonically increasing or decreasing. Since $\{q_{i,n}\}$ is also bounded, it must converge. This implies by Vitali's convergence theorem that $\{q_{i,n}\}$ converges uniformly on all closed subsets of $D$ and that this limit is an analytic function. Call this limit $q_i = \lim_{n} q_{i,n}$.

We now define $\tilde{f}_{i,n}$ by the quantities $f_i$ and $q_{i,n}$:
\begin{align}\label{eq:ftilde}
	\tilde{f}_{i,n} = \frac1{f_i}\frac1N\sum_{l=1}^{n_i} \frac{p_{il}q_{i,{n}}}{1+p_{il}q_{i,{n}}}.
\end{align}
Clearly, $\{\tilde{f}_{i,n}\}$ is a sequence of analytic bounded functions, converging for $z\in D$ to
\begin{align*}
	\tilde{f}_{i} \triangleq \frac1{f_i}\frac1N\sum_{l=1}^{n_i} \frac{p_{il}q_i}{1+p_{il}q_i}.
\end{align*}
With the above definition, $q_{i,{n+1}}$ satisfies
\begin{align}
	q_{i,{n+1}} &= \frac{f_i}{(1-c_i)\bar{c}_i + \frac1N\sum_{l=1}^{n_i}\frac1{1+p_{il}q_{i,n}}}\nonumber \\
&= \frac{f_i}{\bar{c}_i - \frac1N\sum_{l=1}^{n_i}\frac{p_{il}q_{i,n}}{1+p_{il}q_{i,n}}}\nonumber \\
&= \frac{f_i}{\bar{c}_i - f_i\tilde{f}_{i,n}}.
\end{align}
Thus, we can write, from \eqref{eq:ftilde},
\begin{align}
	\label{eq:algo_fin}
	\tilde{f}_{i,n+1} &= \frac1N\sum_{l=1}^{n_i} \frac{p_{il}}{(\bar{c}_i-f_i\tilde{f}_{i,n})\LB 1 + \frac{p_{il}f_i}{\bar{c}_i-f_i\tilde{f}_{i,n}}\RB} = \frac1N\sum_{l=1}^{n_i} \frac{p_{il}}{\bar{c}_i-f_i\tilde{f}_{i,n} + p_{il}f_i}.
\end{align}

As a consequence, the restriction of $\tilde{f}_i$ to $z<0$ is identical to $\bar{f}_i$ and the fixed-point algorithm defined by \eqref{eq:algo_fin} with $\tilde{f}_{i,0}=0$ converges to $\tilde{f}_i$ for $z\in D$. From this point on, we therefore extend the definition of $\bar{f}_i$ to $D$ by $\bar{f}_i(z)=\tilde{f}_i(z)$.

From \eqref{eq:useful_bound} and for $z\in D$, we have
\begin{align*}
 \Re\{\bar{f}_i\} = \frac1N\sum_{l=1}^{n_i} p_{il}\frac{p_{il}|q_i|^2\Re\{f_i\} + \Re\{f_i\}\Re\{q_i\} + \Im\{f_i\}\Im\{q_i\}}{\left|f_i + p_{il}q_if_i \right|^2}\ge 0.
\end{align*}
Since $\Fm=\sum_{k=1}^K\bar{f}_k\Rm_k$ and the matrices $\Rm_k$ are Hermitian nonnegative definite, it follows that $\left|\frac1N\trace\Am\LB\Fm-z\Id_N\RB^{-1}\right|\le \frac{||\Am||}{|x|}$ for $z\in D$.

From the Vitali convergence theorem, the identity theorem, the analyticity of the functions under study, and the fact that $\frac1N\tr\A\left(\B_N-z\I_N\right)^{-1}$ and $\frac1N\tr\A\left(\F-z\I_N\right)^{-1}$ are uniformly bounded on all closed subsets of $z\in D$, we have that the convergence
\begin{align*}
	\frac1N\tr\A\left(\B_N-z\I_N\right)^{-1} - \frac1N\tr \A\left(\sum_{i=1}^K \bar{f}_i \R_i - z\I_N \right)^{-1}\asto 0
\end{align*}
holds true for all $z\in D$. 

Applying the result for $\A = \R_j$, this is in particular 
\begin{equation}
	\label{eq:convfj}
	f_j - \frac1N\tr \R_j\left(\sum_{i=1}^K \bar{f}_i \R_i - z\I_N \right)^{-1}\asto 0
\end{equation}
for $z\in D$, where $\bar{f}_i$ is defined as the above limit. 
For $\A = \I_N$, this implies
\begin{equation*}
	m_N(z) - \frac1N\tr \left(\sum_{i=1}^K \bar{f}_i \R_i - z\I_N \right)^{-1}\asto 0
\end{equation*}
which finally proves the convergence.

\bigskip
{\bf Step 2: Existence and Uniqueness}

We will now prove the existence and the uniqueness of positive solutions $e_1(z),\ldots,e_K(z)$ for $z<0$ and the convergence of the classical fixed point algorithm to these values. In addition, we will show that the $e_i(z)$ have analytic extensions on $\CC\setminus\RR_+$ which are Stieltjes transforms of finite measures over $\RR_+$ and satisfy the fundamental equations for $z\in D$. We first introduce some notations and useful identities. Until stated otherwise, we assume $z<0$. Note that, similar to the auxiliary variables $\delta_i$ and $q_i$ in Step 1, we can define, for any pair of variables $x_i$ and $\bar{x}_i$, with $\bar{x}_i$ defined as the solution $y$ to $y=\frac1N\sum_{l=1}^{n_i}\frac{p_{il}}{\bar{c}_i-x_iy+x_ip_{il}}$ such that $0\leq y< c_j\bar{c}_i/x_i$, the auxiliary variables $\Delta_1,\ldots,\Delta_K$, with the properties
\begin{align*}
  x_i &= \Delta_i\left((1-c_i)\bar{c}_i+\frac1N\sum_{l=1}^{n_i}\frac1{1+p_{il}\Delta_i} \right) \\ 
  &=\Delta_i\left(\bar{c}_i-\frac1N\sum_{l=1}^{n_i}\frac{p_{il}\Delta_i}{1+p_{il}\Delta_i} \right)
\end{align*}
and
\begin{align}\nonumber
	\bar{c}_i-x_i\bar{x}_i &= (1-c_i)\bar{c}_i+\frac1N\sum_{l=1}^{n_i}\frac1{1+p_{il}\Delta_i} \\ \label{eq:useful_equality}
	&= \bar{c}_i-\frac1N\sum_{l=1}^{n_i}\frac{p_{il}\Delta_i}{1+p_{il}\Delta_i}\ .
\end{align}

First note that mapping between $x_i$ and $\Delta_i$ is unique. This unfolds from noticing, with some abuse of notation,
\begin{equation*}
 \frac{d\,x_i}{d\Delta_i} =\frac{d}{d\Delta_i} \left[\Delta_i\left((1-c_i)\bar{c}_i+\frac1N\sum_{l=1}^{n_i}\frac1{1+p_{il}\Delta_i} \right)\right] = (1-c_i)\bar{c}_i+\frac1N\sum_{l=1}^{n_i}\frac1{(1+p_{il}\Delta_i)^2} > 0
\end{equation*}
and therefore $x_i$ and $\Delta_i$ are one-to-one. Additionally, $x_i$ is a strictly increasing function of $\Delta_i$ with $\Delta_i=0$ for $x_i=0$. This ensures that $\Delta_i>0$ if and only if $x_i>0$.

Secondly, from the definition of $\bar{x}_i$, we have
\begin{align*}
  \bar{c}_i-x_i\bar{x}_i &= \bar{c}_i-x_i\frac1N\sum_{l=1}^{n_i}\frac{p_{il}}{(\bar{c}_i-x_i\bar{x}_i)+p_{il}x_i} \\
  &= \bar{c}_i- \Delta_i\left(\bar{c}_i-\frac1N\sum_{l=1}^{n_i}\frac{p_{il}\Delta_i}{1+p_{il}\Delta_i}\right)\frac1N\sum_{l=1}^{n_i}\frac{p_{il}}{\bar{c}_i-x_i\bar{x}_i+p_{il}\Delta_i\left(\bar{c}_i-\frac1N\sum_{l'=1}^{n_i}\frac{p_{il'}\Delta_i}{1+p_{il'}\Delta_i}\right)}\ .
\end{align*}

Note in particular that, since $x_i\bar{x}_i=\frac1N\sum_{l'=1}^{n_i}\frac{p_{il'}\Delta_i}{1+p_{il'}\Delta_i}$, the above equation simplifies to 
\begin{align*}
  &\bar{c}_i-\Delta_i\left(\bar{c}_i-\frac1N\sum_{l=1}^{n_i}\frac{p_{il}\Delta_i}{1+p_{il}\Delta_i}\right) \frac1N\sum_{l=1}^{n_1}\frac{p_{il}}{\left(\bar{c}_i-\frac1N\sum_{l'=1}^{n_i}\frac{p_{il'}\Delta_i}{1+p_{il'}\Delta_i}\right)+p_{il} \Delta_i\left(\bar{c}_i-\frac1N\sum_{l'=1}^{n_i}\frac{p_{il'}\Delta_i}{1+p_{il'}\Delta_i}\right)}  \\ 
  &=\bar{c}_i-\frac1N\sum_{l=1}^{n_i}\frac{p_{il}\Delta_i}{1+p_{il}\Delta_i}
\end{align*}
and therefore $\bar{c}_i-\frac1N\sum_{l=1}^{n_i}\frac{p_{il}\Delta_i}{1+p_{il}\Delta_i}$ is one of the solutions of the implicit equation in $u$,
\begin{equation*}
  u = \bar{c}_i-x_i\frac1N\sum_{l=1}^{n_i}\frac{p_{il}}{u+p_{il}x_i}\ .
\end{equation*}
Equivalently, writing $u=\bar{c}_i-x_iy$, it follows that $\frac1{x_i}\frac1N\sum_{l=1}^{n_i}\frac{p_{il}\Delta_i}{1+p_{il}\Delta_i}$ is one of the solutions of the equation in $y$
\begin{equation*}
  y = \frac1N\sum_{l=1}^{n_i}\frac{p_{il}}{\bar{c}_i-x_i y+p_{il}x_i}\ .
\end{equation*}

Since 
\begin{equation*}
  x_i\left(\frac1{x_i}\frac1N\sum_{l=1}^{n_i}\frac{p_{il}\Delta_i}{1+p_{il}\Delta_i}\right)< c_i\bar{c}_i
\end{equation*}
this solution lies in $[0,c_i\bar{c}_i/x_i)$ and is exactly equal to $\bar{x}_i$. This proves that the equations in $(x_i,\bar{x}_i)$ can be written under the form of the equations in $(\Delta_i,\bar{x}_i)$, as presented above.

We take the opportunity of the above definitions to notice that, for $x_i>x_i'$ and $\bar{x}_i'$, $\Delta_i'$ defined similarly as $\bar{x}_i$ and $\Delta_i$,
\begin{equation}
	\label{eq:propertyxibxi}
	x_i\bar{x}_i-x_i'\bar{x}_i' = \frac1N\sum_{l=1}^{n_i}\frac{p_{il}(\Delta_i-\Delta_i')}{(1+p_{il}\Delta_i)(1+p_{il}\Delta_i')}>0
\end{equation}
whenever $\P_i\neq 0$. Therefore $x_i\bar{x}_i$ is a growing function of $x_i$ (or equivalently of $\Delta_i$). This will turn out to be a useful remark later.

We are now in position to prove the step of uniqueness. Define, for $i\in\{1,\ldots,K\}$, the functions
\begin{equation*}
	h_i:(x_1,\ldots,x_K)\mapsto \frac1N\tr \R_i\left(\sum_{j=1}^K \bar{x}_j\R_j - z\I_N\right)^{-1}
\end{equation*}
with $\bar{x}_j$ the unique solution of the equation in $y$
\begin{equation}
	\label{eq:uniquesoly}
	y=\frac1N\sum_{l=1}^{n_j}\frac{p_{jl}}{\bar{c}_j+x_jp_{jl}-x_jy}
\end{equation}
such that $0\leq  \bar{x}_j < c_j\bar{c}_j/x_j$.

We will prove in the following that the multivariate function $\h=(h_1,\ldots,h_K)$ is a {\it standard function} (or {\it standard interference function}), defined in \cite{YAT95}, as follows:
\begin{definition}
	\label{def:standardfunctions}
A function $\h(x_1,\ldots,x_K)\in \RR^K$ is said to be standard if it fulfills the following conditions:
\begin{enumerate}
	\item {\it Positivity:} for each $j$, if $x_1,\ldots,x_K\geq 0$, then $h_j(x_1,\ldots,x_K)>0$.
	\item {\it Monotonicity:} if $x_1\geq x_1',\ldots,x_K\geq x_K'$, then for all $j$, $h_j(x_1,\ldots,x_K)\geq h_j(x_1',\ldots,x_K')$.
	\item {\it Scalability:} for all $\alpha>1$ and for all $j$, $\alpha h_j(x_1,\ldots,x_K)>h_j(\alpha x_1,\ldots,\alpha x_K)$.
\end{enumerate}
\end{definition}

The important result regarding standard functions, \cite[Theorem 2]{YAT95}, is given as follows:
\begin{theorem}
	\label{th:standardfunctions}
If a $K$-variate function $\h(x_1,\ldots,x_K)$ is standard and there exists $(x_1,\ldots,x_K)$ such that for all $j$, $x_j\geq h_j(x_1,\ldots,x_K)$, then the fixed-point algorithm that consists in setting
\begin{equation*}
	x_j^{(t+1)} = h_j(x_1^{(t)},\ldots,x_K^{(t)})
\end{equation*}
for $t\geq 1$ and for any initial values $x_1^{(0)},\ldots,x_K^{(0)}>0$ converges to the unique jointly positive solution of the system of $K$ equations
\begin{equation*}
x_j = h_j(x_1,\ldots,x_K)
\end{equation*}
with $j\in\{1,\ldots,K\}$.
\end{theorem}

In order to prove that there exist $x_1,\dots,x_K$ such that $x_j\ge h_j(x_1,\ldots,x_K)$ for all $j$, it is sufficient to notice that $h_j(x_1,\ldots,x_K)\le R/|z|$ for all $j$. Thus, for $x_j\ge R/|z|$ for all $j$, $x_j\ge h_j(x_1,\ldots,x_K)$ holds for all $j$.  Therefore, by showing that $\h\triangleq (h_1,\ldots,h_K)$ is a standard function, we will prove that the classical fixed point algorithm converges to the unique set of positive solutions $e_1,\ldots,e_K$, when $z<0$.

The positivity condition is straightforward as $\bar{x}_i$ is positive for $x_i$ positive and therefore $h_j(x_1,\ldots,x_K)$ is always positive whenever $x_1,\ldots,x_K$ are nonnegative.

The scalability is also rather direct. Let $\alpha>1$, then
\begin{align*}
	&\alpha h_j(x_1,\ldots,x_K) - h_j(\alpha x_1,\ldots,\alpha x_K) \nonumber \\ 
	&= \frac1N\tr \R_j\left(\sum_{k=1}^K\frac{\bar{x}_k}{\alpha}\R_k-\frac{z}{\alpha}\I_N\right)^{-1} - \frac1N\tr \R_j\left(\sum_{k=1}^K \bar{x}^{(\alpha)}_k\R_k- z\I_N\right)^{-1}
\end{align*}
where we denoted $\bar{x}^{(\alpha)}_j$ the unique solution to \eqref{eq:uniquesoly} within $[0,c_j\bar{c}_j/(\alpha x_j))$ with $x_j$ replaced by $\alpha x_j$.
From Lemma \ref{lem:traceinequ}, it suffices to show that
\begin{equation*}
	\sum_{k=1}^K\left[\bar{x}^{(\alpha)}_k - \frac{\bar{x}_k}{\alpha}\right]\R_k + \left[ z -\frac{z}{\alpha} \right]\I_N
\end{equation*}
is positive definite. Since $\alpha x_i> x_i$, we have from the property \eqref{eq:propertyxibxi} that
\begin{equation*}
	\alpha x_k\bar{x}^{(\alpha)}_k - x_k \bar{x}_k > 0
\end{equation*}
or equivalently
\begin{equation*}
	\bar{x}^{(\alpha)}_k-\frac{\bar{x}_k}{\alpha} > 0.
\end{equation*}
Along with $1-1/\alpha>0$ and $z<0$, this ensures that $\alpha h_j(x_1,\ldots,x_K) > h_j(\alpha x_1,\ldots,\alpha x_K)$.

The monotonicity requires some more calculus. This unfolds from considering $\bar{x}_i$ as a function of $\Delta_i$, by verifying that $\frac{d}{d\Delta_i}\bar{x}_i$ is negative.
\begin{align*}
  \frac{d}{d\Delta_i}\bar{x}_i &= \frac{1}{\Delta_i^2}\left(1-\frac{\bar{c}_i}{\bar{c}_i-\frac1N\sum_{l=1}^{n_i}\frac{p_{il}\Delta_i}{1+p_{il}\Delta_i}}\right)+ \frac{\bar{c}_i}{\Delta_i^2}\left(\frac{\frac1N\sum_{l=1}^{n_i}\frac{p_{il}\Delta_i}{(1+p_{il}\Delta_i)^2}}{\left(\bar{c}_i-\frac1N\sum_{l=1}^{n_i}\frac{p_{il}\Delta_i}{1+p_{il}\Delta_i}\right)^2}\right)\\
&= \frac{1}{\Delta_i^2\left(\bar{c}_i-\frac1N\sum_{l=1}^{n_i}\frac{p_{il}\Delta_i}{1+p_{il}\Delta_i}\right)^2} \left[-\frac1N\left(\sum_{l=1}^{n_i}\frac{p_{il}\Delta_i}{1+p_{il}\Delta_i}\right)\left(\bar{c}_i-\frac1N\sum_{l=1}^{n_i}\frac{p_{il}\Delta_i}{1+p_{il}\Delta_i}\right) +\frac{\bar{c}_i}N\sum_{l=1}^{n_i}\frac{p_{il}\Delta_i}{(1+p_{il}\Delta_i)^2}  \right]\\
&=\frac{1}{\Delta_i^2\left(\bar{c}_i-\frac1N\sum_{l=1}^{n_i}\frac{p_{il}\Delta_i}{1+p_{il}\Delta_i}\right)^2} \left[\left(\frac1N\sum_{l=1}^{n_i}\frac{p_{il}\Delta_i}{1+p_{il}\Delta_i}\right)^2 -\frac{\bar{c}_i}N\sum_{l=1}^{n_i}\frac{p_{il}\Delta_i}{1+p_{il}\Delta_i}+\frac{\bar{c}_i}N\sum_{l=1}^{n_i}\frac{p_{il}\Delta_i}{(1+p_{il}\Delta_i)^2} \right]\\
&= \frac{1}{\Delta_i^2\left(\bar{c}_i-\frac1N\sum_{l=1}^{n_i}\frac{p_{il}\Delta_i}{1+p_{il}\Delta_i}\right)^2} \left[\left(\frac1N\sum_{l=1}^{n_i}\frac{p_{il}\Delta_i}{1+p_{il}\Delta_i}\right)^2 -\frac{\bar{c}_i}N\sum_{l=1}^{n_i}\frac{(p_{il}\Delta_i)^2}{(1+p_{il}\Delta_i)^2} \right].
\end{align*}
From the Cauchy-Schwarz inequality, we have
\begin{align}
  \label{eq:CS}
  \left(\sum_{l=1}^{n_i}\frac1N\frac{p_{il}\Delta_i}{1+p_{il}\Delta_i}\right)^2\leq \sum_{k=1}^{n_i}\frac{1}{N^2}\sum_{l=1}^{n_i}\frac{(p_{il}\Delta_i)^2}{(1+p_{il}\Delta_i)^2}=c_i\bar{c}_i \frac1N\sum_{l=1}^{n_i}\frac{(p_{il}\Delta_i)^2}{(1+p_{il}\Delta_i)^2} <\frac{\bar{c}_i}N\sum_{l=1}^{n_i}\frac{(p_{il}\Delta_i)^2}{(1+p_{il}\Delta_i)^2}
\end{align}
which is sufficient to conclude that $\frac{d}{d\Delta_i}\bar{x}_i<0$.
Since $\Delta_i$ is an increasing function of $x_i$, we have that $\bar{x}_i$ is a decreasing function of $x_i$, i.e., $\frac{d}{dx_i}\bar{x}_i<0$. Therefore, for two sets $x_1,\ldots,x_K$ and $x_1',\ldots,x_K'$ of positive values such that $x_j>x_j'$, defining $\bar{x}_j'$ equivalently as $\bar{x}_j$ for the terms $x_j'$, we have $\bar{x}_k' > \bar{x}_k$. Therefore, from Lemma \ref{lem:traceinequ}, we finally have
\begin{align}
	h_j(x_1,\ldots,x_K) - h_j(x_1',\ldots,x_K') &= \frac1N\tr \R_j\left(\sum_{k=1}^K \bar{x}_k \R_k-z\I_N\right)^{-1} - \frac1N \tr \R_j\left(\sum_{k=1}^K \bar{x}'_k\R_k- z\I_N\right)^{-1} \label{eq:hj_inc} > 0.
\end{align}
This proves the monotonicity condition and, finally, that $\h=(h_1,\ldots,h_K)$ is a standard function.

It follows from Theorem \ref{th:standardfunctions} that $(e_1,\ldots,e_K)$ is uniquely defined and that the classical fixed-point algorithm converges to this solution from any initialization point (remember that, at each step of the algorithm, the set $\bar{e}_1,\ldots,\bar{e}_K$ must be evaluated, possibly thanks to a further fixed-point algorithm). 

We will now show that $e_i(z)$ has an analytic extension on $z\in\CC\setminus\RR_+$ which is the Stieltjes transform of a finite measure supported by $\RR_+$. For this proof, consider the matrices $\P_{[p],i}\in \CC^{n_ip}$ and $\H_{[p],i}\in\CC^{Np\times N_ip}$ for all $i$ defined as the Kronecker products $\P_{[p],i}\triangleq \P_i\otimes \I_p$, $\H_{[p],i}\triangleq \H_i \otimes \I_p$, such that $\P_{[p],i}$ and $\R_{[p],i}=\H_{[p],i}\H_{[p],i}^\herm$ have the same spectral distributions as the matrices $\Pm_i$ and $\Rm_i$, respectively. It is easy to see that the solutions of the implicit equations \eqref{eq:pi} for $z\in\CC\setminus\RR_+$ remain unchanged by substituting the $\P_{[p],i}$ and $\R_{[p],i}$ to the $\P_i$ and $\R_i$, respectively, for any $p$. Denoting similarly $f_{[p],i}$ the $f_i$ adapted to $\P_{[p],i}$ and $\H_{[p],i}$, from the convergence result of Step 1, we can choose $f_{[1],i},f_{[2],i},\ldots$ a sequence of the set of probability one where convergence is ensured as $p$ grows large ($N$ and the $n_i$ are kept fixed). 
Call $e'_i(z)$ the limit. 

We wish to prove that $e'_i$, seen as a function of $z$, is the {\it Stieltjes transform} of a distribution function, whose restriction to $\RR_-$ matches $e_i$. For this, we prove the defining properties of a Stieltjes transform, provided in Lemma \ref{le:properties_ST}. By Vitali's convergence theorem \cite{TIT39}, $e'_i$ is analytic on $\CC^+$ since $e'_i$ is the limit of a sequence of analytic functions, bounded on every compact of $\CC\setminus\RR_+$. It is clear that for $z\in\CC^+$, $\Im[f_{[p],i}(z)]>0$, $\Im[zf_{[p],i}(z)]>0$ and $|yf_{[p],i}({\bf i}y)|\leq R$ for $y>0$. This implies that for $z\in\CC^+$, $\Im[e'_i(z)]\geq 0$, $z\Im[e'_i(z)]\geq 0$ and $\lim_{y\to\infty} -{\bf i}ye_i'({\bf i}y)\leq R$. In addition, note that, for $z\in\CC^+$,
\begin{equation*}
  \Im[f_{[p],i}] \geq \frac1N\frac{r}{(RP+|z|)^2}\Im[z]>0
\end{equation*}
and 
\begin{equation*}
  \Im[zf_{[p],i}] \geq \frac1N\frac{Kr^2t}{(RP+|z|)^2}\Im[z]>0
\end{equation*}
with $r$ a lower bound on the smallest non-zero eigenvalues of $\R_1,\ldots,\R_K$ (we naturally assume all $\R_k$ non-zero) and $t$ a lower bound on the smallest non-zero eigenvalues of $\T_1,\ldots,\T_K$ (again, none assumed identically zero). Take $z\in\CC^+$ and $\varepsilon<\frac12\min(\frac1N\frac{r}{(RP+|z|)^2}\Im[z],\frac1N\frac{Kr^2t}{(RP+|z|)^2}\Im[z])$. There now exists $p_0$ such that $p\geq p_0$ implies $|\Im[f_{[\phi(p)],i}]-\Im[e'_i]|<\varepsilon/2$ and $|\Im[zf_{[\phi(p)],i}]-\Im[ze'_i]|<\varepsilon/2$, and therefore $\Im[e'_i]>\varepsilon/2$ and $z\Im[e'_i(z)]>\varepsilon/2$ so that $e'_i(z)$ is the Stieltjes transform of a finite measure on $\RR_+$. Moreover, since $e_i'(z)=\lim f_{[p],i}(z)$ on $D$, from \eqref{eq:convfj}, $e_i'(z)$ satisfies the equations \eqref{eq:pi} for all $z\in D$. 

Consider now two sets of Stieltjes transforms $(e'_1(z),\ldots,e'_K(z))$ and $(e''_1(z),\ldots,e''_K(z))$, $z\in\CC\setminus \RR_+$, which are solutions of the fixed-point equation for $z<0$. Since $e_i'(z)=e''_i(z)$ for all $z<0$, and $e_i'(z)-e_i''(z)$ is holomorphic on $\CC\setminus \RR_+$ as the difference of Stieltjes transforms, $e_i'(z)=e_i''(z)$ over $\CC\setminus \RR_+$ \cite{RUD86} by the identity theorem. This therefore proves, in addition to {\it point-wise} uniqueness on the negative half-line, the uniqueness of the {\it Stieltjes transform} solution of the functional implicit equation such that, for $z<0$, $0\leq\bar{e_i}<c_i\bar{c}_i/e_i$ for all $i$. Moreover, this solution satisfies the fundamental equations for $z\in D$.

\bigskip
{\bf Step 3: Convergence of $e_i-f_i$}

For this step, we follow the same approach as in \cite{HAC07}. Denote 
\begin{equation*}
	\varepsilon_{N,i} \triangleq f_i-\frac1N\tr\R_i\left(\sum_{k=1}^K \bar{f}_k\R_k -z\I_N \right)^{-1}
\end{equation*}
and recall the definitions of $f_i$,  $e_i$, $\bar{f}_i$ and $\bar{e}_i$:
\begin{align*}
 f_i &= \frac1N\tr\R_i\left(\B_N-z\I_N\right)^{-1}\\
 e_i &= \frac1N\tr\R_i\left(\sum_{j-1}^K\bar{e}_j\R_j-z\I_N\right)^{-1}\\
\bar{f}_i &= \frac1N\sum_{l=1}^{n_i}\frac{p_{il}}{\bar{c}_i-f_i\bar{f}_i + p_{il}f_i},\qquad \bar{f}_i \in [0,c_i\bar{c}_i/f_i)\\
\bar{e}_i &= \frac1N\sum_{l=1}^{n_i}\frac{p_{il}}{\bar{c}_i-e_i\bar{e}_i + p_{il}e_i},\qquad \bar{e}_i \in [0,c_i\bar{c}_i/e_i)\ .
\end{align*}

From the definitions above, we have the following set of inequalities
\begin{align}\label{eqn:inequ}
 f_i \leq \frac{R}{|z|},\quad e_i \leq \frac{R}{|z|},\quad \bar{f}_i \leq \frac{P}{(1-c_i)\bar{c}_i},\quad \bar{e}_i \le\frac{P}{(1-c_i)\bar{c}_i}\ .
\end{align}

We will show in the following that
\begin{equation}\label{eqn:convergence}
e_i - f_i \asto 0
\end{equation}
for all $i\in\{1,\ldots,N\}$. We start by considering the following differences
\begin{align*}
	f_i-e_i &= \sum_{j=1}^K(\bar{e}_j-\bar{f}_j)\frac1N\tr\R_i\left(\sum_{k=1}^K \bar{e}_k\R_k -z\I_N \right)^{-1}\R_j\left(\sum_{k=1}^K \bar{f}_k\R_k -z\I_N  \right)^{-1} + \varepsilon_{N,i} \\
	\bar{e}_i-\bar{f}_i &=  \frac1N\sum_{l=1}^{n_i}\frac{p_{il}^2(f_i-e_i) -p_{il}\left[f_i\bar{f}_i - e_i\bar{e}_i\right]}{(\bar{c}_i-\bar{e}_ie_i+p_{il}e_i)(\bar{c}_i-\bar{f}_if_i+p_{il}f_i)} \\
	f_i\bar{f}_i - e_i\bar{e}_i&=\bar{f}_i(f_i-e_i) + e_i(\bar{f}_i-\bar{e}_i)\ . 
\end{align*}
For notational convenience, we define the following values
\begin{align*}
 \alpha &\triangleq \sup_i\Exp\left[|f_i-e_i|^4\right]\\
\bar{\alpha} &\triangleq \sup_i\Exp\left[|\bar{f}_i-\bar{e}_i|^4\right]\ .
\end{align*}
It is thus sufficient to show that $\alpha$ is summable in order to prove \eqref{eqn:convergence}. 
By applying \eqref{eqn:inequ} to the absolute of the first difference, we obtain
\begin{align*}
 |f_i-e_i| \leq \frac{KR^2}{|z|^2} \sup_i|\bar{f}_i-\bar{e}_i| +\sup_i |\varepsilon_{N,i}|
\end{align*}
and hence 
\begin{align}\label{eqn:alpha}
 \alpha \le& \frac{8K^4R^8}{|z|^8}\bar{\alpha} + \frac{8C}{N^2}
\end{align}
for some $C>0$ such that $\Exp[\sup_i|\varepsilon_{N,i}|^4]\leq 8K\sup_i\Exp[|\varepsilon_{N,i}|^4]\leq C/N^2$.
Similarly, we have for the third difference
\begin{align*}
 |f_i\bar{f}_i - e_i\bar{e}_i|&\leq |\bar{f}_i||f_i-e_i|+ |e_i||\bar{f}_i-\bar{e}_i|\\
 &\leq \frac{P}{(1-c_+)\bar{c}_-}\sup_i|f_i-e_i| + \frac{R}{|z|}\sup_i|\bar{f}_i-\bar{e}_i|\ .
\end{align*}
This result can be used to upperbound the second difference term, which writes
\begin{align*}
  |\bar{f}_i-\bar{e}_i|&\leq \frac{1}{(1-c_+)^2\bar{c}_-^2}\left(P^2\sup_i|f_i-e_i| +P|f_i\bar{f}_i - e_i\bar{e}_i|\right)\\
  &\leq \frac{1}{(1-c_+)^2\bar{c}_-^2}\left(P^2\sup_i|f_i-e_i| +P\left[\frac{P}{(1-c_+)\bar{c}_-}\sup_i|f_i-e_i| + \frac{R}{|z|}\sup_i|\bar{f}_i-\bar{e}_i| \right]\right)\\
  &\leq \frac{P^2(\bar{c}_-+1)}{(1-c_+)^3\bar{c}_-^3}\sup_i|f_i-e_i| + \frac{RP}{|z|(1-c_+)^2\bar{c}_-^2}\sup_i|\bar{f}_i-\bar{e}_i|\ .
\end{align*}
Hence
\begin{align}\label{eqn:alphabar}
  \bar{\alpha} \leq \frac{8P^8(\bar{c}_-+1)^4}{(1-c_+)^{12}\bar{c}_-^{12}}\alpha + \frac{8R^4P^4}{|z|^4(1-c_+)^8\bar{c}_-^8}\bar{\alpha}\ .
\end{align}
For any $z$ satisfying $|z|>\frac{2RP}{(1-c_+)^2}$, we have $\frac{8R^4P^4}{|z|^4(1-c_+)^8}<1/2$ and thus 
\begin{align*}
  \bar{\alpha} < \frac{16P^8(\bar{c}_-+1)^4}{(1-c_+)^{12}\bar{c}_-^{12}}\alpha\ .
\end{align*}
Plugging this result into \eqref{eqn:alpha} yields
\begin{align*}
 \alpha \leq \frac{128K^4R^8P^8(2-c)^4}{|z|^8(1-c_+)^{12}}\alpha + \frac{8C}{N^2}\ .
\end{align*}
Take $0<\varepsilon<1$. It is easy to check that for $|z|>\frac{128^{1/8}RP\sqrt{K(\bar{c}_-+1)}}{(1-c_+)^{3/2}\bar{c}_-^{3/2}(1-\varepsilon)^{1/8}}$, $\frac{128K^4R^8P^8(\bar{c}_-+1)^4}{|z|^8(1-c_+)^{12}\bar{c}_-^{12}}<1-\varepsilon$ and thus
\begin{align}
 \alpha < \frac{8C}{\varepsilon N^2}\ .
\end{align}
Since $C$ does not depend on $N$, $\alpha$ is clearly summable which, along with Markov inequality and the Borel Cantelli lemma, concludes the proof.

Finally, taking the same steps as previously, we also have
\begin{equation*}
  \Exp\left[ \left|m_N(z) - \bar{m}_N(z)\right|^4\right] \leq \frac{8C}{\varepsilon N^2} 
\end{equation*}
for some $|z|$ large enough. For these $z$, the same conclusion holds: $m_N(z) - \bar{m}_N(z)\asto 0$. From Vitali convergence theorem and the identity theorem, since $f_i$ and $e_i$ are uniformly bounded on all closed sets of $\CC\setminus\RR_+$ and analytic, we finally have that the convergence is true for all $z\in\CC\setminus \RR_+$. The almost sure convergence of the Stieltjes transform implies the almost sure weak convergence of $F_N-\bar{F}_N$ to $0$, uniformly over every compact set of $\RR_+$, which is our final result.

This concludes the proof of Theorem~\ref{th:sumRWT} for surely bounded $\R_i$.

\bigskip
\subsubsection{Almost sure boundedness of $\Vert\R_i\Vert$}
\

To extend Theorem \ref{th:sumRWT} to the case where $\Vert\R_i\Vert$ is only almost surely bounded, we merely apply the Tonelli theorem (Lemma~\ref{le:tonelli}). Call $(\Omega_R,\mathcal F_R,P_R)$ the probability space that generates the sequences of matrices of growing sizes $\{\R_i,1\leq i\leq K,N_i\in\NN\}$, $(\Omega_W,\mathcal F_W,P_W)$ the probability space that generates the sequences of matrices of growing sizes $\{\W_i,1\leq i\leq K,N_i\in\NN\}$, and $(\Omega_R\times \Omega_W,\mathcal F_R\times \mathcal F_W,Q)$ their product space. Denote $A$ the subspace of $\mathcal F_R\times \mathcal F_W$ for which $F_N-\bar{F}_N\to 0$. Then, from Tonelli theorem, Lemma \ref{le:tonelli},
\begin{equation*}
 Q(A)=\int_{\Omega_R\times \Omega_W}1_A(r,w)Q(d(r,w))=\int_{\Omega_R}\int_{\Omega_W}1_A(r,w)P_W(dw) P_R(dr).
\end{equation*}
Take $r$ such that the $\Vert\R_i\Vert$ are all uniformly bounded with growing $N$. Then, from Theorem \ref{th:sumRWT}, for this $r$, $\int_{\Omega_W}1_A(r,w)P_W(dw)=1$. But these $r\in\Omega_R$ belong to a space of probability one, as the intersection of $K$ spaces of probability one, and finally $Q(A)=1$.

\section{Proof of Theorem \ref{th:mutual_info}}
\label{app:mutual_info}

Following for simplicity the notations of Appendix \ref{app:sumRXT}, we use here the variable $e_i(-\sigma^2)$ in place of $a_i(\sigma^2)$.
It is easy to see (e.g. \cite[Definition 3.2]{COUbook}) that, for $F$ a probability distribution function with support in $\RR_+$
\begin{equation*}
	\int_0^\infty \log\left(1+\frac{t}{x}\right)dF(t) = \int_x^\infty \left(-\frac1t+m_F(-t)\right)dF(t)
\end{equation*}
where $m_F(z)$ is the Stieltjes transform of $F$ (this is sometimes called the Shannon-transform in $1/x$). In particular,
\begin{equation*}
	I^{(a)}_N({\sigma^2}) = \frac1N\log\det\left(\I_N + \frac1{{\sigma^2}}\B_N \right) = \int_{{\sigma^2}}^\infty \left(-\frac1t+m_N(-t)\right)dF_N(t).
\end{equation*}

We will first show that the expression $\bar{I}^{(a)}_N({\sigma^2})$ given in Theorem \ref{th:mutual_info} satisfies the same property with $\bar{F}_N$.

For notational simplicity, we will write $e_i=e_i(-{\sigma^2})$ and $\bar{e}_i=\bar{e}_i(-{\sigma^2})$.

First note that the system of equations \eqref{eq:pi} is unchanged if we extend the $\P_i$ matrices into $N_i\times N_i$ diagonal matrices filled with $N_i-n_i$ zero eigenvalues. Therefore, we can assume that all $\P_i$ have size $N_i\times N_i$ although we restrict the measure of eigenvalues of $\P_i$ to have a mass $1-c_i$ in zero. Since this does not alter the equations \eqref{eq:pi}, we have in particular $\bar{e}_i<\bar{c}_i/e_i$ for ${\sigma^2}>0$. 

This being said, $\bar{I}^{(a)}_N$ is given by
  \begin{align*}
	  \bar{I}^{(a)}_N({\sigma^2}) &= \frac1N\log\det\left(\I_N + \frac1{{\sigma^2}}\sum_{i=1}^K \bar{e}_i\R_i\right) + \sum_{i=1}^K \left[\frac1N\log\det\left([\bar{c}_i-e_i\bar{e}_i]\I_N + e_i\P_i \right)-\bar{c}_i\log(\bar{c}_i) \right].
  \end{align*} 

  Calling $\bar{I}$ the function
  \begin{align*}
	  &\bar{I}: (x_1,\ldots,x_K,\bar{x}_1,\ldots,\bar{x}_K,{\sigma^2}) \\ 
	  &\mapsto \frac1N\log\det\left(\I_N + \frac1{{\sigma^2}}\sum_{i=1}^K \bar{x}_i\R_i\right) + \sum_{i=1}^K \left[\frac1N\log\det\left([\bar{c}_i-x_i\bar{x}_i]\I_N + x_i\P_i \right) - \bar{c}_i\log(\bar{c}_i) \right],
  \end{align*}
  we have
  \begin{align*}
	  \frac{\partial \bar{I}}{\partial x_i}(e_1,\ldots,e_K,\bar{e}_1,\ldots,\bar{e}_K,{\sigma^2}) &= \bar{e}_i -\bar{e}_i\frac1N \sum_{l=1}^{N_i}\frac1{\bar{c}_i-e_i\bar{e}_i + e_ip_{il}} \\
	  \frac{\partial \bar{I}}{\partial \bar{x}_i}(e_1,\ldots,e_K,\bar{e}_1,\ldots,\bar{e}_K,{\sigma^2}) &= e_i - e_i\frac1N \sum_{l=1}^{N_i}\frac1{\bar{c}_i-e_i\bar{e}_i + e_ip_{il}}.
  \end{align*}
In order to proceed, note that we can write $\bar{c}_i$ in the following way:
  \begin{align*}
    \bar{c}_i &= \frac1N\sum_{l=1}^{N_i} \frac{\bar{c}_i-e_i\bar{e}_i + e_ip_{il}}{\bar{c}_i-e_i\bar{e}_i + e_ip_{il}} \\ 
     &= (\bar{c}_i-e_i\bar{e}_i)\frac1N \sum_{l=1}^{N_i}\frac1{\bar{c}_i-e_i\bar{e}_i + e_ip_{il}} + \frac1N\sum_{l=1}^{N_i}\frac{e_ip_{il}}{\bar{c}_i-e_i\bar{e}_i + e_ip_{il}} \\ 
     &= (\bar{c}_i-e_i\bar{e}_i)\frac1N \sum_{l=1}^{N_i}\frac1{\bar{c}_i-e_i\bar{e}_i + e_ip_{il}}+e_i\bar{e}_i
  \end{align*}
from which it follows that
\begin{equation*}
  \left(\bar{c}_i - e_i\bar{e}_i\right) \left(1 - \frac1N \sum_{l=1}^{N_i}\frac1{\bar{c}_i-e_i\bar{e}_i + e_ip_{il}} \right) = 0.
\end{equation*}
But we also know that $0\leq \bar{e}_i<\bar{c}_i/e_i$ and therefore $\bar{c}_i - e_i\bar{e}_i>0$. This entails 
\begin{equation}
  \label{eq:sumeq1}
	\frac1N \sum_{l=1}^{N_i}\frac1{\bar{c}_i-e_i\bar{e}_i + e_ip_{il}} = 1.
\end{equation}
From \eqref{eq:sumeq1}, we can then conclude
\begin{align*}
	\frac{\partial \bar{I}}{\partial x_i}(e_1,\ldots,e_K,\bar{e}_1,\ldots,\bar{e}_K,{\sigma^2}) &= 0 \\
	\frac{\partial \bar{I}}{\partial \bar{x}_i}(e_1,\ldots,e_K,\bar{e}_1,\ldots,\bar{e}_K,{\sigma^2}) &= 0.
\end{align*}

We therefore have, from the differentiation chain rule, 
\begin{align*}
	\frac{d}{d{\sigma^2}}\bar{I}^{(a)}_N({\sigma^2}) &= \sum_{i=1}^K \left[\frac{\partial \bar{I}}{\partial e_i}\frac{\partial e_i}{\partial {\sigma^2}} + \frac{\partial \bar{I}}{\partial \bar{e}_i}\frac{\partial \bar{e}_i}{\partial {\sigma^2}}\right] + \frac{\partial \bar{I}}{\partial {\sigma^2}} \\
	&= \frac{\partial \bar{I}}{\partial {\sigma^2}} \\
	&= -\frac1{{\sigma^4}}\sum_{i=1}^K \bar{e}_i \frac1N\tr \R_i\left(\I_N + \frac1{{\sigma^2}}\sum_{j=1}^K\bar{e}_j\R_j\right)^{-1} \\
	&= -\frac1{{\sigma^2}} \frac1N\tr\left[\left(\sum_{i=1}^K \frac1{{\sigma^2}} \bar{e}_i\R_i +\I_N - \I_N \right)\left(\I_N + \frac1{{\sigma^2}}\sum_{j=1}^K\bar{e}_j\R_j\right)^{-1}\right]\\
	&= -\frac1{{\sigma^2}} + \frac1N\tr\left({\sigma^2} \I_N + \sum_{j=1}^K\bar{e}_j\R_j\right)^{-1}
\end{align*}

Recognizing the Stieltjes transform of $\bar{F}_N$, we therefore have, along with the fact that $\bar{I}^{(a)}_N(\infty)=0$,
\begin{equation*}
	\bar{I}^{(a)}_N({\sigma^2}) = \int_{{\sigma^2}}^\infty \left( \frac1t - \frac1{t^2}\bar{m}_N\left(-\frac1t\right)\right) dt
\end{equation*}
and therefore
\begin{equation*}
	\bar{I}^{(a)}_N({\sigma^2}) = \int_0^\infty \log\left(1+\frac{t}{{\sigma^2}}\right)d\bar{F}_N(t).
\end{equation*}

In order to prove the almost sure convergence $I^{(a)}_N({\sigma^2})-\bar{I}^{(a)}_N({\sigma^2})\asto 0$, we simply need to remark that the support of the eigenvalues of $\B_N$ is bounded. Indeed, the non-zero eigenvalues of $\W_i\W_i^\herm$ have unit modulus and therefore $\Vert \B_N\Vert \leq KPR$. Similarly, the support of $\bar{F}_N$ is the support of the eigenvalues of $\sum_{i=1}^K\bar{e}_i\R_i$, which are bounded by $KPR$ as well.

As a consequence, for $\B_1,\B_2,\ldots$ a realization for which $F_N-\bar{F}_N\Rightarrow 0$ (these lie in a space of probability one), we have, from the dominated convergence theorem
\begin{equation*}
	\int_0^\infty \log\left(1 + \frac{t}{{\sigma^2}} \right)d[F_N-\bar{F}_N](t)\to 0 
\end{equation*}
Hence the almost sure convergence of the instantaneous mutual information. 

Because of sure boundedness of $\Vert \B_N\Vert$, an immediate application of the dominated convergence theorem on the probability space $\Omega$ that engenders the sequences of matrices $\B_1(\omega),\B_2(\omega),\ldots$, $\omega\in\Omega$, entails convergence in the first mean as well.

\section{Proof of Theorem \ref{th:SINR_MMSE}}
\label{app:SINR_MMSE}

In this section, we follow closely the derivations of Appendix \ref{app:sumRXT} and use the variable $e_i(-\sigma^2)$ in place of $a_i(\sigma^2)$.
To prove Theorem \ref{th:SINR_MMSE}, we will pursue a similar approach as for the proof of Theorem \ref{th:sumRWT}, but we can now take advantage of all results derived so far.

First denote $d_i$ the unique positive solution, for $e_i>0$, to
\begin{equation*}
  e_i = d_i\left(\bar{c}_i - \frac1N\sum_{l=1}^{n_i}\frac{p_{il}d_i}{1+p_{il}d_i} \right)\ .
\end{equation*}
This solution exists and is unique due to the arguments given in the introduction of Step 2 of the proof of Theorem~\ref{th:sumRWT}.

Similar to the proof of Theroem~\ref{th:mutual_info}, we proceed by extending the matrix $\P_i$ to an $N_i$-dimensional matrix with the last $N_i-n_i$ diagonal entries filled with zeros. This way, we can write
\begin{equation*}
  e_i = d_i \left(\frac1N\sum_{l=1}^{N_i} \left[1 -  \frac{p_{il}d_i}{1+p_{il}d_i}\right]\right) = \frac1N\sum_{l=1}^{N_i}\frac{d_i}{1+p_{il}d_i}.
\end{equation*}

Since $d_i$ is a continuous mapping of $e_i$ and $e_i\leq \frac{P}{|z|}$, it follows that $d_i$ is bounded from above.

Recall now that for $\lim\sup c_i< 1$ for all $i$ and, for some $z_0<0$, we have that $z<z_0$ implies
\begin{equation*}
  \Exp[|f_i-e_i|^4]= \Exp\left[\left|f_i-\frac1N\sum_{l=1}^{N_i}\frac{d_i}{1+p_{il}d_i}\right|^4\right] \leq \frac{C}{N^2}
\end{equation*}
for some constant $C>0$, where $f_i$ is defined in \eqref{eq:fi_def}. Also, from \eqref{eq:deltai_moment},
\begin{equation*}
  \Exp\left[\left|f_i-\frac1N\sum_{l=1}^{N_i}\frac{\delta_i}{1+p_{il}\delta_i}\right|^4\right]\leq \frac{C_1}{N^2}
\end{equation*}
for some $C_1>C$. From these two inequalities, we have
\begin{equation*}
  \Exp\left[\left|\frac1N\sum_{l=1}^{N_i}\frac{\delta_i}{1+p_{il}\delta_i}-\frac1N\sum_{l=1}^{N_i}\frac{d_i}{1+p_{il}d_i}\right|^4\right] \leq \frac{16C_1}{N^2}.
\end{equation*}
Also, from an immediate application of the trace lemma, Lemma \ref{le:trace_Haar}, we remind that
\begin{align*}
  \Exp\left[\left|\w_{il}^\herm \H_i^\herm\left(\B_{(i,l)}-z\I_N\right)^{-1}\H_i \w_{il} - \delta_i \right|^4\right]\leq \frac{C_2}{N^2}
\end{align*}
for some $C_2>C_1$.

Together, this implies that for $z$ small enough and for any $k\in\{1,\ldots,n_k\}$,
\begin{align*}
  &\quad\Exp\left[\left|\frac1N\sum_{l=1}^{N_i}\frac{d_i}{1+p_{il}d_i} - \frac1N\sum_{l=1}^{N_i}\frac{\w_{ik}^\herm \H_i^\herm\left(\B_{(i,k)}-z\I_N\right)^{-1}\H_i \w_{ik}}{1+p_{il}\w_{ik}^\herm \H_i^\herm\left(\B_{(i,k)}-z\I_N\right)^{-1}\H_i \w_{ik}}\right|^4\right] \nonumber \\ 
  \leq&\quad 8 \left[\Exp\left[\left|\frac1N\sum_{l=1}^{N_i}\frac{d_i}{1+p_{il}d_i} - \frac1N\sum_{l=1}^{N_i}\frac{\delta_i}{1+p_{il}\delta_i}\right|^4\right]\right. \nonumber \\ 
  &\qquad+\left. \Exp\left[\left|\frac1N\sum_{l=1}^{N_i}\frac{\delta_i}{1+p_{il}\delta_i} - \frac1N\sum_{l=1}^{N_i}\frac{\w_{ik}^\herm \H_i^\herm\left(\B_{(i,k)}-z\I_N\right)^{-1}\H_i \w_{ik}}{1+p_{il}\w_{ik}^\herm \H_i^\herm\left(\B_{(i,k)}-z\I_N\right)^{-1}\H_i \w_{ik}} \right|^4\right] \right] \\
  =&\quad 8 \left[\Exp\left[\left|\frac1N\sum_{l=1}^{N_i}\frac{d_i}{1+p_{il}d_i} - \frac1N\sum_{l=1}^{N_i}\frac{\delta_i}{1+p_{il}\delta_i}\right|^4\right] \right. \nonumber \\ 
  &\qquad+ \left. \Exp\left[\left|\frac 1N\sum_{l=1}^{N_i}\frac{\delta_i - \w_{ik}^\herm \H_i^\herm\left(\B_{(i,k)}-z\I_N\right)^{-1}\H_i \w_{ik}}{(1+p_{il}\delta_i)(1+p_{il}\w_{ik}^\herm \H_i^\herm\left(\B_{(i,k)}-z\I_N\right)^{-1}\H_i \w_{ik})} \right|^4\right] \right] \\
  \leq&\quad \frac{136C_2}{N^2}\ .
\end{align*}

This ensures that for $z<z_0$,
\begin{equation}
  \label{eq:convdi}
  \frac1N\sum_{l=1}^{N_i}\frac{d_i}{1+p_{il}d_i} - \frac1N\sum_{l=1}^{N_i} \frac{\w_{ik}^\herm \H_i^\herm\left(\B_{(i,k)}-z\I_N\right)^{-1}\H_i \w_{ik}}{1+p_{il}\w_{ik}^\herm \H_i^\herm\left(\B_{(i,k)}-z\I_N\right)^{-1}\H_i \w_{ik}} \asto 0
\end{equation}
irrespectively of the choice of $k$.

Since the function $f:x\mapsto \frac1N\sum_{l=1}^{N_i}\frac{x}{1+p_{il}x}$ is continuous and has positive derivative, it is a one-to-one continuous function. Therefore, for $\B_1,\B_2,\ldots$ a realization such that the convergence of \eqref{eq:convdi} is ensured, we also have by continuity $d_i-\w_{ik}^\herm \H_i^\herm\left(\B_{(i,k)}-z\I_N\right)^{-1}\H_i \w_{ik}\to 0$. Finally,
\begin{equation}\label{eq:diwik}
  d_i-\w_{ik}^\herm \H_i^\herm\left(\B_{(i,k)}-z\I_N\right)^{-1}\H_i \w_{ik}\asto 0.
\end{equation}

Noticing from \eqref{eq:useful_equality} that $d_i=\frac{e_i}{\bar{c}_i-e_i\bar{e}_i}$, we have proved the convergence for $z<z_0$. The Vitali theorem then ensures that the convergence holds true for all $z<0$ since $e_i$ and $\bar{e}_i$ have analytic extensions on a neighborhood of $\RR_-$ (see the proof of Theorem~\ref{th:sumRWT}, Step 1).

Since the quantities $d_i$ and $\w_{ik}^\herm \H_i^\herm\left(\B_{(i,k)}-z\I_N\right)^{-1}\H_i \w_{ik}$ are uniformly bounded for all $N$ (a result that holds surely since we assumed the $\H_i$ deterministic), the dominated convergence theorem also ensures that the convergence holds in the first mean.

In order to prove Corollary \ref{cor:MMSEa} in the almost sure form, we simply invoke the continuous mapping theorem \cite[Theorem 2.3]{vdv} for the function $\phi:x\mapsto \frac1N\sum_{k=1}^K\sum_{i=1}^{n_k}\log(1+p_{ik}x)$ on the convergence \eqref{eq:diwik}. The convergence in the mean sense is obtained using the boundedness of $d_i$ and $\w_{ik}^\herm \H_i^\herm\left(\B_{(i,k)}-z\I_N\right)^{-1}\H_i \w_{ik}$ uniformly on $N$ and hence the boundedness of their image by $\phi$. The dominated convergence theorem then gives the result.

\section{Proof of Theorem~\ref{th:fundeq}}
\label{app:proofs}
\label{app:proof_fundeq}

It was shown in \eqref{eq:uniquesoly} that, for any fixed $b_k({\sigma^2})\ge0$, the following equation in $\bar{b}_k({\sigma^2})$:
\begin{align*}
 \bar{b}_k({\sigma^2}) = \frac1N\trace\Pm_k\Big( b_k({\sigma^2})\Pm_k + \LSB\bar{c}_k-b_k({\sigma^2})\bar{b}_k({\sigma^2})\RSB\Id_{n_k}\Big)^{-1}
\end{align*}
has a unique solution, satisfying $0\leq \bar{b}_k({\sigma^2})< c_k\bar{c}_k/b_k({\sigma^2})$. Thus, $\bar{b}_k({\sigma^2})$ is uniquely determined by $b_k({\sigma^2})$. 
Consider now the following functions for $k\in\{1,\dots,K\}$ and ${\sigma^2}>0$:
\begin{align*}
 h_k(x_1,\dots,x_K) \mapsto \frac1N\sum_{j=1}^{N_k}\frac{\zeta_{kj}({\sigma^2})}{1+\bar{b}_k\zeta_{kj}({\sigma^2})}
\end{align*}
where $\bar{b}_k\in[0,c_k\bar{c}_k/x_k)$ and $\zeta_{kj}({\sigma^2})\ge 0$ are the unique solutions to the following fixed-point equations:
\begin{align}\label{eq:bare}
 \bar{b}_k & = \frac1N\trace\Pm_k\Big( x_k\Pm_k + \LSB\bar{c}_k-x_k\bar{b}_k\RSB\Id_{n_k}\Big)^{-1}\\\label{eq:T}
\zeta_{kj}({\sigma^2}) &=\frac1N\trace\Rm_{kj} \LB\frac{1}{N}\sum_{k=1}^K\sum_{j=1}^{N_k}\frac{\bar{b}_k\Rm_{k,j}}{1+\bar{b}_k\zeta_{kj}({\sigma^2})}+{\sigma^2}\Id_N\RB^{-1}.
\end{align}
Similar to the proof of Theorem~\ref{th:fundamental_eq}, it is now sufficient to prove that the $K$-variate function $\hv: (x_1,\dots,x_K)\mapsto (h_1,\ldots,h_K)$ is a standard function and to apply Theorem~\ref{th:standardfunctions} to conclude on the existence and uniqueness of a solution to $x_k = h_k(x_1,\dots,x_K)$ for all $k$. The associated fixed-point algorithm follows the recursive equations
\begin{align*}
 x_k^{(t+1)} = h_k(x_1^{(t)},\dots,x_K^{(t)}),\qquad k=1,\dots,K
\end{align*}
for $t\ge0$ and for any set of initial values $x_1^{(0)},\dots,x_K^{(0)}>0$, which then converge, as $t\to\infty$, to the fixed-point.

Showing positivity is straightforward: For ${\sigma^2}>0$, we have $\zeta_{kj}({\sigma^2})> 0$ by Theorem~\ref{th:detequcorr} in Appendix~\ref{sec:appendixB} and $\bar{b}_k\ge0$ by its definition. Thus, $h_k(x_1,\dots,x_K)> 0$ for all $x_1,\ldots,x_K>0$.

To prove monotonicity of $h_k(x_1,\dots,x_K)$, we first recall the following result from \eqref{eq:propertyxibxi}. Let $x_k > x_k'$, and consider $\bar{b}_k$ and $\bar{b}_k'$ the corresponding solutions to \eqref{eq:bare}. Then,
\begin{align}
	\label{lem:gk}
 \text{(i)}\ \ \bar{b}_k<\bar{b}_k' \qquad\qquad \text{(ii)}\ \  x_k\bar{b}_k>x_k'\bar{b}_k'.
\end{align}

We now prove a further result. Let ${\sigma^2}>0$ and assume $\bar{b}_k > \bar{b}_k'$. Consider $\zeta_{kj}({\sigma^2})$ and $\zeta_{kj}'({\sigma^2})$ as the unique solutions to \eqref{eq:T} for $\bar{b}_k$ and $\bar{b}_k'$, respectively. Then,
\begin{align}
\label{lem:gkbar}
 \text{(i)}\ \ \zeta_{kj}({\sigma^2})\leq \zeta'_{kj}({\sigma^2}) \qquad\qquad \text{(ii)}\ \ \bar{b}_k\zeta_{kj}({\sigma^2})>\bar{b}_k'\zeta'_{kj}({\sigma^2}).
\end{align}

\begin{IEEEproof}
The proof is based on the consideration of an extended version of the random matrix model assumed in Theorem~\ref{th:detequcorr}. Let us consider the following random matrices $\Hm_k^L\in\CC^{LN\times LN_k}$, given as
\begin{align}\label{eq:extended_model}
 \Hm_k^L = \frac{1}{\sqrt{LN}}\LSB\LB\Rm^L_{k1}\RB^{\frac12}\Zm^L_{k1},\dots,\LB\Rm^L_{kN_k}\RB^{\frac12}\Zm^L_{kN_k}\RSB
\end{align}
where $\Rm^L_{kj}=\diag(\Rm_{kj},\dots,\Rm_{kj})\in\CC^{LN\times LN}$ are block-diagonal matrices consisting of $L$ copies of the matrix $\Rm_{kj}$ and $\Zm^L_{kj}\in\CC^{LN\times L}$ are random matrices composed of i.i.d.\@ entries with zero mean, unit variance and finite moment of order $4+\epsilon$, for some $\epsilon>0$. 
We define the following matrices which will be of repeated use:
\begin{align*}
\tilde{\Bm}^L&=\sum_{k=1}^K\bar{b}_k\Hm_k^L\LB\Hm_k^L\RB\htp, \qquad \tilde{\Bm'}^L=\bar{b}_k'\Hm_k^L\LB\Hm_k^L\RB\htp + \sum_{l=1,l\ne k}^K\bar{b}_l\Hm_l^L\LB\Hm_l^L\RB\htp\\
\Qm&=\LB\tilde{\Bm}^L +{\sigma^2}\Id_{NL}\RB^{-1}, \qquad \Qm'=\LB\tilde{\Bm'}^L +{\sigma^2}\Id_{NL}\RB^{-1}.
\end{align*}
One can verify from Theorem~\ref{th:detequcorr} that for any fixed $N,N_1,\dots,N_K$, the following limit holds:
\begin{align*}
 \frac{1}{LN}\trace\Rm_{kj}^L\LB\tilde{\Bm}^L +{\sigma^2}\Id_{NK}\RB^{-1}  \xrightarrow[L\to\infty]{\text{a.s}}\zeta_{kj}({\sigma^2}).
\end{align*}
Thus, any properties of the random quantities on the left-hand side of the previous equation also hold for the deterministic quantities $\zeta_{kj}({\sigma^2})$. We will exploit this fact for the termination of the proof.
The matrices $\tilde{\Bm}^L$ and $ \tilde{\Bm'}^L$ differ only by $\bar{b}_k$. This assumption will be sufficient for the proof since the case $\bar{b}_l>\bar{b}_l'$ for $l\in\{1,\dots,K\}$ follows by simple iteration of the case $\bar{b}_l=\bar{b}_l'$ for $l\neq k$ and $\bar{b}_k>\bar{b}_k'$.

To prove (i), it is now sufficient to show that, for any $L$,
\begin{align*}
 \frac1N\trace\Rm^L_{k,j}\LB\Qm-\Qm'\RB <0.
\end{align*}
By Lemma~\ref{lem:traceinequ}, this is equivalent to proving $\LB\Qm\RB^{-1}-\LB\Qm'\RB^{-1}\succ0$, which is straightforward since
\begin{align*}
 \LB\Qm\RB^{-1}-\LB\Qm'\RB^{-1} &= \tilde{\Bm}^L - \tilde{\Bm'}^L = (\bar{b}_k-\bar{b}_k')\Hm_k^L\LB\Hm_k^L\RB\htp \succ 0.
\end{align*}
Thus,
\begin{align*}
	\frac1{NL}\trace\Rm^L_{k,j}\LB\Qm-\Qm'\RB \xrightarrow[L\to\infty]{\text{a.s}}\zeta_{kj}({\sigma^2})-\zeta'_{kj}({\sigma^2})\leq 0
\end{align*}
since $\zeta_{kj}({\sigma^2})$ and $\zeta'_{kj}({\sigma^2})$ do not depend on $L$.

For (ii), we need to show that
\begin{align*}
 \bar{b}_k\frac{1}{LN}\trace\Rm_{kj}^L\Qm - \bar{b}_k'\frac{1}{LN}\trace\Rm_{kj}^L\Qm'>0.
\end{align*}
Similarly to the previous part of the proof, it is sufficient to show that $\LB\bar{b}_k\Qm\RB^{-1}-\LB\bar{b}_k'\Qm'\RB^{-1}\prec0$. Hence,
\begin{align*}
 \LB\bar{b}_k\Qm\RB^{-1}-\LB\bar{b}_k'\Qm'\RB^{-1} &= \frac{1}{\bar{b}_k}\LB\tilde{\Bm}^L+{\sigma^2}\Id_{NL}\RB - \frac{1}{\bar{b}'_k}\LB\tilde{\Bm'}^L+{\sigma^2}\Id_{NL}\RB\\
&= {\sigma^2}\LB\frac{1}{\bar{b}_k}-\frac{1}{\bar{b}'_k}\RB\Id_{NL} + \LB\frac{1}{\bar{b}_k}-\frac{1}{\bar{b}'_k}\RB\sum_{l=1,l\ne k}^K\bar{b}_l\Hm_l^L\LB\Hm_l^L\RB\htp\\
&\prec 0
\end{align*}
since ${\sigma^2}>0$, $\bar{b}_k>\bar{b}_k'$ and $\bar{b}_l\ge0$ for all $l$.
\end{IEEEproof}\vspace{10pt}

Consider now $(x_1,\dots,x_K)$ and $(x_1',\dots,x_K')$, such that $x_k>x_k'\ \forall k$, and denote by $(\bar{b}_1,\dots,\bar{b}_K$) and $(\bar{b}'_1,\dots,\bar{b}'_K)$ the corresponding solutions to \eqref{eq:bare}. Denote by $\zeta_{kj}({\sigma^2})$ and $\zeta'_{kj}({\sigma^2})$ the unique solutions to \eqref{eq:T} for $(\bar{b}_1,\dots,\bar{b}_K)$ and $(\bar{b}'_1,\dots,\bar{b}'_K)$, respectively.
It follows from \eqref{lem:gk} that $\bar{b}_k<\bar{b}_k'\ \forall k$. Equation~\eqref{lem:gkbar} now implies that $\zeta_{kj}({\sigma^2})\geq \zeta'_{kj}({\sigma^2})$ and $\bar{b}_k\zeta_{kj}({\sigma^2})<\bar{b}_k'\zeta_{kj}'({\sigma^2})$. Combining these results yields
\begin{align*}
 h_k(x_1,\dots,x_K) = \frac1N\sum_{j=1}^{N_k}\frac{\zeta_{kj}({\sigma^2})}{1+\bar{b}_k\zeta_{kj}({\sigma^2})} > \frac1N\sum_{j=1}^{N_k}\frac{\zeta'_{kj}({\sigma^2})}{1+\bar{b}_k'\zeta'_{kj}({\sigma^2})} =h_k(x_1',\dots,x_K')
\end{align*}
which proves monotonicity.

To prove scalability, let $\alpha>1$, and consider the following difference:
\begin{align*}
 \alpha h_k(x_1,\dots,x_K) -  h_k(\alpha x_1,\dots,\alpha x_K) &= \frac1N \sum_{j=1}^{N_k} \frac{\alpha \zeta_{kj}({\sigma^2})}{1+\bar{b}_k\zeta_{kj}({\sigma^2})}-\frac{\zeta^{(\alpha)}_{kj}({\sigma^2})}{1+\bar{b}_k^{(\alpha)}\zeta_{kj}^{(\alpha)}({\sigma^2})}\\
&= \frac1N \sum_{i=1}^{N_k} \frac{\LSB\alpha \zeta_{kj}({\sigma^2}) - \zeta_{kj}^{(\alpha)}({\sigma^2})\RSB+ \zeta_{kj}({\sigma^2})  \zeta_{kj}^{(\alpha)}({\sigma^2})\LSB\alpha\bar{b}_k^{(\alpha)} - \bar{b}_k\RSB}{\LSB1+\bar{b}_k  \zeta_{kj}({\sigma^2})\RSB\LSB1+\bar{b}_k^{(\alpha)} \zeta_{kj}^{(\alpha)}({\sigma^2})\RSB}
 \end{align*}
where we have denoted by $\bar{b}_k^{(\alpha)}$ the solution to \eqref{eq:bare} with $x_k$ replaced by $\alpha x_k$  and by $\zeta_{kj}^{(\alpha)}({\sigma^2})$ the solution to \eqref{eq:T} for  $\bar{b}_k^{(\alpha)}$. We have from \eqref{lem:gk}-(i) that $\bar{b}_k^{(\alpha)}<\bar{b}_k$ and from \eqref{lem:gk}-(ii) that
\begin{align}\label{eq:inqu1}
 \alpha x_k \bar{b}_k^{(\alpha)} > x_k \bar{b}_k \Longleftrightarrow \alpha \bar{b}_k^{(\alpha)} -\bar{b}_k>0.
\end{align}
It remains now to show that also $\alpha \zeta_{kj}({\sigma^2})-\zeta_{kj}^{(\alpha)}({\sigma^2})>0$. 
To this end, consider the following difference:
\begin{align}\nonumber
 \alpha \zeta_{kj}({\sigma^2})-\zeta_{kj}^{(\alpha)}({\sigma^2})& = \frac1N\trace\Rm_{kj}\LB\alpha\Tm({\sigma^2})-\Tm^{(\alpha)}({\sigma^2})\RB
\end{align}
where 
\begin{align*}
 \Tm({\sigma^2})&=\LB\frac{1}{N}\sum_{k=1}^K\sum_{j=1}^{N_k}\frac{\bar{b}_k\Rm_{k,j}}{1+\bar{b}_k\zeta_{kj}({\sigma^2})}+{\sigma^2}\Id_N\RB^{-1}\\
\Tm^{(\alpha)}({\sigma^2})&=\LB\frac{1}{N}\sum_{k=1}^K\sum_{j=1}^{N_k}\frac{\bar{b}_k^{(\alpha)}\Rm_{k,j}}{1+\bar{b}_k^{(\alpha)}\zeta_{kj}^{(\alpha)}({\sigma^2})}+{\sigma^2}\Id_N\RB^{-1}.
\end{align*}
By Lemma~\ref{lem:traceinequ}, it is now sufficient to show that $\LB\Tm^{(\alpha)}(z)\RB^{-1}\succ\LB\alpha\Tm(z)\RB^{-1}$. Write therefore
\begin{align*}
&\ \LB\Tm^{(\alpha)}({\sigma^2})\RB^{-1}- \LB\alpha\Tm({\sigma^2})\RB^{-1}\\=\ &\ {\sigma^2}\LB1-\frac1\alpha\RB\Id_N +\frac1N\sum_{k=1}^K\sum_{j=1}^{N_k}\frac{\LSB\alpha\bar{b}_k^{(\alpha)}-\bar{b}_k\RSB+\bar{b}_k^{(\alpha)}\bar{b}_k\LSB\alpha \zeta_{kj}({\sigma^2})-\zeta_{kj}^{(\alpha)}({\sigma^2})\RSB}{\alpha\LSB1+\bar{b}_k \zeta_{kj}({\sigma^2})\RSB\LSB1+\bar{b}_k^{(\alpha)}\zeta_{kj}^{(\alpha)}({\sigma^2})\RSB}\Rm_{kj}.
\end{align*}
The first summand is positive definite since ${\sigma^2}>0$ and $\alpha>1$. All other terms are also positive definite since $\alpha\bar{b}_k^{(\alpha)}-\bar{b}_k>0$ from \eqref{eq:inqu1} and $\alpha\bar{b}_k^{(\alpha)}\bar{b}_k\zeta_{kj}({\sigma^2})>\bar{b}_k\bar{b}_k^{(\alpha)}\zeta_{kj}^{(\alpha)}({\sigma^2})$, since $\alpha\bar{b}_k^{(\alpha)}>\bar{b}_k$ and $\bar{b}_k\zeta_{kj}({\sigma^2})>\bar{b}_k^{(\alpha)}\zeta_{kj}^{(\alpha)}({\sigma^2})$ by \eqref{lem:gkbar}-(ii) and \eqref{lem:gk}-(i). Since the sum of positive definite matrices is also positive definite, we have $\alpha \zeta_{kj}({\sigma^2})-\zeta_{kj}^{(\alpha)}({\sigma^2})>0$. This terminates the proof of scalability.

Thus, we have shown $\hv: (x_1,\dots,x_K)\mapsto (h_1,\ldots,h_K)$ to be a standard function.
Moreover, from the fixed-point algorithms described in Theorem \ref{th:fundamental_eq} and Theorem~\ref{th:detequcorr}, and the fact that the $\zeta_{kj}$ are bounded (and therefore there exist $x_1,\ldots,x_K$ such that $x_i\geq h_i(x_1,\ldots,x_K)$ for each $i$), we have the following algorithm to compute $\bar{b}_k$ and $\zeta_{kj}({\sigma^2})$:
\begin{align*}
 \bar{b}_k = \lim_{t\to\infty}  \bar{b}_k^{(t)}, \qquad\zeta_{kj}({\sigma^2})= \lim_{t\to\infty}  \zeta_{kj}^{(t)}({\sigma^2})
\end{align*}
where
\begin{align*}
 \bar{b}_k^{(t)} & = \frac1N\trace\Pm_k\Big( x_k\Pm_k + \LSB\bar{c}_k-x_k\bar{b}_k^{(t-1)}\RSB\Id_{n_k}\Big)^{-1}\\
 \zeta_{kj}^{(t)}({\sigma^2}) &= \frac1N\trace\Rm_{kj} \LB\frac{1}{N}\sum_{k=1}^K\sum_{j=1}^{N_k}\frac{\bar{b}_k\Rm_{k,j}}{1+\bar{b}_k\zeta_{kj}^{(t-1)}({\sigma^2})}+{\sigma^2}\Id_N\RB^{-1}
\end{align*}
and $\bar{b}_k^{(0)}$ can take any value in $[0,c_k\bar{c}_k/x_k)$ and $\zeta_{kj}^{(0)}({\sigma^2})=1/{\sigma^2}$ for all $k,j$.

\section{Proof of Theorem~\ref{th:mutinf}}
\label{app:proof_mutinf}
 We begin by proving the following result: 
\begin{align}\label{eq:conv1}
 \max_k|\bar{a}_k({\sigma^2}) - \bar{b}_k({\sigma^2})| &\asto0\\\label{eq:conv2}
 \max_k|a_k({\sigma^2}) - b_k({\sigma^2})| &\asto0
\end{align}
where $\bar{a}_k({\sigma^2})$, $a_k({\sigma^2})$ are defined in Theorem~\ref{th:fundequdet} and $\bar{b}_k({\sigma^2})$, $b_k({\sigma^2})$ are defined in Theorem~\ref{th:fundeq}, assuming that the matrices $\Hm_k$ are random and modeled as described in \eqref{eq:channelmodel}.
For notational simplicity, we will drop from now on the dependence on ${\sigma^2}$. From standard lemmas of matrix analysis, we have
\begin{align}\nonumber
 a_k &= \frac1N\trace \Hm_k\Hm_k\htp\LB\sum_{i=1}^K \bar{a}_{i} \Hm_i\Hm_i\htp +{\sigma^2}\Id_N \RB^{-1}\\\nonumber
&= \frac1N\sum_{j=1}^{N_k} \hv_{kj}\htp\LB\sum_{i=1}^K \bar{a}_{i} \Hm_i\Hm_i\htp +{\sigma^2}\Id_N \RB^{-1} \hv_{kj}\\\nonumber
&=\frac1N\sum_{j=1}^{N_k} \frac{\hv_{kj}\htp\LB\sum_{i=1}^K \bar{a}_{i} \Hm_i\Hm_i\htp - \bar{a}_{k}\hv_{kj}\hv_{kj}\htp +{\sigma^2}\Id_N \RB^{-1} \hv_{kj}}{1+\bar{a}_{k}\hv_{kj}\htp\LB\sum_{i=1}^K \bar{a}_{i} \Hm_i\Hm_i\htp- \bar{a}_{k}\hv_{kj}\hv_{kj}\htp +{\sigma^2}\Id_N \RB^{-1} \hv_{kj}}
\end{align}
where the last step follows from Lemma~\ref{lem:inversion}. If $\bar{a}_i$ were not dependent on $\hv_{kj}$, we could now simply proceed by applying Lemma~\ref{lem:trace} to the individual quadratic forms, i.e.:
\begin{align*}
 \hv_{kj}\htp\LB\sum_{i=1}^K \bar{a}_{i} \Hm_i\Hm_i\htp - \bar{a}_{k}\hv_{kj}\hv_{kj}\htp +{\sigma^2}\Id_N \RB^{-1} \hv_{kj} \asymp \frac1N\trace\Rm_{kj}\LB\sum_{i=1}^K \bar{a}_{i} \Hm_i\Hm_i\htp - \bar{a}_{k}\hv_{kj}\hv_{kj}\htp +{\sigma^2}\Id_N \RB^{-1}
\end{align*}
where, in the following, for $\{a_N\}$ and $\{b_N\}$ two sequences of random variables, we denote $a_N\asymp b_N$ the equivalence relation $a_N-b_N\asto0$ for $N\to\infty$.

However, in order to show that this step is correct, in a similar manner as in the proof of Theorem \ref{th:sumRWT}, we need the following intermediate arguments. Define $\bar{a}_{i,kj}$ and $a_{i,kj}$ as the unique solutions to the following fixed-point equations:
\begin{align*}
 a_{i,{kj}} &= \frac1N\trace \Hm_{i,kj}\Hm_{i,kj}\htp\LB\sum_{l=1}^K \bar{a}_{l,kj} \Hm_{l,kj}\Hm_{l,kj}\htp +{\sigma^2}\Id_N \RB^{-1}\\
\bar{a}_{i,kj} &= \frac1N\trace\Pm_i\LB a_{i,{kj}}\Pm_i + \LSB\bar{c}_k -  a_{i,{kj}} \bar{a}_{i,{kj}}\Id_{n_i}\RSB\RB^{-1}
\end{align*}
for $i\in\{1,\dots,K\}$, where
\begin{align*}
 \Hm_{i,kj} = \begin{cases}
               \Hm_i, & k\ne i\\
	      \left[ \hv_{k1}\cdots\hv_{k{j-1}}\hv_{kj+1}\cdots \hv_{kN_i}\right], & k=i
              \end{cases}.
\end{align*}
Thus, $\bar{a}_{i,kj}$ and $a_{i,kj}$ are independent of $\hv_{kj}$. Following similar steps as in the proof of Theorem~\ref{th:sumRWT} (Step 3), one can show that for $i\in\{1,\dots,K\}$ and all $k,j$,
\begin{align}\label{eq:intconv}
 a_{i,{kj}} - a_{i} \asto 0,\quad  \bar{a}_{i,{kj}} - \bar{a}_{i} \asto 0. 
\end{align}
Thus, we have

\begin{align}\nonumber
&\ \frac1N\sum_{j=1}^{N_k} \frac{\hv_{kj}\htp\LB\sum_{i=1}^K \bar{a}_{i} \Hm_i\Hm_i\htp - \bar{a}_{k}\hv_{kj}\hv_{kj}\htp +{\sigma^2}\Id_N \RB^{-1} \hv_{kj}}{1+\bar{a}_{k}\hv_{kj}\htp\LB\sum_{i=1}^K \bar{a}_{i} \Hm_i\Hm_i\htp- \bar{a}_{k}\hv_{kj}\hv_{kj}\htp +{\sigma^2}\Id_N \RB^{-1} \hv_{kj}}\\\nonumber
\overset{\text{(a)}}{\asymp}&\ \frac1N\sum_{j=1}^{N_k} \frac{\hv_{kj}\htp\LB\sum_{i=1}^K \bar{a}_{i,kj} \Hm_i\Hm_i\htp - \bar{a}_{k,kj}\hv_{kj}\hv_{kj}\htp +{\sigma^2}\Id_N \RB^{-1} \hv_{kj}}{1+\bar{a}_{k}\hv_{kj}\htp\LB\sum_{i=1}^K \bar{a}_{i,kj} \Hm_i\Hm_i\htp- \bar{a}_{k,kj}\hv_{kj}\hv_{kj}\htp +{\sigma^2}\Id_N \RB^{-1} \hv_{kj}}\\\nonumber
\overset{\text{(b)}}{\asymp}&\ \frac1N\sum_{j=1}^{N_k} \frac{\frac1N\trace\Rm_{kj}\LB\sum_{i=1}^K \bar{a}_{i,kj} \Hm_i\Hm_i\htp - \bar{a}_{k,kj}\hv_{kj}\hv_{kj}\htp +{\sigma^2}\Id_N \RB^{-1} }{1+\bar{a}_{k}\frac1N\trace\Rm_{kj}\LB\sum_{i=1}^K \bar{a}_{i,kj} \Hm_i\Hm_i\htp- \bar{a}_{k,kj}\hv_{kj}\hv_{kj}\htp +{\sigma^2}\Id_N \RB^{-1}}\\\nonumber
\overset{\text{(c)}}{\asymp}&\  \frac1N\sum_{j=1}^{N_k} \frac{\frac1N\trace\Rm_{kj}\LB\sum_{i=1}^K \bar{a}_{i} \Hm_i\Hm_i\htp +{\sigma^2}\Id_N \RB^{-1} }{1+\bar{a}_{k}\frac1N\trace\Rm_{kj}\LB\sum_{i=1}^K \bar{a}_{i} \Hm_i\Hm_i\htp +{\sigma^2}\Id_N \RB^{-1}}\\\label{eq:ekconv}
\overset{\text{(d)}}{\asymp}&\ \frac1N\sum_{j=1}^{N_k} \frac{\frac1N\trace\Rm_{kj}\bar{\Tm} }{1+\bar{a}_{k}\frac1N\trace\Rm_{kj}\bar{\Tm}}
\end{align}
where (a) follows from \eqref{eq:intconv}, (b) follows from Lemma~\ref{lem:trace} and Lemma~\ref{lem:convergence_ratios}, (c) is again due to \eqref{eq:intconv} and Lemma~\ref{lem:rank1perturbation}, and (d) follows from an application of Theorem~\ref{th:detequcorr}, where we have defined
\begin{align*}
 \bar{\Tm} = \LB\frac{1}{N}\sum_{k=1}^K\sum_{j=1}^{N_k}\frac{\bar{a}_{k}\Rm_{kj}}{1+\bar{a}_{k}\frac1N\trace\Rm_{kj}\bar{\Tm}}+{\sigma^2}\Id_N\RB^{-1}.
\end{align*}
Note again that Theorem~\ref{th:detequcorr} cannot be directly applied here since the quantities $\bar{a}_i$ depend on the matrices $\Hm_i$. However, it is immediate to show that the result extends in this case, by replacing $\bar{a}_i$ by $\bar{a}_{i,kj}$ at each necessary step of the proof. 

Hence, we can write
\begin{align*}
  a_k & =\frac1N\trace \Hm_k\Hm_k\htp\LB\sum_{i=1}^K \bar{a}_i \Hm_i\Hm_i\htp +{\sigma^2}\Id_N \RB^{-1}= \frac1N\sum_{j=1}^{N_k} \frac{\frac1N\trace\Rm_{kj}\bar{\Tm} }{1+\bar{a}_k\frac1N\trace\Rm_{kj}\bar{\Tm}} + \epsilon_{N,k}
\end{align*}
for some sequences of reals $\epsilon_{N,k}$, satisfying $\epsilon_{N,k}\to0$.

Recall now the following definitions for $k=1,\dots,K$:
\begin{align*}
 a_k & = \frac1N\sum_{j=1}^{N_k} \frac{\frac1N\trace\Rm_{kj}\bar{\Tm} }{1+\bar{a}_k\frac1N\trace\Rm_{kj}\bar{\Tm}} + \epsilon_{N,k}\\
b_k & = \frac1N\sum_{j=1}^{N_k} \frac{\frac1N\trace\Rm_{kj}\Tm }{1+\bar{b}_k\frac1N\trace\Rm_{kj}\Tm} \\
\bar{a}_k & = \frac1N\sum_{j=1}^{n_k}\frac{p_{kj}}{\bar{c}_k - a_k\bar{a}_k + a_k p_{kj}},\qquad 0\le\bar{a}_k<c_k\bar{c}_k/a_k\\
\bar{b}_k & = \frac1N\sum_{j=1}^{n_k}\frac{p_{kj}}{\bar{c}_k - b_k\bar{b}_k + b_k p_{kj}},\qquad 0\le\bar{b}_k<c_k\bar{c}_k/b_k
\end{align*}
where 
\begin{align*}
 \bar{\Tm} &= \LB\frac{1}{N}\sum_{k=1}^K\sum_{j=1}^{N_k}\frac{\bar{a}_{k}\Rm_{kj}}{1+\bar{f}_{N,k}\frac1N\trace\Rm_{kj}\bar{\Tm}}+{\sigma^2}\Id_N\RB^{-1}\\
\Tm &= \LB\frac{1}{N}\sum_{k=1}^K\sum_{j=1}^{N_k}\frac{\bar{b}_k\Rm_{kj}}{1+\bar{b}_k\frac1N\trace\Rm_{kj}\Tm}+{\sigma^2}\Id_N\RB^{-1}.
\end{align*}

Denote $P = \max_k\{\lim\sup\lVert\Pm_k\rVert\}$, $R = \max_m\{\lim\sup\lVert\tilde{\Rm}_{m}\rVert\}$, $c_+ = \max_k\{\lim\sup c_k\}$ and $\bar{c}_- = \min_k\{\lim\inf\bar{c}_k\}$, $\bar{c}_+ = \max_k\{\lim\sup\bar{c}_k\}$. 
Since we are interested in the asymptotic limit $N\to\infty$, we assume from the beginning that $N$ is sufficiently large, so that the following inequalities hold for all $k$:
\begin{align*}
 c_k \leq c_+,\quad \bar{c}_- \leq \bar{c}_k \leq \bar{c}_+,\quad \lVert\Pm_k\rVert\leq P,\quad \lVert\Rm_{kj}\rVert\leq R.
\end{align*}
We then have the following properties:
\begin{align}\label{eq:inequcase1}
 \bar{a}_k \leq \frac{P}{(1-c_+)\bar{c}_-},\quad \bar{b}_k \leq \frac{P}{(1-c_+)\bar{c}_-},\quad b_k\bar{b}_k < c_+\bar{c}_+, \quad a_k\bar{a}_k < c_+\bar{c}_+.
\end{align}
For notational simplicity, we define the following quantities:
\begin{align*}
 \xi = \max_k|a_k-b_k|,\qquad  \bar{\xi} = \max_k|\bar{a}_k-\bar{b}_k|.
\end{align*}
We will show in the sequel that $\xi\asto 0$ and $ \bar{\xi}\asto 0$ as $N\to\infty$.

Consider first the following difference:
\begin{align*}
 \sup_{k,j}\left|\frac1N\trace\Rm_{kj}\LB\Tm-\bar{\Tm}\RB \right|&= \sup_{k,j}\left|\frac1N\trace\Rm_{kj}\Tm\LB\frac1N\sum_{l=1}^K\sum_{m=1}^{N_l}\frac{\bar{a}_l\Rm_{lm}}{1+\bar{a}_l\frac1N\trace\Rm_{lm}\bar{\Tm}}-\frac{\bar{b}_l\Rm_{lm}}{1+\bar{b}_l\frac1N\trace\Rm_{lm}\bar{\Tm}}\RB\bar{\Tm}\right|\\
&= \sup_{k,j}\left|\frac1N\sum_{l=1}^K\sum_{m=1}^{N_l}\frac{\bar{a}_l-\bar{b}_l + \bar{a}_l\bar{b}_l\LB\frac1N\trace\Rm_{lm}\Tm-\frac1N\trace\Rm_{lm}\bar{\Tm}\RB}{\LB1+\bar{a}_l\frac1N\trace\Rm_{lm}\bar{\Tm}\RB\LB1+\bar{b}_l\frac1N\trace\Rm_{lm}\bar{\Tm}\RB} \frac1N\trace\Rm_{kj}\bar{\Tm}\Rm_{lm}\Tm \right|\\
&\leq \frac{R^2}{{\sigma^4}}K\max_k\bar{c}_k\LSB \max_{k}|\bar{a}_k-\bar{b}_k|+\max_k|\bar{a}_k\bar{b}_k|\sup_{k,j}\left|\frac1N\trace\Rm_{kj}\LB\Tm-\bar{\Tm}\RB\right|\RSB\\
&\leq \frac{R^2}{{\sigma^4}}K\bar{c}_+\LSB\bar{\xi}+\frac{P^2}{(1-c_+)^2\bar{c}_-^2}\sup_{k,j}\left|\frac1N\trace\Rm_{kj}\LB\Tm-\bar{\Tm}\RB\right|\RSB
\end{align*}
where the first equality follows from Lemma~\ref{lem:resolvent}. Rearranging the terms yields:
\begin{align}\label{eq:inequ1}
 \sup_{k,j}\left|\frac1N\trace\Rm_{kj}\LB\Tm-\bar{\Tm}\RB\right| \leq  \frac{P^2K\bar{c}_+}{{\sigma^4} - \frac{R^2P^2}{(1-c_+)^2\bar{c}_-^2}}\ \bar{\xi}
\end{align}
for ${\sigma^2}>\frac{RP}{(1-c_+)\bar{c}_-} $.

Consider now the term $\xi=\max_k|a_k-b_k|$:
\begin{align}\nonumber
 \xi &= \max_k\left|\frac1N\sum_{j=1}^{N_k}\frac{\frac1N\trace\Rm_{kj}\LB\bar{\Tm}-\Tm\RB + (\bar{b}_k-\bar{a}_k)\frac1N\trace\Rm_{kj}\frac1N\trace\Rm_{kj}\bar{\Tm}}{\LB1+\bar{a}_k\frac1N\trace\Rm_{kj}\bar{\Tm}\RB\LB1+\bar{b}_k\frac1N\trace\Rm_{kj}\Tm\RB} +\epsilon_{N,k} \right|\\\nonumber
&\leq \bar{c}_+\sup_{kj}\left|\frac1N\trace\Rm_{kj}\LB\Tm-\bar{\Tm}\RB\right| + \bar{c}_+\frac{R^2}{{\sigma^4}}\max_k |\bar{a}_k-\bar{b}_k|+\max_k|\epsilon_{N,k}|\\\nonumber
&\leq \frac{P^2K\bar{c}_+^2}{{\sigma^4} - \frac{R^2P^2}{(1-c_+)^2\bar{c}_-^2}} \bar{\xi} + \frac{\bar{c}_+R^2}{{\sigma^4}} \bar{\xi} + \max_k|\epsilon_{N,k}|\\\label{eq:inequalpha}
&= \LSB\frac{P^2K\bar{c}_+^2}{{\sigma^4} - \frac{R^2P^2}{(1-c_+)^2\bar{c}_-^2}} + \frac{\bar{c}_+R^2}{{\sigma^4}} \RSB \bar{\xi} + \max_k|\epsilon_{N,k}|
\end{align}
where the last inequality follows from \eqref{eq:inequ1}.
Similarly, we have for $\bar{\xi}=\max_k|\bar{a}_k-\bar{b}_k|$:
\begin{align*}
 \bar{\xi} & = \max_k\left|\frac1N\sum_{j=1}^{n_k}p_{kj} \frac{a_k\bar{a}_k-b_k\bar{b}_k+p_{kj}(b_k-a_k)}{(\bar{c}_k-a_k\bar{a}_k +a_k p_{kj})(\bar{c}_k-b_k\bar{b}_k +b_k p_{kj})} \right|\\
&\leq \frac1N\sum_{j=1}^{n_k} \frac{p_{kj}^2\max_k|a_k-b_k|}{(1-c_+)^2\bar{c}_-^2} + p_{kj}\frac{\max_k\LSB \bar{a}_k|a_k-b_k\RSB| + \max_k\LSB b_k|\bar{a}_k-\bar{b}_k|\RSB  }{(1-c_+)^2\bar{c}_-^2}\\
&\le\frac{P^2}{(1-c_+)^2\bar{c}_-^2}\LB1 + \frac{1}{(1-c_+)\bar{c}_-}\RB \xi + \frac{PR\bar{c}_+}{{\sigma^2}(1-c_+)^2\bar{c}_-^2}\bar{\xi}.
\end{align*}
Thus, for ${\sigma^2}\ge\max\left\{ \frac{2PR\bar{c}_+}{(1-c_+)^2\bar{c}_-^2},\frac{RP}{(1-c_+)\bar{c}_-}\right\}$, we have
\begin{align}\label{eq:inequalphabar}
 \bar{\xi}\leq \frac{2P^2}{(1-c_+)^2\bar{c}_-^2}\LB1 + \frac{1}{(1-c_+)\bar{c}_-}\RB \xi.
\end{align}
Replacing \eqref{eq:inequalphabar} in \eqref{eq:inequalpha} leads to
\begin{align*}
 \xi &\le\LSB\frac{P^2K\bar{c}_+^2}{{\sigma^4} - \frac{R^2P^2}{(1-c_+)^2\bar{c}_-^2}} + \frac{\bar{c}_+R^2}{{\sigma^4}} \RSB \frac{2P^2}{(1-c_+)^2\bar{c}_-^2}\LB1 + \frac{1}{(1-c_+)\bar{c}_-}\RB \xi + \max_k|\epsilon_{N,k}|.
\end{align*}
For ${\sigma^2}$ sufficiently large, we therefore have
\begin{align*}
0\le\xi \leq C \epsilon_{N,k}  \asto 0
\end{align*}
for some $C>0$. This implies that $\xi\asto 0$ and, by \eqref{eq:inequalphabar}, that $\bar{\xi} \asto 0$ .
Since $a_k,b_k,\bar{a}_k,\bar{b}_k$ have analytic extensions in a neighborhood of $\RR_-$ (see the Proof of Theorem~\ref{th:sumRWT} for similar arguments) on which they are (almost surely) uniformly bounded, we have from Vitali's convergence theorem \cite{TIT39} that the almost sure convergence holds true for all ${\sigma^2}\in\RR_+$. This terminates the proof.

\subsection{Convergence of the mutual information}
Consider now the first term of $V_N({\sigma^2})$ in Theorem~\ref{th:logdetcor}. Due to the convergence of $\bar{a}_k-\bar{b}_k\asto0$ and the almost sure boundedness of the $\H_k\H_k^\herm$ matrices, it follows that $\lVert\sum_{k=1}^K(\bar{a}_k-\bar{b}_k)\Hm_k\Hm_k\htp\rVert\asto0$, and we can immediately conclude, by convergence mapping arguments, that
\begin{align*}
 \frac1N\log\det\LB\Id_N+\frac{1}{{\sigma^2}}\sum_{k=1}^K\bar{a}_k\Hm_k\Hm_k\htp\RB-\frac1N\log\det\LB\Id_N+\frac{1}{{\sigma^2}}\sum_{k=1}^K\bar{b}_k\Hm_k\Hm_k\htp\RB\asto0.
\end{align*}
Applying Corollary~\ref{cor:logdet} to the second term yields
\begin{align}\label{eq:conv3}
	\frac1N\log\det\LB\Id_N+\frac{1}{{\sigma^2}}\sum_{k=1}^K\bar{b}_k\Hm_k\Hm_k\htp\RB - \bar{V}_N({\sigma^2})\asto0.
\end{align}
Consider now $\bar{I}^{(a)}_N({\sigma^2})$ and $\bar{I}^{(b)}_N({\sigma^2})$ as defined in Theorems~\ref{th:logdetdet} and \ref{th:mutinf}. It follows from \eqref{eq:conv1}, \eqref{eq:conv2} and \eqref{eq:conv3}, that
\begin{align*}
	\bar{I}^{(a)}_N({\sigma^2}) - \bar{I}^{(b)}_N({\sigma^2}) \asto0.
\end{align*}
This implies also that
\begin{align}\label{eq:conv4}
	I^{(b)}_N({\sigma^2}) - \bar{I}^{(b)}_N({\sigma^2})\asto0.
\end{align}
To prove convergence in the mean, we can no longer use the fact that $I^{(b)}_N(\sigma^2)$ is bounded for all $N$ as in Appendix~\ref{app:mutual_info}, which is now untrue. Instead, we will use the same arguments as in \cite{HAC07}. Denote 
\begin{align*}
	\quad	m_N^{(b)}(z) = \frac1N\tr(\B_N-z\I_N)^{-1}, \quad \bar{m}_N^{(b)}(z) = \frac1N\tr\left(\frac1N\sum_{k=1}^K\sum_{j=1}^{N_k} \frac{\bar{b}_k(-z)\R_{k,j}}{1+b_k(-z)\zeta_{kj}(-z)}-z\I_N \right)^{-1}
\end{align*}
where $m_N^{(b)}(z)$ is the Stieltjes transform of $\B_N$. It is easy to see that
\begin{align*}
	\Exp I^{(b)}_N({\sigma^2}) - \bar{I}^{(b)}_N({\sigma^2}) &= \int_{{\sigma^2}}^\infty\left(\left[\frac1\omega - \EE m^{(b)}_N(-\omega)\right] - \left[\frac1\omega - \bar{m}^{(b)}_N(-\omega)\right] \right)d\omega.
\end{align*}
We now apply the argument from \cite[pp. 923]{HAC07} which shows that
\begin{align*}
	&\ \left|\int_{{\sigma^2}}^\infty\left(\left[\frac1\omega - \EE m^{(b)}_N(-\omega)\right] - \left[\frac1\omega - \bar{m}^{(b)}_N(-\omega)\right] \right)d\omega \right| \\ 
	\leq&\ \int_{{\sigma^2}}^{\infty} \frac1{\omega^2} \left(\left|\Exp\int_0^\infty t dF^{(b)}_N(t)\right|+\left|\frac1N \tr \left(\frac1N\sum_{k=1}^K\sum_{j=1}^{N_k} \frac{\bar{b}_k(\omega)\R_{k,j}}{1+b_k(\omega)\zeta_{kj}(\omega)} \right) \right|\right)d \omega
\end{align*}
the right-hand side of which exists for all $N$ and is uniformly bounded by $\frac2{{\sigma^2}}(KPR)$. Since $m^{(b)}_N(-\omega) - \bar{m}^{(b)}_N(-\omega)\asto 0$ (as a consequence of the convergence $\bar{a}_k-\bar{b}_k\asto 0$), the boundedness of $m^{(b)}_N(-\omega)$ then ensures (by dominated convergence) that $\EE m^{(b)}_N(-\omega) - \bar{m}^{(b)}_N(-\omega)\to 0$. Since the integrand tends to zero and is summable independently of $N$, the dominated convergence theorem now ensures that
\begin{equation*}
	\Exp I^{(b)}_N({\sigma^2}) - \bar{I}^{(b)}_N({\sigma^2}) \to 0.
\end{equation*}

\begin{IEEEproof}[Proof of Theorem \ref{th:sinr}]
	The proof follows directly from \eqref{eq:conv1}, \eqref{eq:conv2}, and Theorem \ref{th:SINR_MMSE}.
\end{IEEEproof}

\vspace{20pt}\begin{IEEEproof}[Proof of Corollary \ref{cor:MMSE}]
	The almost sure convergence follows directly from Theorem \ref{th:mutinf} and the continuous mapping theorem \cite[Theorem 2.3]{vdv}. For the convergence in mean, note first that, as a standard result of information theory, $I_N^{(b)}(\sigma^2)-R_N^{(b)}(\sigma^2) \ge 0 $ for all $N$.
Consider now the extended matrix model where $\Hm_k^L\in\CC^{LN\times LN_k}$ is defined in \eqref{eq:extended_model}, $\Pm_k^L=\Pm_k\otimes\Id_L\in\CC^{Ln_k\times Ln_k}$ and $\Wm_k^L\in\CC^{L N_k \times Ln_k}$ is constructed from $Ln_k$ columns of a $LN_k\times LN_k$ random unitary matrix. Denote $I_{N,L}^{(b)}(\sigma^2)$ and $R_{N,L}^{(b)}(\sigma^2)$ the associated mutual information and MMSE sum-rate for this channel model. One can verify that for this model and by Theorem~\ref{th:mutinf} and the convergence of $R_N^{(b)}(\sigma^2)-\bar{R}_N^{(b)}(\sigma^2)$ in the almost sure sense, the following holds
\begin{align*}
 I_{N,L}^{(b)}(\sigma^2) \xrightarrow[L\to\infty]{\text{a.s.}} & \bar{I}_{N}^{(b)}(\sigma^2) \\
R_{N,L}^{(b)}(\sigma^2) \xrightarrow[L\to\infty]{\text{a.s.}} & \bar{R}_{N}^{(b)}(\sigma^2).
\end{align*}
Thus,
\begin{align*}
 I_{N,L}^{(b)}(\sigma^2) - R_{N,L}^{(b)}(\sigma^2) &=  I_{N,L}^{(b)}(\sigma^2) - \bar{I}_{N}^{(b)}(\sigma^2) + \bar{I}_{N}^{(b)}(\sigma^2) - \bar{R}_{N}^{(b)}(\sigma^2) + \bar{R}_{N}^{(b)}(\sigma^2)- R_{N,L}^{(b)}(\sigma^2)\\ &\xrightarrow[L\to\infty]{\text{a.s.}} \bar{I}_{N}^{(b)}(\sigma^2) - \bar{R}_{N}^{(b)}(\sigma^2)
\end{align*}
from which we can conclude that $\bar{I}_{N}^{(b)}(\sigma^2) - \bar{R}_{N}^{(b)}(\sigma^2) \ge 0$ for all $N$. Using this result, it follows that
\begin{align*}
	\left|R_N^{(b)}(\sigma^2)- \bar{R}_N^{(b)}(\sigma^2)\right | & \leq I_N^{(b)}(\sigma^2) + \bar{I}_N^{(b)}(\sigma^2) \leq I_N^{(b)}(\sigma^2)-\bar{I}_N^{(b)}(\sigma^2) + 2\sup_N\bar{I}_N^{(b)}(\sigma^2) \triangleq v_N.
\end{align*}
Since $v_N\asto 2\sup_N\bar{I}_N^{(b)}(\sigma^2)<\infty$ and $\mathbb{E} v_N\to 2\sup_N\bar{I}_N^{(b)}(\sigma^2)$ by Theorem~\ref{th:mutinf}, it finally follows from \cite[Problem 16.4 (a)]{billingsley} that 
\begin{align*}
  \mathbb{E}R_N^{(b)}(\sigma^2)- \bar{R}_N^{(b)}(\sigma^2) \to 0.
\end{align*}
\end{IEEEproof}

\section{Fundamental lemmas}\label{sec:fundlemmas}

\begin{lemma}[Defining properties of Stieltjes transforms, Theorem 3.2 in \cite{COUbook}]
	\label{le:properties_ST}
  If $m$ is a function analytic on $\CC^+$ such that $m(z)\in\CC^+$ if $z\in\CC^+$ and 
  \begin{equation}
    \label{eq:limy}
    \lim_{y\to\infty}-{\bf i}y~m({\bf i}y) = 1
  \end{equation}
  then $m$ is the Stieltjes transform of a distribution function $F$ given by 
  \begin{equation*}
	  F(b)-F(a) = \lim_{y\to 0}\frac1{\pi}\int_a^b\Im [m(x+{\bf i}y)]dx.
  \end{equation*}
  If, moreover, $zm(z)\in\CC^+$ for $z\in\CC^+$, then $F(0^-)=0$, in which case $m$ has an analytic continuation on $\CC\setminus\RR_+$.
\end{lemma}

\begin{lemma}[Resolvent identity]\label{le:res_id}\label{lem:resolvent} For invertible matrices $\A$ and $\B$, we have the following identity:
$$\A^{-1} - \B^{-1} \ = \ \A^{-1}(\B-\A)\B^{-1}\ . $$
\end{lemma}

\begin{lemma}[A matrix inversion lemma, Equation (2.2) in \cite{SIL95}]
  \label{le:sil_matrix_inversion}
  \label{lem:inversion}
  Let $\A\in\CC^{N\times N}$ be Hermitian invertible, then for any vector $\x\in\CC^N$ and any scalar $\tau\in\CC$ such that $\A+\tau \x\x^\herm$ is invertible
  \begin{equation*} \x^\herm (\A+\tau \x\x^\herm)^{-1} = \frac{\x^\herm\A^{-1}}{1+\tau \x^\herm \A^{-1}\x}. 
  \end{equation*} 
\end{lemma} 

\begin{lemma}[Trace lemma {\cite[Lemma 2.7]{SIL98}}]\label{lem:trace}
 Let $\Am_1,\Am_2,\dots$, with $\Am_N\in\CC^{N\times N}$, be a sequence of matrices with uniformly bounded spectral norm and let $\xv_N=\in\CC^N$ be random vectors of i.i.d.\@ entries with zero mean, variance $1/N$ and eighth order moment of order $\Oc(1/N^4)$, independent of $\Am_N$. Then, as $N\to\infty$,
\begin{align}
 \xv_N\htp\Am_N\xv_N-\frac1N\trace\Am_N\asto 0.
\end{align}
\end{lemma}

\begin{lemma}[Trace lemma for isometric matrices, \cite{DEB03}]
  \label{le:trace_Haar}
  Let $\W$ be $n<N$ columns of an $N\times N$ Haar matrix and suppose $\w$ is a column of $\W$. Let $\B_N$ be an $N\times N$ random matrix, which is a function of all columns of $\W$ except $\w$ and $B=\sup_N \Vert \B_N \Vert < \infty$, then
  \begin{equation*}
    \Exp \left[\left\vert \w^\herm \B_N \w -\frac{1}{N-n} \tr({\bf \Pi}\B_N) \right\vert^4\right] \leq \frac{C}{N^2}, 
  \end{equation*}
  where ${\bf \Pi}=\I_N-\W\W^\herm+\w\w^\herm$ and $C$ is a constant which depends only on $B$ and $\frac{n}{N}$. 
\end{lemma} 

\begin{lemma}[Trace inequality]\label{lem:traceinequ}
 Let $\Am,\Bm,\Rm\in\CC^{N\times N}$, where $\Am$, $\Bm$, and $\Rm$ are nonnegative-definite, satisfying $\Bm\succ \Am$. Then
\begin{align}
 \trace\Rm\LB\Am^{-1} - \Bm^{-1}\RB > 0.
\end{align}
\end{lemma}
\begin{IEEEproof}
Note that $\Bm\succ\Am$ implies by \cite[Corollary 7.7.4]{HOR85} $\Bm^{-1}\prec\Am^{-1}$. Thus, for any vector $\xv\in\CC^N$,
\begin{align}
 \xv\htp\LB\Am^{-1}-\Bm^{-1}\RB\xv> 0 .
\end{align}
Consider now the eigenvalue decomposition of the matrix $\Rm=\Um\Lambdam\Um\htp$, where $\Um=\LSB\uv_1,\dots,\uv_{N}\RSB$ and $\Lambdam=\diag(\lambda_1,\dots,\lambda_{N})$. Since $\lambda_i\ge 0\ \forall i$, we have 
\begin{align}
 \trace\Rm\LB\Am^{-1} - \Bm^{-1}\RB &= \sum_{i=1}^N \lambda_i \uv_i\htp\LB\Am^{-1} - \Bm^{-1}\RB\uv_i > 0.
\end{align}
\end{IEEEproof}

\begin{lemma}[Rank-$1$ perturbation lemma \cite{SIL95}]
  \label{le:rank1perturbation}
  \label{lem:rank1perturbation}
  Let $z<0$, $\A\in\CC^{N\times N}$, $\B\in\CC^{N\times N}$ with $\B$ Hermitian nonnegative definite, and $\v\in \CC^N$. Then,
  \begin{equation*}
    \left|\tr\left((\B-z\I_N)^{-1}-(\B+\v\v^\herm-z\I_N)^{-1}\right)\A\right|\leq\frac{\Vert \A \Vert}{|z|}. 
  \end{equation*}
\end{lemma} 

\begin{lemma}\cite[Lemma 1]{PEA08}\label{lem:convergence_ratios}
Denote $a_N$, $\overline{a}_N$, $b_N$ and $\overline{b}_N$ four infinite sequences of complex random variables indexed by $N$ and assume $a_N\asymp \overline{a}_N$ and $b_N\asymp \overline{b}_N$. If $|a_N|$, $|\overline{b}_N|$ and/or $|\overline{a}_N|$,$|b_N|$ are uniformly bounded above over $N$ (almost surely), then $a_Nb_N\asymp \overline{a}_N \overline{b}_N$. Similarly, if $|a_N|$, $|\overline{b}_N|^{-1}$ and/or $|\overline{a}_N|$,$|b_N|^{-1}$ are uniformly bounded above over $N$ (almost surely), then $a_N/b_N\asymp \overline{a}_N/ \overline{b}_N$.
\end{lemma}

\begin{lemma}[Tonelli theorem {\cite[Theorem 18.3]{BIL08}}]
  \label{le:tonelli}
  If $(\Omega,\mathcal F,P)$ and $(\Omega',\mathcal F',P')$ are two probability spaces, then for $f$ an integrable function with respect to the product measure $Q$ on $\mathcal F\times \mathcal F'$,
  \begin{equation*}
    \int_{\Omega\times \Omega'}f(x,y)Q(d(x,y)) = \int_\Omega \left[\int_{\Omega'}f(x,y) P'(dy) \right]P(dx)
  \end{equation*}
  and
  \begin{equation*}
    \int_{\Omega\times \Omega'}f(x,y)Q(d(x,y)) = \int_{\Omega'} \left[\int_{\Omega}f(x,y) P(dy) \right]P'(dx). 
  \end{equation*} 
\end{lemma} 

\section{Related results}
\label{sec:appendixB}
\begin{theorem}[{\cite[Theorem 1]{WAG10}}]\label{th:detequcorr}
 Let $\Bm_N=\Xm\Xm\htp$, where $\Xm\in\CC^{N\times n}$ is random. The $j$th column $\xv_{j}$ of $\Xm$ is given as $\xv_{j} = \Rm_j^{\frac12} \zv_{j}$, where the entries of $\zv_{j}\in\CC^{N}$ are i.i.d.\@ with zero mean, variance $1/N$ and finite moment of order $4+\epsilon$, for some common $\epsilon>0$, and $\Rm_{j}\in\CC^{N\times N}$ are Hermitian nonnegative definite matrices. Let $\Dm_N\in\CC^{N\times N}$ be a deterministic Hermitian matrix. Assume that both $\Rm_{j}$ and $\Dm_N$ have uniformly bounded spectral norms (with respect to $N$). Then, as $n,N\to\infty$ such that $0< \lim\inf N/n \leq \lim\sup N/n < \infty$, the following holds for any $z\in\CC\setminus\RR_+$:
\begin{align*}
 \frac1N\trace\Dm_N\LB\Bm_N-z\Id_N\RB^{-1} - \frac1N\trace\Dm_N\Tm_N(z) \asto 0
\end{align*}
where $\Tm_N(z)\in\CC^{N\times N}$ is defined as
\begin{align*}
 \Tm_N(z) = \LB\frac1N\sum_{j=1}^n\frac{\Rm_{j}}{1+\delta_{j}(z)}-z\Id_N\RB^{-1}
\end{align*}
and where $\delta_{1}(z),\dots,\delta_{n}(z)$ are given as the unique solution to the following set of implicit equations:
\begin{align}\label{eq:fixedpoint}
 \delta_{j}(z)=\frac1N\trace\Rm_{j}\LB\frac1N\sum_{j=1}^n\frac{\Rm_{j}}{1+\delta_{j}(z)}-z\Id_N\RB^{-1},\qquad j=1,\dots,n
\end{align}
such that $(\delta_{1}(z),\dots,\delta_{n}(z))\in\Sc^n$. For $z<0$, $\delta_{1}(z),\dots,\delta_{N,n}(z)$ are the unique nonnegative solutions to \eqref{eq:fixedpoint} and can be obtained by a standard fixed-point algorithm with initial values $\delta_{j}^{(0)}(z)=-1/z$ for $j=1,\dots,n$.
 Moreover, let $F_N$ be the empirical spectral distribution (e.s.d.) of $\Bm_N$ and denote by $\bar{F}_N$ the distribution function with Stieltjes transform $\frac1N\trace\Tm_N(z)$. Then, almost surely,
\begin{align*}
 F_N - \bar{F}_N \Rightarrow 0.
\end{align*}
\end{theorem}

\vspace{10pt}\begin{theorem}[{\cite{wagnerphd}}]\label{th:logdetcor}
	Under the assumptions of Theorem~\ref{th:detequcorr}, let ${\sigma^2}>0$ and define $V_N({\sigma^2})=\frac1N\log\det\LB\Id_N+\frac{1}{{\sigma^2}}\Bm_N\RB$. Then, as $N,n\to \infty$,
\begin{align*}
	\mathbb{E}V_N({\sigma^2}) - \bar{V}_N({\sigma^2}) \asto 0
\end{align*}
where
\begin{align*}
	\bar{V}_N({\sigma^2}) &= \frac1N\log\det\LB\Id_N  +  \frac{1}{{\sigma^2}}\frac1N\sum_{j=1}^n\frac{\Rm_{j}}{1+\delta_{j}}\RB + \frac1N\sum_{j=1}^n\log\LB1+\delta_{j}\RB - \frac1N\sum_{j=1}^n \frac{\delta_{j}}{1+\delta_{j}}
\end{align*}
and where $\delta_{j}=\delta_{j}(-{\sigma^2})$ for $j=1,\dots,n$ are given by Theorem~\ref{th:detequcorr}.
\end{theorem}

\vspace{10pt}\begin{corollary}\label{cor:logdet}
 Under the assumptions of Theorem~\ref{th:logdetcor}, assume additionally that the matrices $\Rm_{j}$, $j=1,\dots,n$, are drawn from a finite set of Hermitian nonnegative-definite matrices. Then, as $N,n\to\infty$,
\begin{align}
	V_N({\sigma^2}) - \bar{V}_N({\sigma^2}) \asto 0
\end{align}
where $V_N({\sigma^2})$ and $\bar{V}_N({\sigma^2})$ are defined as in Theorem~\ref{th:logdetcor}.
\end{corollary}
\begin{IEEEproof}
 It was shown in \cite[Proof of Theorem 3]{hoydis2011} that $\Bm_N$ has almost surely uniformly bounded spectral norm as $N,n\to\infty$ if the matrices $\Rm_{j}$ are drawn from a finite set of matrices. Thus, $F_N$ and $\bar{F}_N$ as defined in Theorem~\ref{th:detequcorr} have (almost surely) bounded support. Consider now a set $A\subset\Omega$, $\Omega$ generating the matrices $\Bm_N$, for which $\Bm_N$ has bounded spectral norm, and a set $B\subset\Omega$ for which $F_N - \bar{F}_N \Rightarrow 0$. Since $P(A)=P(B)=P(A\cap B)=1$, it follows from \cite[Theorem 25.8 (ii)]{billingsley}, that, as $N,n\to\infty$
\begin{align}
	\int\log(1+x^{-1}\lambda)dF_N(\lambda) - \int\log(1+x^{-1}\lambda)d\overline{F}_N(\lambda)\asto 0
\end{align}
which is equivalent to stating that $V_N(x) - \bar{V}_N(x) \asto 0$.
\end{IEEEproof}

\bibliography{IEEEconf,IEEEabrv,tutorial_RMT,bibliography}

\end{document}